%% file: main.tex
\documentclass[10pt]{article} 
\usepackage[preprint]{tmlr}

\input{math_commands.tex}

\usepackage{caption}
\usepackage{algorithm}
\usepackage{algorithmic}
\usepackage{latexsym}
\usepackage{amsmath}
\usepackage{amsthm,amssymb}
\usepackage{booktabs}
\usepackage{enumitem}
\usepackage{graphicx}
\usepackage{color}
\usepackage{subfigure}
\usepackage{hyperref}
\usepackage{url}

\newtheorem{theorem}{Theorem}
\newtheorem{lemma}[theorem]{Lemma}
\newtheorem{corollary}[theorem]{Corollary}

\newtheorem{definition}{Definition}

\theoremstyle{remark}
\newtheorem{remark}[theorem]{Remark}
\newtheorem{claim}[theorem]{Claim}

\usepackage{soul}

\usepackage{adjustbox}
\usepackage{footnote}
\usepackage{nicefrac}
\usepackage[usestackEOL]{stackengine}
\usepackage{float}
\makesavenoteenv{tabular}
\usepackage{pgfplots}
\usepgfplotslibrary{colorbrewer}

\usepackage[colorinlistoftodos]{todonotes}
\usepackage{tikz} 
\usepackage{multirow}
\usetikzlibrary{arrows}
\tikzset{
    vertex/.style={circle,draw,minimum size=1.5em},
    edge/.style={->,> = latex'}
}
\usetikzlibrary{arrows.meta}
\usepackage{enumitem}
\usepackage{lipsum}

\newcommand{\OPT}{{\rm OPT}}

\newtheorem{manualtheoreminner}{\normalfont{\textbf{Theorem}}}
\newenvironment{manualtheorem}[1]{%
  \renewcommand\themanualtheoreminner{#1}%
  \manualtheoreminner
}{\endmanualtheoreminner}

\title{Regularized Unconstrained Weakly Submodular Maximization}

\author{\name Yanhui Zhu \email yanhui@iastate.edu \\
      \addr Department of Computer Science\\
      Iowa State University
      \AND
      \name Samik Basu \email sbasu@iastate.edu \\
      \addr Department of Computer Science\\
      Iowa State University
      \AND
      \name A. Pavan \email pavan@cs.iastate.edu\\
      \addr Department of Computer Science\\
      Iowa State University}

\begin{document}

\maketitle

\begin{abstract}
Submodular optimization finds applications in machine learning and data mining. In this paper, we study the problem of maximizing functions of the form $h = f-c$,  where $f$ is a monotone, non-negative, weakly submodular set function and $c$ is a modular function. We design a deterministic approximation algorithm that runs with ${\mathcal{O}}(\frac{n}{\epsilon}\log \frac{n}{\gamma \epsilon})$ oracle calls to function $h$, and outputs a set $\Tilde{S}$ such that $h(\Tilde{S}) \geq \gamma(1-\epsilon)f(\OPT)-c(\OPT)-\frac{c(\OPT)}{\gamma(1-\epsilon)}\log\frac{f(\OPT)}{c(\OPT)}$, where $\gamma$ is the submodularity ratio of $f$. Existing algorithms for this problem either admit a worse approximation ratio or have quadratic runtime. We also present an approximation ratio of our algorithm for this problem with an approximate oracle of $f$. We validate our theoretical results through extensive empirical evaluations on real-world applications, including vertex cover and influence diffusion problems for submodular utility function $f$, and Bayesian A-Optimal design for weakly submodular $f$. Our experimental results demonstrate that our algorithms efficiently achieve high-quality solutions.
\end{abstract}

\input{sections/1-introduction}
\input{sections/2-preliminaries}

\input{sections/3-unconstrained}
\input{sections/5-results}

\input{sections/6-conclusion}

\appendix

\input{sections/7.1-related-work}
\input{sections/7.2-full-proof}

\input{sections/7.3-approximate}

\bibliography{ref}
\bibliographystyle{tmlr}

\end{document}

%% file: math_commands.tex
\usepackage{amsmath,amsfonts,bm}

\def\eqref#1{equation~\ref{#1}}

\def\1{\bm{1}}

\DeclareMathAlphabet{\mathsfit}{\encodingdefault}{\sfdefault}{m}{sl}
\SetMathAlphabet{\mathsfit}{bold}{\encodingdefault}{\sfdefault}{bx}{n}

\DeclareMathOperator*{\argmax}{arg\,max}

%% file: sections/1-introduction.tex
\section{Introduction}
\label{sec:intro}
Submodular function optimization is a fundamental combinatorial optimization problem with broad applications such as feature compression, document summarization, movie recommendation,  information diffusion, ad allocation, among others \cite{bateni2019categorical,sipos2012temporal,el2022data,li2023submodularity,kempe2003maximizing,tang2016optimizing}. Due to their broad applicability, the last decade has seen a surge of work on variants of submodular optimization problems. 
Even though many submodular optimization problems are NP-hard, several variants have been shown to admit efficient approximation algorithms and are amenable to rigorous theoretical analysis. 

\textbf{Driving Problem.\ } In the study of submodular maximization under knapsack/cost constraints (SMK), there is an additional cost function $c$ over the ground set. The cost function is modular, i.e., $c(S) = \sum_{x \in S} c(x)$. Given a monotone submodular function $f$, a cost function $c$, and a budget $B$,  the goal is to find a set $S$ that maximizes $f(S)$ under the constraint that $c(S) \leq B$. When the cost function is uniform, SMK reduces to the classical submodular maximization with a cardinality constraint problem (SMC).
SMK has numerous real-world applications; for instance, in viral marketing~\cite{kempe2003maximizing}, the revenue function (also called influence function) is often modeled as a monotone submodular function $f$, and every user is associated with a cost. The goal is to find a seed set that maximizes the revenue under the constraint that the total cost of the seeds does not exceed $B$.
However, maximizing the revenue may not guarantee maximal {\em profit}, which is the difference between achieved revenue and the cost. Thus, an alternate way to formulate the profit maximization problem is to maximize the function $h = f-c$. This maximization problem is known as {\em regularized unconstrained submodular maximization}~\cite{bodek2022maximizing}.

In certain application scenarios, such as the Bayesian A-Optimal design problem, the underlying function $f$ is \textit{non-submodular}, but only {\em weakly submodular}~\cite{harshaw2019submodular}.
This motivated the researchers to study the maximization of $h=f-c$ problem where $f$ is non-submodular. Specifically, given $f$ be a monotone weakly submodular function and $c$ be a modular cost function. The \textit{Regularized Unconstrained Weakly-Submodular Maximization} ($\textbf{\sc RUWSM}$) problem is defined:
\[ \textbf{\sc RUWSM} :  \argmax_{S \subseteq V} f(S)-c(S) \]

In certain scenarios, evaluating the submodular function $f$ might be intractable, and one may only have the ability to query a surrogate function $\Tilde{f}$ that is \textit{approximate} to $f$ \cite{cohen2014sketch}.  In this case, we refer to the maximization of $f-c$ problem as 
\textit{Regularized Unconstrained Approximate-Submodular Maximization} (\textbf{\sc RUASM}).

\textbf{Hardness. } We make a few remarks about the general difficulty of the above problems.  When $f$ is monotone (weakly) submodular and $c$ is modular, the function $h=f-c$ is (weakly) submodular. However, the function $h$ can be non-monotone and is potentially negative. Much of the prior work on submodular maximization focuses on the case when the underlying function is positive. The fact that $h$ could be negative places significant obstacles in designing approximation algorithms. 
Indeed, it is known that the problem of maximizing non-positive submodular functions does not admit constant multiplicative factor approximation algorithms~\cite{papadimitriou1988optimization, feige1998threshold,bodek2022maximizing}.

\textbf{Randomized vs. Deterministic Algorithms.} 
While randomization has proven valuable in developing approximation algorithms for submodular maximization \cite{mirzasoleiman2015lazier, harshaw2019submodular}, its practical use presents challenges. Unlike deterministic algorithms, which offer consistent results, randomized algorithms provide approximation ratios only in expectation. This means that there is a chance of obtaining a poor solution with a constant probability. To achieve the desired approximation ratio with high probability, typically, $\mathcal{O}(\log n)$ independent repetitions of the algorithm are necessary, which might not be feasible or practical \cite{chen2023approximation}. Furthermore, although some algorithms for submodular optimization can be derandomized, this often comes at the cost of a polynomial increase in time complexity \cite{buchbinder2018deterministic}.

\begin{table*}[t]
\caption{Comparisons of related works for unconstrained $f-c$ problem. $\gamma$ is the submodularity ratio of a weakly submodular $f$. ``WS" checks if the algorithm applies to weakly submodular $f$, and {\sc Det} means ``deterministic".}
\label{table:unconstrained}
\begin{center}
\begin{small}
\begin{sc}
\begin{tabular}{c c c c c l}
\toprule
Algorithm & Approximation Bound & Query Complex.  & WS $f$?   & Det? \\
\midrule
\textsc{UDG}~\cite{harshaw2019submodular} & $(1-e^{-\gamma})f(\OPT)-c(\OPT)$ & $\mathcal{O}(n)$ & Yes & No  \\ 
\textsc{ROI}~\cite{jin2021unconstrained} & $f(\OPT)-c(\OPT)-c(\OPT)\log \frac{f(\OPT)}{c(\OPT)}$ & $\mathcal{O}(n^2)$ & No & Yes  \\
\textsc{UP}~(this paper) & $\gamma(1-\epsilon)f(\OPT)-c(\OPT)- \frac{c(\OPT)}{\gamma(1-\epsilon)}\log \frac{f(\OPT)}{c(\OPT)}$ & $\mathcal{O}(\frac{n}{\epsilon}\log \frac{n}{\gamma \epsilon})$ & Yes & Yes \\ 
\bottomrule
\end{tabular}
\end{sc}
\end{small}
\end{center}
\end{table*}

\textbf{State of the Art Algorithms. }  \citet{harshaw2019submodular} presented a {\em randomized} algorithm referred to as \emph{unconstrained distorted greedy} (UDG) to address the RUWSM problem.  Their algorithm produces a set $S$ such that {\em expected value of $f(S)$} $\geq (1-e^{-\gamma})f(\OPT) - c(\OPT)$, and the algorithm makes $\mathcal{O}(n)$ oracle calls to function $f$.\footnote{$\gamma$ is the submodularity ratio for a weakly submodular function $f$; when $\gamma = 1$, the function is submodular, see  Section \ref{sec:preliminaries} for a formal definition of submodularity ratio. } Note that the approximation guarantee holds only in expectation; to make this a high-probability event, one should run the algorithm multiple ($O(\log n)$)  times and output the median.  Subsequently, \citet{jin2021unconstrained} developed a {\em deterministic} algorithm, referred to as ROI greedy, for addressing
the RUSM problem (when $f$ is submodular). The algorithm produces a set $S$ such that 
$h(S) \geq f(\OPT)-c(\OPT) - {c(\OPT)}\log \frac{f(\OPT)}{c(\OPT)}$ within $\mathcal{O}(n^2)$ oracle calls to $f$.  The authors also proved the tightness of this approximation bound. 

In the context of these works, a natural question arises for the RUWSM problem: {\em Can one design a deterministic algorithm that matches the approximation guarantee of the ROI algorithm while achieving a runtime closer to that of the UDG algorithm?} 

In this work, we affirmatively answer this question.

\paragraph{\bf Our Contributions.}

We design a fast approximation algorithm for the regularized unconstraint submodular maximization problems. To the best of our knowledge, this is the first {\em near linear time deterministic} algorithm that applies to submodular, weakly submodular and even approximately weakly submodular functions $f$. 

We establish almost tight approximation bounds of the proposed algorithm in the context of submodular, weakly submodular and approximately submodular functions $f$. Specifically, when $f$ is submodular, our approximation ratio almost matches the tight approximation ratio achieved by ROI~\cite{jin2021unconstrained} with an improved runtime. Compared to the linear-time randomized algorithm \cite{harshaw2019submodular}, our approximation is better when $\epsilon \in (0, 0.15]$. Additionally, we show that our algorithm and can be applied to more generalized scenarios when the algorithm only has access to an approximate oracle for a (weakly) submodular function $f$.  Table~\ref{table:unconstrained} summarizes the results.

We have conducted extensive experiments with various types of (weakly) submodular functions in the applications of vertex cover, influence maximization and Bayesian A-optimal design. For fair comparisons to the baselines, we have designed a modified ROI algorithm and derived its approximation guarantees when $f$ is weakly submodular. Experimental results present the high quality of our results, and the execution time indicates the viability of our algorithm in addressing problems with real data sets.

\subsection{Related Works} 
The classical work of \citet{nemhauser1978analysis} presented a greedy algorithm with a $(1-\nicefrac{1}{e})$ for maximizing a monotone submodular function $f$ with cardinality constraint. 
This algorithm makes $O(kn)$ calls to the submodular function. 
Currently, the best-known deterministic algorithms for this problem achieve an approximation ratio of $(1-1/e-\epsilon)$ while making $O(\frac{n}{\epsilon})$ calls~\cite{li2022submodular,Kuhnle21}.
For knapsack constraint,  the greedy algorithm due to \citet{sviridenko2004note} achieves $(1-\nicefrac{1}{e})$ approximation ratio at the cost of a high time complexity $O(n^5)$. Subsequent research has focused on improving the algorithm's runtime while accepting a slight reduction in approximation quality~\cite{badanidiyuru2014fast, ene2017nearly, feldman2020practical,yaroslavtsev2020bring, li2022submodular,zhu2024improved}. For the matroid constraint, the natural greedy algorithm archives an approximation ratio of $1/2$~\cite{nemhauser1978analysis}. The seminal work of \citet{calinescu2011maximizing} developed a randomized algorithm that achieves the optimal approximation ratio of $1-\nicefrac{1}{e}$. Furthermore, \citet{buchbinder2019deterministic} proposed the first deterministic algorithm with an approximation ratio of $0.5008$.  Additionally, more submodular maximization algorithms are developed and analyzed in the context of submodular constraints, dynamic constraints \cite{roostapour2022pareto,bian2021fast},  dynamic stream settings \cite{lattanzi2020fully,dutting2022deletion} and  distributed fashion \cite{bateni2018optimal,fahrbach2019submodular,zhang2023communication}.

\textbf{Regularized Submodular Maximization.} \citet{buchbinder2015tight} proposed a greedy algorithm admitting an approximation ratio of $\nicefrac{1}{2}$, provided  $f(S)>c(S)$ holds for every $S \subseteq V$. \citet{harshaw2019submodular} studied the RUWSM problem and its variant where a cardinality constraint is imposed. 
\citet{nikolakaki2021efficient} developed a framework that produces a set $S$ such that $h(S) \geq \nicefrac{1}{2}f(\OPT) - c(\OPT)$ for cardinality and online unconstrained problems when $f$ is submodular. 
\citet{jin2021unconstrained} proposed ROI-Greedy for the RUSM problem that guarantees a \emph{tight} positive bound as long as $f(\OPT)>c(\OPT)$.  The works of~\cite{bodek2022maximizing, Qi22} studied this problem (and generalizations) when the submodular function $f$ is not necessarily monotone. Other constrained $f-c$ maximization problems are also studied in \cite{sviridenko2017optimal, cui2023constrained}. There has been a large body of work that focused on various submodular optimization problems when the algorithm has only access to an approximate/noisy oracle~\cite{hassidim2017submodular, qian2017subset,crawford2019submodular,nie2023framework,el2020optimal}. The works of~\cite{gong2023algorithms,geng2022bicriteria,wang2021maximizing} studied constrained regularized submodular maximization in the presence of approximate oracles.  There also exist fast approximation algorithms for the variant when the function $c$ is also submodular \cite{PZBP23,perrault2022approximation}.

%% file: sections/2-preliminaries.tex
\section{Preliminaries}
\label{sec:preliminaries}
\textbf{Notations.} 
Given a ground set $V$ of size $n$ and a utility function $f:2^V\rightarrow \mathbb{R}^{\geq 0}$, $f$ is monotone if for any subsets $S \subseteq T \subseteq V$, $f(T) \geq f(S)$. The {\em marginal gain} of adding an element $x \in V \setminus S$ into $S\subseteq V$ is $f(\{x\}|S) \triangleq f(S\cup \{x\})-f(S)$. We write $f(\{x\}|S) = f(x|S)$ for brevity. In this paper, we consider normalized $f$, i.e., $f(\emptyset)=0$. A function $f$ is (weakly) submodular if, for any set $S \subseteq T \subseteq V$ and any element $x \in V \setminus T$, the marginal gain of the function value of adding $x$ to $S$ is at least the marginal gain in the function value of adding $x$ to a larger set $T$.  Formally, 

\begin{definition}[Submodularity]
    \label{def:sub}
    For any set $S \subseteq T \subseteq V$ and $x \in V \setminus T$, a set function $f:2^V \rightarrow \mathbb{R}^{\geq 0}$ is submodular if
    \begin{equation}
\label{eq:marginal-gain}
    f(S \cup \{x\})-f(S)\geq f(T \cup \{x\})-f(T).
\end{equation}
\end{definition}

This property is referred to as {\em diminishing return}. A weakly submodular set function $f$ \cite{das2011submodular} is an extension of the submodular function and is defined as follows.
\begin{definition}[$\gamma$-submodularity]
    \label{def:weak-sub}
        A set function $f:2^V \rightarrow \mathbb{R}^{\geq 0}$ is $\gamma$-weakly submodular if for every set $S \subseteq T \subseteq V$:
        \begin{equation}
        \label{eq:weak-submodularity}
            \sum\nolimits_{u\in T \setminus S} f(u \mid S)  \geq  \gamma \big(f(T)-f(S)\big).
        \end{equation}
        $\gamma \in (0,1]$ is referred to as the submodularity ratio and $f$ is submodular iff $\gamma=1$.  
\end{definition}
We define the $\delta$-approximate function $\Tilde{f}$ that is within $\delta$-error of the actual oracle $f$.
\begin{definition}[$\delta$-approximate]
\label{def:approximate}
A set function $\Tilde{f}$ is $\delta$-approximate to the utility function $f$ if for every $S \subseteq V$,
    \[  |f(S) - \Tilde{f}(S) | \leq \delta \]
\end{definition}

A cost function $c: 2^V \rightarrow \mathbb{R}^+$ is {\em modular} or {\em linear} if the equality of Eq. (\ref{eq:marginal-gain}) holds. The {\em density} of an element $e$ adding to set $S$ is defined to be $\frac{f(e|S)}{c(e)}$. Define $c_{max} = \max_{e \in V} c(e)$ and $c_{min} = \min_{e \in V} c(e)$.

\textbf{Oracles.} In this paper, our algorithms are designed under the standard value oracle model, which assumes that an oracle is capable of returning the value $f(S)$ when provided with a set $S \subseteq V$. The specific value of $f(S)$ may depend on the problem at hand, such as the expected influence spread in influence diffusion~\cite{kempe2003maximizing}, the covered vertices in directed vertex coverage, or the reduction in variance in Bayesian A-Optimal design~\cite{harshaw2019submodular}. 
As is customary, the computational complexity analyses of our algorithms are based on the number of oracle calls. For RUASM, we assume the algorithm can only access a $\delta$-approximate oracle for $f$. 
For experiments, we also report and compare the function/oracle evaluations regarding running time.

%% file: sections/3-unconstrained.tex
\let\classAND\AND
\let\AND\relax

\begin{algorithm}[t]
   \caption{\textsc{UP}($f, c, \gamma, \epsilon$)}
   \label{alg:up}
\begin{algorithmic}[1]
   \STATE {\bfseries Input:} utility function $f$, cost function $c$, submodularity ratio $\gamma$, error threshold $\epsilon$.
   \STATE {\bfseries Output:} $\Tilde{S} = \argmax_{S_i, i \in [n]} f(S_i)-c(S_i)$.
   \STATE $S_0 \gets \emptyset$; $i\gets 1$
    \STATE $u(e) \gets 0$ for all $e \in V$ 
    \STATE $PQ \gets$ priority\_queue$(\{(\frac{f(e)}{c(e)}, e) \mid e \in V\})$
    \WHILE{True}
    \STATE Remove entries from $PQ$ where the density(key) $\leq \gamma$
        \IF{PQ is empty}
            \STATE \textsc{Terminate the Algorithm}  
        \ENDIF
        \STATE $(\tau_i, v_i) \gets PQ.getTopAndRemove()$ 
        \STATE $u(v_i)++$ 
        \IF{$\frac{f(v_i|S_{i-1})}{c(v_i)} \geq \max\{\gamma, (1-\epsilon)\tau_i \}$}
            \STATE $S_i \gets S_{i-1} \cup \{v_i\}$
            \STATE $i++$
            \STATE \textsc{Continue to the Next Iteration} 
        \ENDIF
        \IF{$u(v_i) \leq \frac{\log \frac{n}{\gamma \epsilon}}{\epsilon}$}
            \STATE $PQ.push((\frac{f(v_i|S_{i-1})}{c(v_i)}, v_i)$) 
        \ENDIF    
    \ENDWHILE
\end{algorithmic}
\end{algorithm}

\section{The Algorithm and Approximation Results}
\label{sec:unconstrained}
In this section, we present a deterministic fast Unconstrained Profit~(\textsc{UP}) algorithm (Algorithm~\ref{alg:up}) and the approximation results for the RUWSM and RUASM problems. 

\subsection{Overview and Techniques}
\label{sec:overview}
Our goal is to attain the approximation guarantee similar to ROI-Greedy~\cite{jin2021unconstrained} using an efficient (close to linear time) deterministic algorithm for \textit{weakly submodular} $f$. Lazy evaluations \cite{minoux2005accelerated} are commonly used for accelerations in submodular maximization problems. However, applying lazy evaluations to weakly submodular $f$ is challenging due to unpredictable marginal gains~\cite{harshaw2019submodular}. Moreover, lazy evaluations though do not improve asymptotic time bounds. Another possible approach is to  investigate the applicability of \emph{threshold technique} \cite{badanidiyuru2014fast} to accelerate 
the Greedy algorithm.

\paragraph{\bf Traditional Threshold Algorithms.} 
For constrained monotone submodular maximization with a cardinality constraint (solution set size $\leq k$), the natural greedy algorithm iterates $k$ times, adding \emph{one} element with the maximum marginal gain in each iteration. In contrast, the thresholding technique-based algorithm maintains a global threshold and adds potentially \emph{multiple elements} whose marginal gain exceeds the threshold in each iteration. The threshold decreases geometrically, limiting the total iterations to $O(\frac{1}{\epsilon}\log\frac{n}{\epsilon})$. As there are at most $n$ oracle calls for marginal gain computation in each iteration, the overall runtime is $O(\frac{n}{\epsilon}\log\frac{n}{\epsilon})$, which is faster than the natural greedy algorithm with $O(nk)$ runtime. \citet{badanidiyuru2014fast} proved that the thresholding technique maintains the quality of results, achieving an approximation ratio of $(1-1/e-\epsilon)$ (as opposed to $(1-1/e)$ by the natural greedy algorithm). Variants of threshold algorithms \cite{ene2017nearly,amanatidis2020fast} have found effective applications in other submodular maximization problems.

A natural idea is to adapt this threshold idea to address the RUWSM problem. There are, however, obstacles to a direct adaption: (1)
The threshold algorithm usually requires that the maximized function is monotone and submodular. In contrast, the RUWSM problem is about maximizing $h=f-c$, where $h$ can be non-monotone or non-submodular;
(2) The running time is unbounded since it will use the maximal density (unbounded for a general $f$) to initialize the global threshold. The justifications are deferred to the Appendix (Section \ref{sec:modified-thresh}).

\paragraph{\bf Our Algorithm.} 
We now discuss our strategy for addressing the RUWSM problem. Instead of using a single \textit{global} threshold, we maintain a threshold for each element -- whenever the value exceeds the corresponding threshold within a specific error range, we add the element to the partial solution set. 

Algorithm \ref{alg:up} initializes an empty set $S_0$, sets up counters $u(\cdot)$ for each element $e \in V$ (keep track of how many times an element has been evaluated), and creates a priority queue $PQ$ sorted in non-decreasing order based on the densities ($f(e)/c(e)$). 

The algorithm proceeds as follows. First, line 7 removes entries from the priority queue $PQ$ where the density falls below the submodularity ratio $\gamma$. This procedure will reduce unnecessary evaluations because the removed elements negatively impact the objective value $f-c$. Within the main loop, the algorithm retrieves the element $v_i$ with the highest density $\tau_i$ from the priority queue. Then it evaluates the new density of $v_i$ w.r.t. current set $S_{i-1}$ and increments the counter for $v_i$ (lines 12-13). If the new density is within $\epsilon$ error threshold of $\tau_i$, the algorithm adds $v_i$ to the current subset $S_{i-1}$. To avoid excessive evaluation of densities corresponding to an element, the algorithm (line 19) checks if the evaluations for $v_i$ have been performed for $(\log \frac{n}{\gamma \epsilon})/\epsilon$ times. If it is, then $v_i$ will be discarded since it will not make significant contributions to the objective value; otherwise, the element is pushed back into the priority queue for potential evaluations in future iterations. The algorithm continues this process until no elements are left in the priority queue. The returned solution set $\Tilde{S}$ is the one that realizes the best objective value among all intermediate solution sets. 

The proceeding theorem states the approximation guarantees of Algorithm~\ref{alg:up} for addressing the RUWSM problem (weakly submodular utility function $f$). The approximation extends to the scenario of submodular $f$ when setting $\gamma=1$.

\begin{theorem}
\label{th:unconstrained}
    Given a monotone $\gamma$-submodular function $f:2^{V} \rightarrow \mathbb{R}^{\geq 0}$ and a modular cost function $c$, an error threshold $\epsilon \in (0,1)$, after $\mathcal{O}(\frac{n}{\epsilon}\log \frac{n}{\gamma \epsilon})$ oracle calls to $f$,  Algorithm \ref{alg:up} outputs $\Tilde{S}$ such that
    \[f(\Tilde{S})-c(\Tilde{S}) \geq \gamma(1-\epsilon)f(\OPT)-c(\OPT)- \frac{1}{\gamma(1-\epsilon)}A\]
    where $A = c(\OPT) \log \frac{f(\OPT)}{c(\OPT)}$.
\end{theorem}

\subsection{Theoretical Analysis of Theorem \ref{th:unconstrained}}
We now proceed by presenting a roadmap of the proof for Theorem~\ref{th:unconstrained} by illustrating the definitions, lemmas, and claims used to discharge the proof of the theorem. As depicted in Figure \ref{fig:roadmap}, each node corresponds to a
theorem, lemma, claim, etc., and the incoming directed edge captures the necessity of the source node in
the proof of the destination node. Some edges are annotated with conditions under which the
source node leads to the destination. For instance, if the condition in Case 2.2
holds, Lemma \ref{lem:unconstrained-h_St} is valid and implies the validity of Theorem \ref{th:unconstrained}.  

In our proof, we will proceed with the leaf-level lemmas and claims and work our way toward the root of the tree.

\begin{figure}
\begin{center}
\begin{tikzpicture}[scale=0.80]
\node (1) at (0,0) {\scriptsize Thm. \ref{th:unconstrained}};
\node (2) at (-2.6,-2) {\scriptsize Def. \ref{def:weak-sub}};
\node (3) at (2.8, -1) {$\Box$};
\node (4) at (-0.8, -3.5) {\scriptsize Alg. \ref{alg:up}, line~7, 18};
\node (5) at (6, -3) {\scriptsize Lem. \ref{lem:unconstrained-h_St}};
\node (6) at (3, -3) {\scriptsize Claim \ref{claim:cost-St}};
\node (7) at (4, -4) {\scriptsize Lem. \ref{lem:h-upper-bound}};
\node (8) at (6, -4) {\scriptsize Cor. \ref{col:lower-opt}};
\node (9) at (2, -4.5) {\scriptsize Lem. \ref{lem:f-Si}};
\node (10) at (6, -4.5) {\scriptsize Claim \ref{claim:single-cost}};
\node (11) at (4, -5) {\scriptsize Lem. \ref{lem:marginal-gain}};
\draw[edge] (2) -- (1) node[midway,above,sloped] {{\bf \scriptsize Case 1}};
\draw[edge] (2) -- (1) node[midway,below,sloped] {{\bf \scriptsize $\scriptstyle \gamma(1-\epsilon)c(S_\ell) < A$}};

\draw[edge] (3) -- (1) node[midway,above,sloped] {{\bf \scriptsize Case 2}};
\draw[edge] (3) -- (1) node[midway,below,sloped] {{\bf \scriptsize $\scriptstyle\gamma(1-\epsilon)c(S_\ell) \geq A$}};

\draw[edge] (4) -- (3) node[midway,above,sloped] {{\bf \scriptsize Case 2.1}};
\draw[edge] (4) -- (3) node[midway,below,sloped] {{\bf \scriptsize $\scriptstyle f(S_j) \geq \gamma(1-\epsilon)f(\OPT)$}};

\draw[edge] (5) -- (3) node (12) [midway,below,sloped]  {{\bf \scriptsize Case 2.2}};
\draw[edge] (5) -- (3) node (12) [midway,above,sloped]  {{\bf \scriptsize $\scriptstyle f(S_j) < \gamma(1-\epsilon)f(\OPT)$}};

\draw[edge] (6) -- (5);
\draw[edge] (7) -- (5);
\draw[edge] (8) -- (5);
\draw[edge] (7) -- (8);
\draw[edge] (9) -- (7);
\draw[edge] (10) -- (7);
\draw[edge] (11) -- (9);
\draw[edge] (11) -- (10);
\end{tikzpicture}
\end{center}
\caption{Roadmap of the proof of Theorem \ref{th:unconstrained}.}
\label{fig:roadmap}
\end{figure}

\subsubsection{Useful Lemmas}
We start with Lemmas~\ref{lem:marginal-gain} and \ref{lem:f-Si}. Let $\ell$ be the value of $i$ at the termination of Algorithm \ref{alg:up}. Lemma \ref{lem:marginal-gain} lower bounds the marginal gain when adding an element to the solution set. With Lemma \ref{lem:marginal-gain}, Lemma \ref{lem:f-Si} gives the lower bound of the $f$ value of a solution set at some iteration.

\begin{lemma}
\label{lem:marginal-gain}
For all $i \in [1,\ell]$ and if $v_i$ is the element added to $S_{i-1}$, then
\[ 
f(v_i|S_{i-1}) \geq \gamma(1-\epsilon) \frac{c(v_i)}{c(\OPT)} \big ( f(\OPT) - f(S_{i-1}) \big ). 
\]
\end{lemma}

\begin{proof} 
Since $v_i$ is the element added to $S_{i-1}$ at line~14 of the algorithm,
$\frac{f(v_i|S_{i-1})}{c(v_i)}\geq {(1-\epsilon)\tau_i}$ is satisfied. Furthermore, for every 
$\forall u\in \OPT \setminus S_{i-1}$, $u$ is either equal to $v_i$ or does not satisfy the condition in line~13. Hence, 
\begin{equation}
\frac{f(v_i|S_{i-1})}{c(v_i)}\geq {(1-\epsilon)\tau_i} \geq (1-\epsilon) \frac{f(u|S_{i-1})}{c(u)}.
\label{eq:i-2}
\end{equation}

We then derive the following inequality.
\begin{align*}
    & f(\OPT)-f(S_{i-1}) \leq \frac{1}{\gamma} \sum_{u\in \OPT \setminus S_{i-1}} f(u|S_{i-1}) \quad\mbox{\tt  by Def.~(\ref{def:weak-sub})} \\
    &\leq \frac{1}{\gamma} \sum_{u\in \OPT \setminus S_{i-1}} \frac{c(u)}{1-\epsilon} \cdot \frac{f(v_i|S_{i-1})}{c(v_i)}  \quad\quad\quad\mbox{\tt  due to Eq.~(\ref{eq:i-2})} \\
    &\leq \frac{1}{\gamma(1-\epsilon)} \frac{f(v_i|S_{i-1})}{c(v_i)} c(\OPT).
\end{align*}

Rearranging the last inequality concludes the proof.     
\end{proof}

\begin{lemma}
\label{lem:f-Si}
For all $i\in [1,\ell]$, if $v_i$ is the element added to $S_{i-1}$ and $f(S_{i})=f(S_{i-1} \cup \{v_i\})$, then
\begin{equation}
    \label{eq:f-Si}
    f(S_{i}) \geq \left(   1- \prod_{s=1}^{i} \left(1- \gamma(1-\epsilon) \frac{c(v_i)}{c(\OPT)} \right)   \right) f(\OPT).
\end{equation}
\end{lemma}

{\sc Proof Sketch.}
We prove this lemma by induction. We verify the base case ($i=1$) using Lemma \ref{lem:marginal-gain}; note that
$f(\emptyset) = 0$. Next, for every $i>1$, as an induction hypothesis, assume $\forall j \in [1, i-1]$, Eq. (\ref{eq:f-Si}) holds. We can express $f(S_{j+1})$ in the form of $f(S_{j+1}) = f(S_j)+f(v_{j+1}|S_j)$, and apply Lemma ~\ref{lem:marginal-gain} to $f(v_{j+1}|S_j)$ and the induction hypothesis to $f(S_j)$ to wrap up the proof. \hfill $\Box$

\subsubsection{Proof of Theorem \ref{th:unconstrained} }

Let $S_\ell$ with size $\ell$ be the set at the termination of Algorithm \ref{alg:up}. We divide the proof of Theorem~\ref{th:unconstrained} into two cases based on the relationship between $A$ and $c(S_\ell)$.

\textbf{Case 1}: $\gamma (1-\epsilon)c(S_\ell) < A$.
    Consider each element $u \in \OPT \setminus S_\ell$ that was not included in $S_\ell$; such an element can be categorized into two sets $O_1$ and $O_2$ with corresponding reasons.
    
     \begin{itemize}
         \item \textit{$O_1$ is the set of elements removed at line 7.\ }\\
     The density of these elements is $\leq \gamma$ at some point before the algorithm ends. Therefore, for $u \in O_1$, if $\frac{f(u|S_u)}{c(u)}$ is the density (key) of $u$ the last time it appears in the priority queue, then we have $S_u \subseteq S_\ell$. By the property of diminishing return, we have
        \begin{equation}
        \label{eq:o1}
            \frac{f(u|S_\ell)}{c(u)} \leq \frac{f(u|S_u)}{c(u)} \leq \gamma.
        \end{equation}
        \item \textit{$O_2$ is the set of elements for which the line 18 condition was not satisfied.\ }\\
    In other words, these elements have been evaluated at line 13 for $(\log \frac{n}{\gamma\epsilon})/{\epsilon}$ times. Note that for any $u \in O_2$, the initial density of $u$ in the priority queue is $\frac{f(u)}{c(u)}$. Thus, we have: 
        \begin{equation}
            \label{eq:o2}
            \frac{f(u|S_\ell)}{c(u)} \leq \frac{f(u|S_u)}{c(u)} \leq (1-\epsilon)^{\frac{\log \frac{n}{\gamma \epsilon}}{\epsilon}}  \frac{f(u)}{c(u)}.
        \end{equation}
     
    From Eq. (\ref{eq:o2}), we can further derive $\forall u \in O_2$,
    \begin{align}
    \label{eq:derive-o2}
        f(u|S_u) &\leq (1-\epsilon)^{\frac{\log \frac{n}{\gamma \epsilon}}{\epsilon}} f(u) \nonumber\\
        &\leq (1+\epsilon)^{-\frac{\log \frac{n}{\gamma \epsilon}}{\epsilon}}f(u) \quad \quad \mbox{\it due to $1-\epsilon \leq \frac{1}{1+\epsilon}$} \nonumber\\
        &\leq (1+\epsilon)^{-\frac{\log \frac{n}{\gamma \epsilon}}{\log 1+\epsilon}}f(u) \quad \quad \mbox{\it due to $\epsilon \leq \log{1+\epsilon}$}\nonumber\\
        &= (1+\epsilon)^{-\log_{1+\epsilon} \frac{n}{\gamma \epsilon}}f(u) ~~=~~\frac{\gamma \epsilon}{n} f(u).
    \end{align}
    \end{itemize}
    
    Therefore, for Case 1, we have
    \begin{align*}
        f(\OPT)-f(S_\ell)-c(\OPT) &\leq \frac{1}{\gamma}\!\!\sum_{u \in \OPT \setminus S_{\ell}}\!\!\!\!\!\! f(u|S_\ell)-c(\OPT)\\
        &\leq \frac{1}{\gamma}\left(\sum_{u \in O_1} f(u|S_\ell)+\sum_{u \in O_2} f(u|S_\ell)\right)-c(\OPT)\\
        &\leq \frac{1}{\gamma}\left(\sum_{u \in O_1} f(u|S_\ell)+\sum_{u \in O_2} f(u|S_\ell)\right)-\sum_{u \in O_1} c(u) \\
        &=\left(\sum_{u \in O_1} \frac{1}{\gamma} f(u|S_\ell) - c(u)\right)+\frac{1}{\gamma}\sum_{u \in O_2} f(u|S_\ell)\\
        &\leq 0 + \frac{1}{\gamma} \sum_{u \in O_2} \frac{\gamma \epsilon}{n} f(u) \quad \quad \mbox{\tt from  Eq. (\ref{eq:o1}) and (\ref{eq:derive-o2})}\\
        &\leq \epsilon f(\OPT) .
    \end{align*}
    The first inequality is due to monotonicity and weak-submodularity (Def.~\ref{def:weak-sub}) of $f$.  The last inequality holds because $O_2 \subseteq \OPT$, $|\OPT| \leq n$ and $f(u) \leq f(\OPT)$ for $u \in \OPT \setminus S_\ell$. Thus, we have  $f(\OPT)-f(S_\ell)-c(\OPT) \leq \frac{1}{\gamma} \sum_{u \in O_2} \frac{\gamma \epsilon}{n} f(u) \leq \epsilon f(\OPT)$. 
     
     Combining the above inequality with the Case 1 condition ($\gamma (1-\epsilon)c(S_\ell) < A$), we have 
     \begin{align*}
         f(\Tilde{S})-c(\Tilde{S})& \geq f(S_\ell)-c(S_\ell)\geq (1-\epsilon)f(\OPT)-c(\OPT)-\frac{1}{\gamma (1-\epsilon)}A.
     \end{align*}
    {This concludes the proof of the theorem for Case 1}. \hfill $\Box$

    \textbf{Case 2}: $\gamma (1-\epsilon)c(S_\ell) \geq A$. \\
    In this case, there exist $t\in[1, \ell]$ such that:
    \begin{equation}
        c(S_{t-1})< \frac{1}{\gamma (1-\epsilon)} A  \leq c(S_t)
        \label{eq:th1-case2}
    \end{equation}
    We further divide Case 2 into two sub-cases: (2.1) and (2.2).

\textbf{Case 2.1}: \textit{For some $j\in [1,t-1]$, $f(S_{j}) \geq \gamma (1-\epsilon)f(\OPT)$}.\\
Note that Algorithm~\ref{alg:up} outputs $\Tilde{S} = \argmax_{S_i, i \in [n]} f(S_i)-c(S_i)$.
Therefore, for some $j\in [1,t-1]$
\begin{align*}
h(\Tilde{S}) &\geq f(S_j) - c(S_j)  \geq \gamma (1-\epsilon)f(\OPT)-\frac{A}{\gamma(1-\epsilon)} \\
& \geq \gamma (1-\epsilon)f(\OPT) - c(\OPT) - \frac{A}{\gamma(1-\epsilon)} .
\end{align*}

This concludes the proof of the theorem for sub-case 2.1.  \hfill $\Box$

\textbf{Case 2.2}: \textit{For all $j \in [1,t-1]$, $f(S_{j}) < \gamma (1-\epsilon)f(\OPT)$}.\\
As illustrated in the proof path, we have already
proved the Lemmas~\ref{lem:marginal-gain} and \ref{lem:f-Si}. We will now prove the necessary properties conditioned on the sub-case 2.2. The following two claims give the properties of the cost of every single element in set $S_t$.

\begin{claim}
\label{claim:single-cost}
Under the Case 2.2, the following holds: $\forall~ j \in [1, t-1]:~ c(v_j) \leq c(\OPT)$.
\end{claim}
{\sc Proof Sketch:} We prove this claim by contradiction. Assume there exist some $j \in [1, t-1]$ that $c(v_j)>c(\OPT)$. Under Case 2.2 condition, $f(S_j) < \gamma(1-\epsilon)f(\OPT)$ is true. Then if we apply Lemma~\ref{lem:marginal-gain} to $f(v_j|S_{j-1})$ of $f(S_j)= f(S_{j-1}) + f(v_j|S_{j-1})$, we can derive $f(S_j) \geq \gamma(1-\epsilon)f(\OPT)$, which is a contradiction. \hfill $\Box$

\begin{claim}
\label{claim:cost-St}
Under the Case 2.2, if we denote $B_j = A-\gamma(1-\epsilon) c(S_{j})$, then we have  $c(v_t) > B_{t-1}$.
\end{claim}

\begin{proof}
\begin{align*}
    c(v_t) &= c(S_t) - c(S_{t-1}) \ \geq \frac{1}{\gamma(1-\epsilon)}A- c(S_{t-1}) \\
    &= \frac{1}{\gamma(1-\epsilon)}\big(A-\gamma (1-\epsilon)c(S_{t-1})\big) \\
       & \geq   A- \gamma (1-\epsilon)c(S_{t-1})\ = \ B_{t-1}.
\end{align*}
The first inequality follows from the case 2 condition Eq.~(\ref{eq:th1-case2}).
\end{proof}

Next, we present Lemma \ref{lem:h-upper-bound} that upper bounds the optimal objective value $f(\OPT) - c(\OPT)$. Furthermore, under the condition of Case 2.2, a corollary that upper bounds $f(S_j)$ can be derived from the lemma. As illustrated in the proof roadmap in Figure \ref{fig:roadmap}, the lemma and corollary are necessary for the proof of Lemma \ref{lem:unconstrained-h_St}, which wraps up the proof of Case 2.2 of the theorem.

\begin{lemma} 
\label{lem:h-upper-bound}
Under the Case 2.2,
the following holds: for all $j\in [1,t-1]$, if $B_j=A-\gamma(1-\epsilon) c(S_{j})$, then
\[ 
   f(\OPT) - c(\OPT) \leq \frac{B_j(f(\OPT) - f(S_{j}))}{c(\OPT)} +f(S_{j}). 
\]
\end{lemma}

\begin{proof}
We divide this proof into two cases based on the valuations of $B_j$ and $c(\OPT)$.

\begin{itemize}
    \item If \textbf{$B_j > c(\OPT)$}: The lemma holds immediately.
    \item If \textbf{$B_j \leq c(\OPT)$}: We first solve the right hands side of the inequality by splitting it into a part with $f(S_{j})$ and without $f(S_{j})$. Then we can further apply Lemma~\ref{lem:f-Si} to the former.
\end{itemize}
\begin{align*}
    & \displaystyle\frac{B_j(f(\OPT) - f(S_{j}))}{c(\OPT)} +f(S_{j}) = \displaystyle\left(1-\frac{B_j}{c(\OPT)}\right)f(S_{j}) + \frac{B_j}{c(\OPT)}f(\OPT)\\
    &\geq  \displaystyle\left(1-\frac{B_j}{c(\OPT)}\right)\underbrace{\left(   1- \prod_{k=1}^{j}  \left(1- \gamma(1-\epsilon) \frac{c(v_k)}{c(\OPT)} \right)   \right) f(\OPT)}_{\mbox{\tt follows from Lemma~\ref{lem:f-Si}}} +\displaystyle\frac{B_j}{c(\OPT)}f(\OPT)\\
        &=f(\OPT) - \displaystyle\left(\left(1-\frac{B_j}{c(\OPT)}\right)\prod_{k=1}^{j} \left(1- \gamma(1-\epsilon) \frac{c(v_k)}{c(\OPT)} \right)\right)f(\OPT)\\
        &=\left( 1- \displaystyle\left(1-\frac{B_j}{c(\OPT)}\right)\prod_{k=1}^{j} \left(1- \gamma(1-\epsilon) \frac{c(v_k)}{c(\OPT)} \right) \right)f(\OPT)\\
        &\geq \left( 1- \displaystyle\left(1- \frac{B_j+\gamma(1-\epsilon)\cdot \sum_{k=1}^{j}c(v_k)}{(j+1)\cdot c(\OPT)} \right)^{j+1} \right)f(\OPT)
\end{align*}

The last inequality holds due to the following known inequality: 
\begin{equation}
1-\prod_{i=1}^n\left(1-\frac{x_i}{y}\right) \geq 1-\left(1-\frac{\sum_{i=1}^n x_i}{n y}\right)^n
\label{eq:prodsum}
\end{equation}
where $x_i, y \in \mathbb{R}^{+} \mbox{ and } x_i \leq y \mbox{ for } i \in [n]$.

Consider $\forall k\in [1,j], x_k = \gamma(1-\epsilon) c(v_k)$, $x_{k+1} = B_j$ and $y = c(\OPT)$. Note that, 
$\forall k\in [1,j], \gamma(1-\epsilon) c(v_k) \leq c(\OPT)$, which follows from Claim~\ref{claim:single-cost} and 
$\gamma(1-\epsilon) \leq 1$. Also, $B_j \leq c(\OPT)$ follows from the condition for this case. Therefore, the last inequality can be obtained by applying Eq.~(\ref{eq:prodsum}).

Proceeding further by using the definitions of $B_j$ and modular function $c$, the following derivations wrap up the proof.

\begin{align*}
    &\displaystyle\frac{B_j(f(\OPT) - f(S_{j}))}{c(\OPT)} +f(S_{j}) \\
    &\geq \left( 1- \displaystyle\left(1- \frac{A-\gamma(1-\epsilon)c(S_j)+\gamma(1-\epsilon) \sum_{k=1}^{j}c(v_k)}{(j+1)\cdot c(\OPT)} \right)^{j+1} \right)f(\OPT)\\
    &= \left( 1- \displaystyle\left( 1- \frac{A}{(j+1)\cdot c(\OPT)} \right)^{j+1} \right)f(\OPT) \\
        &= \left( 1- \displaystyle\left( 1- \frac{1}{j+1}\cdot \ln \frac{f(\OPT)}{c(\OPT)} \right)^{j+1} \right)f(\OPT) \\
        &\geq \left(1- e^{-\ln \displaystyle\frac{f(\OPT)}{c(\OPT)}}\right)f(\OPT) \quad \quad \mbox{\tt as $(1- \frac{c}{t} )^t \leq e^{-c}$ } \\
        &=\left( 1- \displaystyle\frac{c(\OPT)}{f(\OPT)} \right)f(\OPT)\\
        &= f(\OPT)-c(\OPT).
\end{align*}
\end{proof}

\begin{corollary} 
\label{col:lower-opt}
Under the Case 2.2,
the following holds: 
for all $j \in  [1, t-1],$ if $B_j=A-\gamma(1-\epsilon) c(S_{j})$ then 
\[f(\OPT) -f(S_{j}) \geq \displaystyle\frac{1}{\gamma(1-\epsilon)}c(\OPT).\]
\end{corollary}

\begin{proof}

Note that, for the case $f(S_{j}) -c(S_{j})
            \geq \gamma (1-\epsilon)f(\OPT)-c(\OPT)-\displaystyle\frac{1}{\gamma (1-\epsilon)}A$, Theorem~\ref{th:unconstrained} holds immediately. 
   Hence, we consider the opposite case:
    \begin{align*}
    f(S_{j}) -c(S_{j}) &\leq f(\OPT)-c(\OPT)-\frac{1}{\gamma (1-\epsilon)}A\\
    &=f(\OPT)-c(\OPT)- c(S_{j}) - \frac{1}{\gamma (1-\epsilon)}B_j\\
    &\leq \displaystyle\frac{B_j(f(\OPT) - f(S_{j}))}{c(\OPT)} +f(S_{j}) - c(S_{j}) - \frac{1}{\gamma (1-\epsilon)}B_j .
    \end{align*}
  
    The last inequality is due to Lemma~\ref{lem:h-upper-bound}. 
                Rearranging the above inequality concludes the proof of the corollary.
\end{proof}

The following lemma discharges the proof for Theorem~\ref{th:unconstrained} for the sub-case 2.2. 

\begin{lemma} 
    \label{lem:unconstrained-h_St}
Under the case 2.2,
the following holds:
    \[
    f(S_t)-c(S_t) \geq \gamma(1-\epsilon)f(\OPT)-c(\OPT) - \frac{1}{\gamma(1-\epsilon)}A
    \]
\end{lemma}

\begin{proof}
Consider $v_t$ the element added to $S_{t-1}$.
\begin{align*}
    & f(S_t)-c(S_t) =f(v_t|S_{t-1})+f(S_{t-1})-{c(S_t)}\\
    &\geq \gamma(1-\epsilon) \displaystyle\frac{c(v_t)}{c(\OPT)} \big ( f(\OPT) \!-\! f(S_{t-1}) \big ) +f(S_{t-1})\!-\!c(S_t)\\
    &= { \gamma(1-\epsilon) \displaystyle\frac{B_{t-1}+c(v_t)-B_{t-1}}{c(\OPT)} \big ( f(\OPT) - f(S_{t-1}) \big )}  +f(S_{t-1})- {\left(\frac{1}{\gamma(1-\epsilon)}(A-B_{t-1})+c(v_t)\right)}. 
\end{align*}
The first inequality holds from Lemma \ref{lem:marginal-gain}. The equality follows from the definition of  $B_{t-1} = A-\gamma (1-\epsilon) c(S_{t-1})$.

Proceeding further from the above inequality, 
\[
\begin{array}{rl}
& f(S_t)-c(S_t) \\
                &\geq  {\gamma(1-\epsilon) \displaystyle\frac{B_{t-1}}{c(\OPT)} \left( f(\OPT) - f(S_{t-1})\right)}+ f(S_{t-1})  \\
                &\hfill  + {\gamma(1-\epsilon) \displaystyle\frac{c(v_t)-B_{t-1}}{c(\OPT)} \left( f(\OPT) - f(S_{t-1}) \right) }\\
                 
                & \hfill- \displaystyle\frac{1}{\gamma(1-\epsilon)}A-\left(c(v_t)-\frac{1}{\gamma(1-\epsilon)}B_{t-1}\right)\\
                 
                &\geq  \gamma(1-\epsilon)  \underbrace{\left[\displaystyle\frac{B_{t-1}}{c(\OPT)} \left( f(\OPT) - f(S_{t-1}) \right) + f(S_{t-1})\right]}_{\mbox{\tt apply Lemma~\ref{lem:h-upper-bound}}}   \\
                 
                &\hfill + \gamma(1-\epsilon) \displaystyle\frac{c(v_t)-B_{t-1}}{c(\OPT)} \left( f(\OPT) - f(S_{t-1}) \right) \\
                    &\hfill- {\displaystyle\frac{1}{\gamma(1-\epsilon)}A-\big(c(v_t)-B_{t-1}\big)}\\   
                &\geq \gamma(1-\epsilon) (f(\OPT)-c(\OPT)) - \displaystyle\frac{1}{\gamma(1-\epsilon)}A \\
                &\hfill + \underbrace{\left[c(v_t)-B_{t-1}\right]\left( \gamma(1-\epsilon) \displaystyle\frac{f(\OPT) - f(S_{t-1})}{c(\OPT)} -1 \right)}_{\mbox{\tt residual value }}. 
\end{array}
\]

Note that, in Claim \ref{claim:cost-St}, we have established $c(v_t)\geq B_{t-1}$. Furthermore, from Corollary \ref{col:lower-opt}, we know
$f(\OPT) - f(S_j) \geq \frac{1}{\gamma(1-\epsilon)}{c(\OPT)}$. Therefore, 
the residual value is non-negative. Hence, 
\begin{align*}
   &f(\Tilde{S})-c(\Tilde{S}) \geq f(S_t)-c(S_t)\geq \gamma(1-\epsilon)f(\OPT)-c(\OPT)-\frac{1}{\gamma (1-\epsilon)}A.
\end{align*}
\end{proof}
This concludes the proof of Theorem~\ref{th:unconstrained} for case 2.2. We have proved the approximation guarantee of Theorem \ref{th:unconstrained} for every case. The following lemma finishes the proof of Theorem \ref{th:unconstrained} by bounding the worst-case running time of Algorithm \ref{alg:up}. 

\begin{lemma}
    \label{lem:time}
    The time complexity of Algorithm \ref{alg:up} is 
    $\mathcal{O}(\frac{n}{\epsilon} \log \frac{n}{\epsilon \gamma})$.
\end{lemma}
\begin{proof}
    As per line 18 of Algorithm \ref{alg:up}, each element $e \in V$ can be evaluated for at most $(\log \frac{n}{\gamma \epsilon})/{\epsilon}$ times. Each such consideration involves one oracle call to the function $f$. Therefore, the maximal total number of oracle calls for $n$ elements is $n \cdot (\log \frac{n}{\gamma \epsilon})/{\epsilon} \in \mathcal{O}(\frac{n}{\epsilon} \log \frac{n}{\epsilon \gamma})$. 
\end{proof}

\subsubsection{Approximation Guarantee with Approximate Oracles}
To complement our proposed efficient algorithm, we provide an approximation guarantee of Algorithm \ref{alg:up} for the RUASM problem ($f$ is $\delta$-approximate to a weakly submodular function). 

\begin{theorem}
\label{th:unconstrained-approx}
    Suppose the input utility function is $\delta$-approximate to a monotone $\gamma$-submodular set function $f$, $c$ is a modular cost function, $\epsilon \in (0,1)$ is an error threshold. Let $\OPT =  \argmax_{T \subseteq V}\{f(T)-c(T)\}$, after $\mathcal{O}(\frac{n}{\epsilon}\log \frac{n}{\gamma \epsilon})$ oracle calls to the $\delta$-approximate utility function,  Algorithm \ref{alg:up} outputs $\Tilde{S}$ such that
    \begin{align*}
        f(\Tilde{S})-c(\Tilde{S}) \geq &\gamma(1-\epsilon)f(\OPT)-c(\OPT)-\frac{1}{\gamma (1-\epsilon)}A \\
        &\quad\quad- 2\delta\left(\beta+ \frac{n}{\gamma} + 1 + n(1-\epsilon)\beta'\right)
    \end{align*}
    where $A = c(\OPT) \log \frac{f(\OPT)}{c(\OPT)}$, $\beta = \frac{c(\OPT)}{c_{\min}}$ and $\beta' = \frac{c_{max}}{c(\OPT)}$.
\end{theorem}
Theorem \ref{th:unconstrained-approx} generalizes both Theorem \ref{th:unconstrained} (with $\delta=0$) and ROI \cite{jin2021unconstrained} (with $\delta=0$ and $\epsilon=0$). 
The proof of Theorem \ref{th:unconstrained-approx} follows the roadmap as shown in Figure \ref{fig:roadmap} with moderate modifications to some lemmas. The detailed derivations are presented in the Appendix (Section \ref{sec:unconstrained-approx}).

%% file: sections/5-results.tex
\section{Experiments}
\label{sec:experiments} 

Our experiments are conducted on a MacBook Pro with an M1 Pro chip and 16GB RAM, and all algorithms are implemented in C++ \footnote{The code can be found: \url{https://github.com/yz24/UWSM}.}.
We have implemented lazy evaluations \cite{minoux2005accelerated} for algorithms with submodular function $f$, for example, ROI and our UP. 
We discuss the applications of profit maximization, vertex cover with costs, and Bayesian A-Optimal design. The latter is specifically for weakly submodular $f$ functions.

\subsection{Baseline Algorithms} 
We compare our algorithm \textsc{UP} with the modified \textsc{ROI}  and \textsc{UDG}~\cite{harshaw2019submodular}. 

\paragraph{\bf Modified ROI} Note that ROI greedy as presented in \cite{jin2021unconstrained} does not provide approximation guarantees when the function $f$ is weakly submodular.
We have developed a modified version of the ROI greedy algorithm to accommodate for weakly submodular $f$. The algorithm runs at most $n$ iterations and during each iteration, it adds an element with the maximal density w.r.t. the current partial solution as long as the density is greater than $\gamma$. The output is the maximal objective value of every partial solution (same as Algorithm \ref{alg:up}). Additionally, We have established an approximation bound of the modified algorithm in Theorem \ref{th:modified-roi}.
\begin{theorem}
\label{th:modified-roi}
    Given a monotone $\gamma$-submodular function $f:2^{V} \rightarrow \mathbb{R}^{\geq 0}$ and a modular cost function $c$, with $\mathcal{O}(n^2)$ oracle calls to $f$, the modified ROI algorithm outputs $\Tilde{S}$ such that
    \[f(\Tilde{S})-c(\Tilde{S}) \geq \gamma f(\OPT)-c(\OPT)- \frac{c(\OPT)}{\gamma} \log \frac{f(\OPT)}{c(\OPT)}.\]
\end{theorem}
We notice that when the function $f$ is submodular ($\gamma=1$), the approximation reduces to the approximation result of ROI \cite{jin2021unconstrained}. 
The pseudo-code of the modified algorithm and theoretical analysis of Theorem \ref{th:modified-roi} are presented in the Appendix (Section \ref{sec:modified-roi}). We note that the modified algorithm still works in quadratic time. The rationale behind developing the modified ROI is to appropriately compare our results with the ROI algorithm when the function $f$ is weakly submodular.

\paragraph{\bf UDG} 
The UDG algorithm is randomized and the claimed approximation bound holds in expectation. Thus, we run the UDG algorithm ten times and report the \textit{median} results returned. As for runtime, we also report the \textit{median} running time of the ten runs.

For a fair comparison among different methods, we evaluate the running time based on the number of oracle evaluations as discussed in Section \ref{sec:preliminaries}. 

\subsection{Applications and Experimental Settings}
In this subsection, we formulate three applications that apply to the proposed RUSM and RUWSM problems. We also provide the descriptions of the datasets we use for the experiments and critical parameter setups. The statistics of the datasets are presented in Tables \ref{tab:datasets-networks} and \ref{tab:datasets-others}.

\begin{table}[h]
    \centering
     \caption{Network Data Statistics}
    \begin{tabular}{ |l|c|c|c|c| } 
    \hline
    Application & Network & $n$ & $m$ & Type \\[0.2em]
    \hline
    \multirow{3}{*}{Profit Max.} & Gnutella-08 & 6,301 & 20,777 & Directed\\ [0.1em]
    & NetHEPT & 15,233 & 62,774 & Undirected \\ 
    \hline
    \multirow{3}{*}{Vertex Cover} & Protein & 1,706 & 6,207 & Directed\\ [0.1em]
    &  Eu-Email& 1,005 & 25,571 & Directed \\
    [0.1em]
    \hline
\end{tabular}
    \label{tab:datasets-networks}
\end{table}

\begin{table}[h]
    \centering
     \caption{Bayesian A-Optimal Design Data Statistics ($d$: number of attributes; $n$: number of measurements)}
    \begin{tabular}{ |l|c|c| } 
    \hline
    Dataset & $n$ & $d$  \\[0.2em]
    \hline
     Boston Housing & 506 & 14 \\ [0.1em]
     Segment & 2,310 & 19  \\ 
    [0.1em]
    \hline
\end{tabular}
    \label{tab:datasets-others}
\end{table}

\paragraph{\bf Profit Maximization.} The diffusion model \cite{kempe2003maximizing} is defined using a graph $G=(V,E)$ and a probability function $p(u,v), \forall (u,v)\in E$.   
A monotone submodular function $f(S) = \mathbb{E}[\mathrm{I}(S)]$ measures the expected number of nodes that are influenced by the set of nodes in $S$ based on a diffusion model, such as the independent cascade (IC) model \cite{kempe2003maximizing}. 
Furthermore, it is often the case that some cost $c(v)$ is associated with each node $v \in V$.
-- this is the cost incurred for including an entity in the seed
set. For instance, a highly influential entity is likely to incur higher costs than others.
Following \cite{jin2021unconstrained, tang2018towards}, we use a linear cost function $c:2^V \rightarrow \mathbb{R}^{>0}$ which is proportional to the out-degree $d(v)$ of a node $v \in V$. Formally, for every $v \in V$:
$c(v) = \lambda_1 \cdot d(v)^{\lambda_2}$, where $\lambda_1$ and $\lambda_2$ are positive cost penalty parameters. Thus, the modular cost of a set $S \subseteq V$ is $c(S)=\sum_{v\in S}c(v)$.  Note that we define $c(v)=1$ if $d(v)=0$ to avoid undefined {\em density} of node $v$.

Therefore, the {\em profit maximization} problem finds a subset $S$ of $V$ such that 
\[\argmax_{S\subseteq V} ~\mathbb{E}[\mathrm{I}(S)]-c(S).\]

For fair comparisons, we adapt the profit maximization  framework~\cite{jin2021unconstrained}, which is based on the Random Reachable (RR) approach to speed up the computation of the expected influence value~\cite{borgs2014maximizing,TangXS14,TangSX15,NguyenTD16,HuangWBXL17}. We use the standard benchmarks p2p-Gnutella-08 network \cite{leskovec2007graph} and NetHEPT network \cite{chen2009efficient}. 
We compare our algorithm with baselines by their objective values $(f-c)$ and running time with various cost penalty choices $\lambda_1$ and $\lambda_2$.


\paragraph{\bf Directed Vertex Cover with Costs. }
Let $G=(V, E)$ be a directed graph and $w: 2^V \rightarrow \mathbb{R}^{\geq 0}$ be a modular weight function on a subset of vertices. For a vertex set $S \subseteq$ $V$, let $N(S)$ denote the set of vertices which are pointed to by $S$, formally,
$N(S) \triangleq\{v \in V \mid(u, v) \in E \wedge u \in S\}$. The weighted directed vertex cover function is $f(S)=\sum_{u \in N(S) \cup S} w(u)$, which is monotone submodular. We also assume that each vertex $v \in$ $V$ has an associated non-negative cost defined by $c(v)=1+\max \{d(v)-q, 0\}$, where $d(v)$ is the out-degree of vertex $v$ and the non-negative integer $q$ is the cost penalty \cite{harshaw2019submodular}. The minimum cost of a single vertex is $1$. The larger $q$ is, the smaller the vertex costs are. 

The objective of the {\em vertex cover with costs maximization problem} is to find a subset $S$ such that 
\[\argmax_{S\subseteq V}\sum_{u \in N(S) \cup S} w(u) - \sum_{v \in S} c(v).\]
Here, function $f$ is monotone submodular, so we set $\gamma=1$. 

We use two networks in this application: Protein network \cite{stelzl2005human} and Eu-Email Core network~\cite{leskovec2007graph, yin2017local}. 
Since there is no node weight information in the datasets, for both networks, 
we assign each node a weight of 1 and a cost as defined above. For the experiments, we vary cost penalty choices $q \in [1, 12]$ and compare the objective values $(f-c)$ and running time (function calls) of our algorithm UP with ROI, UDG.

\begin{figure}[t]
\centering
\subfigure[p2p-Gnutella-08; $\lambda_2=1.2$]{
\includegraphics[width=0.22\textwidth]{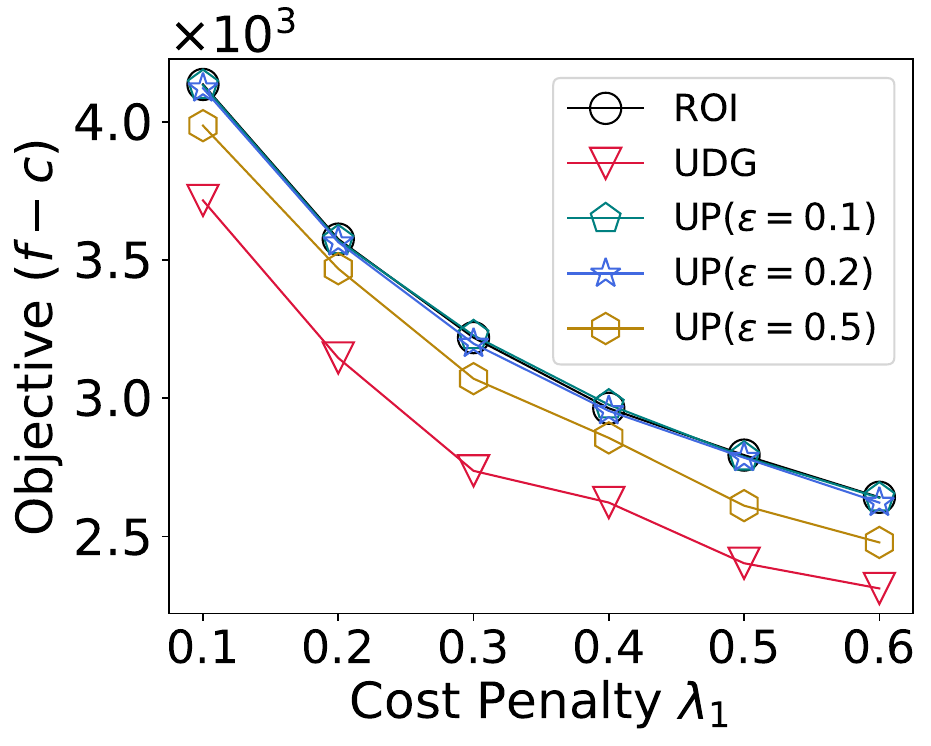}
\label{fig:p2p-1}
}
\subfigure[p2p-Gnutella-08; $\lambda_1=0.1$]{
\includegraphics[width=0.22\textwidth]{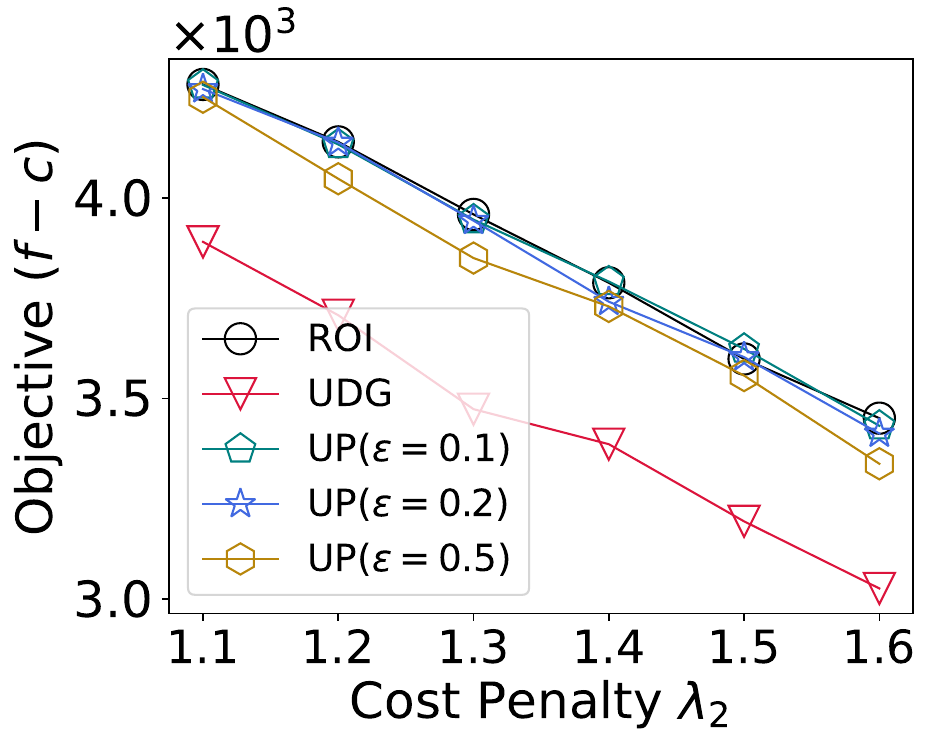}
\label{fig:p2p-2}
}
\subfigure[NetHept; $\lambda_2=1.0$]{
\centering
\includegraphics[width=0.22\textwidth]{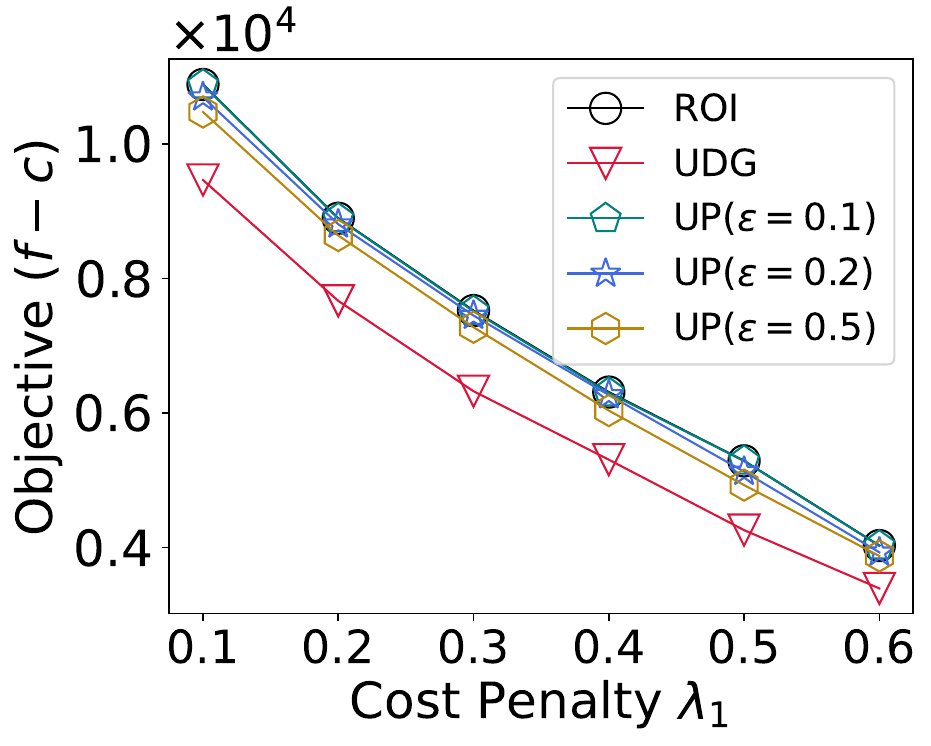}
\label{fig:nethept-1}
}
\subfigure[NetHept; $\lambda_1=0.1$]{
\includegraphics[width=0.22\textwidth]{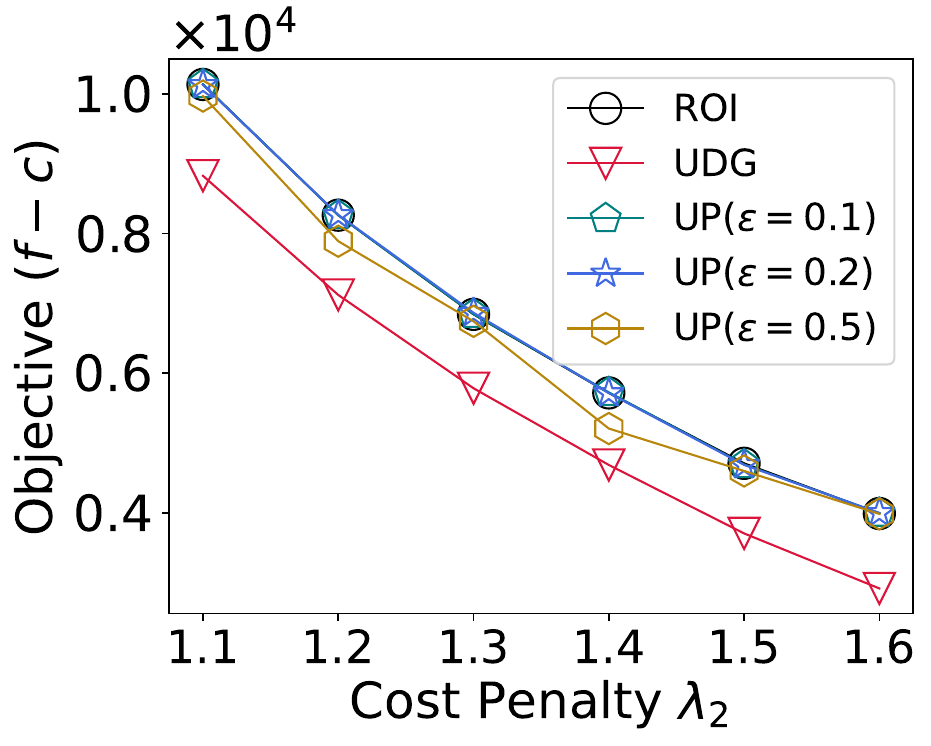}
\label{fig:nethept-2}
}
\caption{Comparisons on p2p-Gnutella31 and NetHept networks with the application of Profit Maximization.}
\end{figure}

\begin{figure}[t]
\centering
\subfigure[p2p-Gnutella31]{
\includegraphics[width=0.22\textwidth]{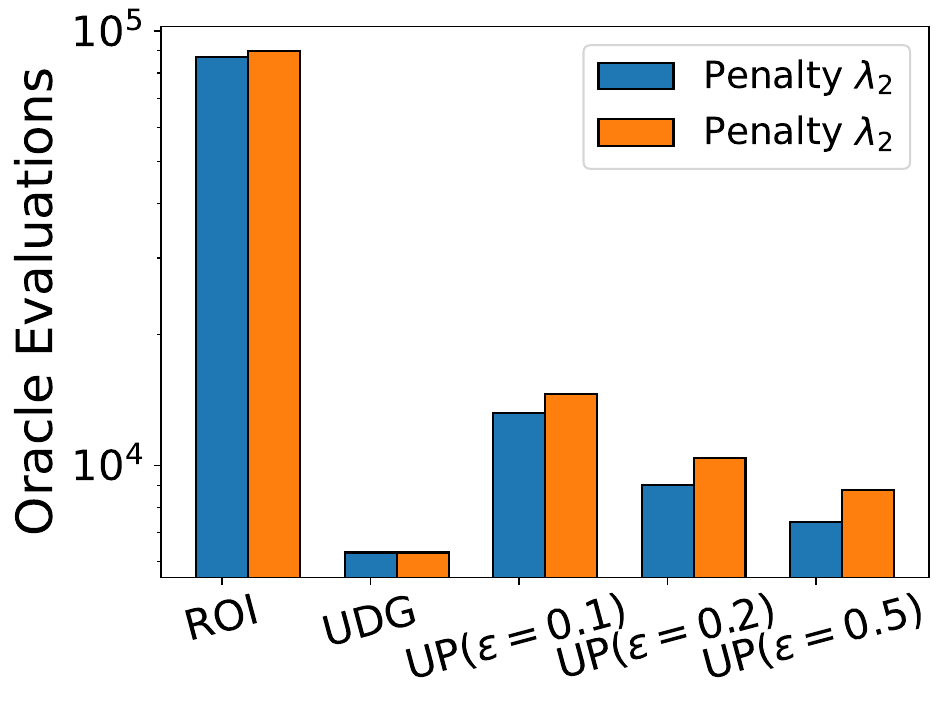}
\label{fig:p2p-time}
}
\subfigure[NetHept]{
\includegraphics[width=0.22\textwidth]{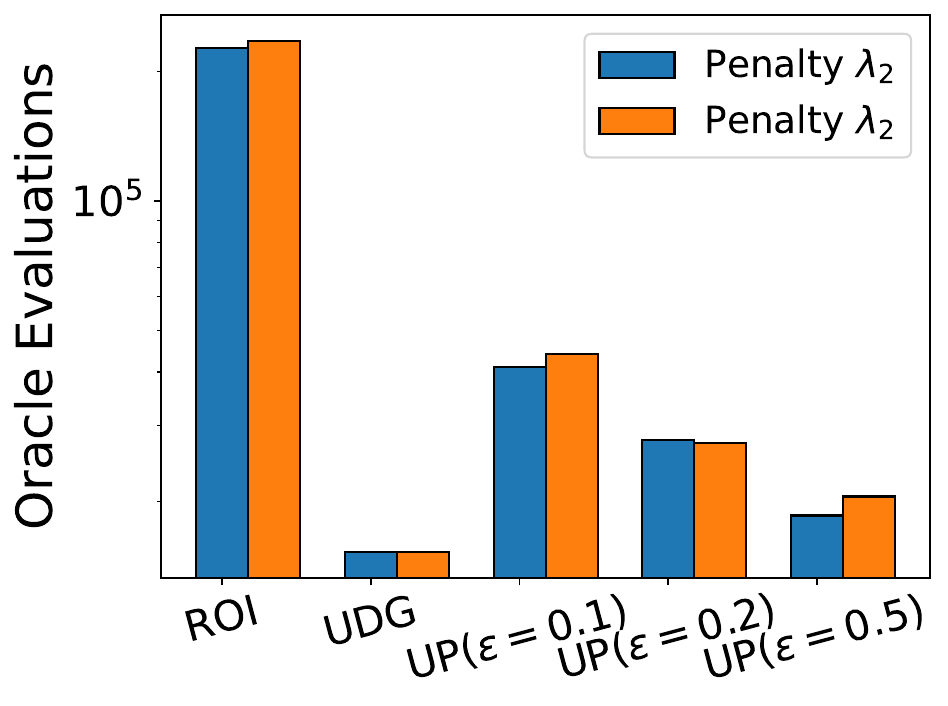}
\label{fig:nethept-time}
}
\subfigure[Vertex Cover]{
\includegraphics[width=0.22\textwidth]{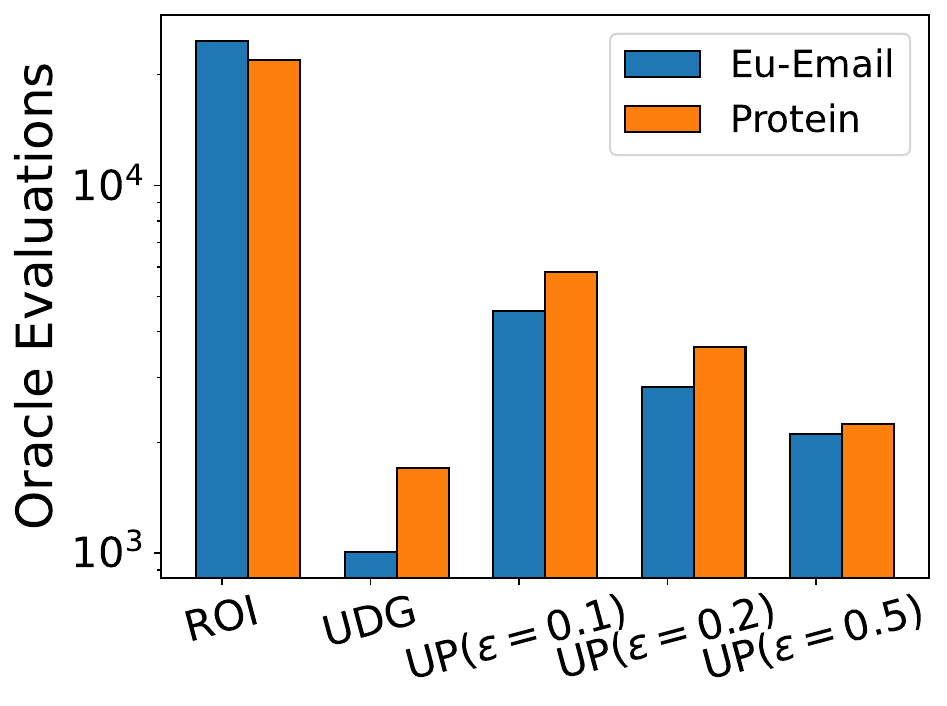}
\label{fig:VC-time}
}
\subfigure[Bayesian A-Optimal Design ]{
\includegraphics[width=0.22\textwidth]{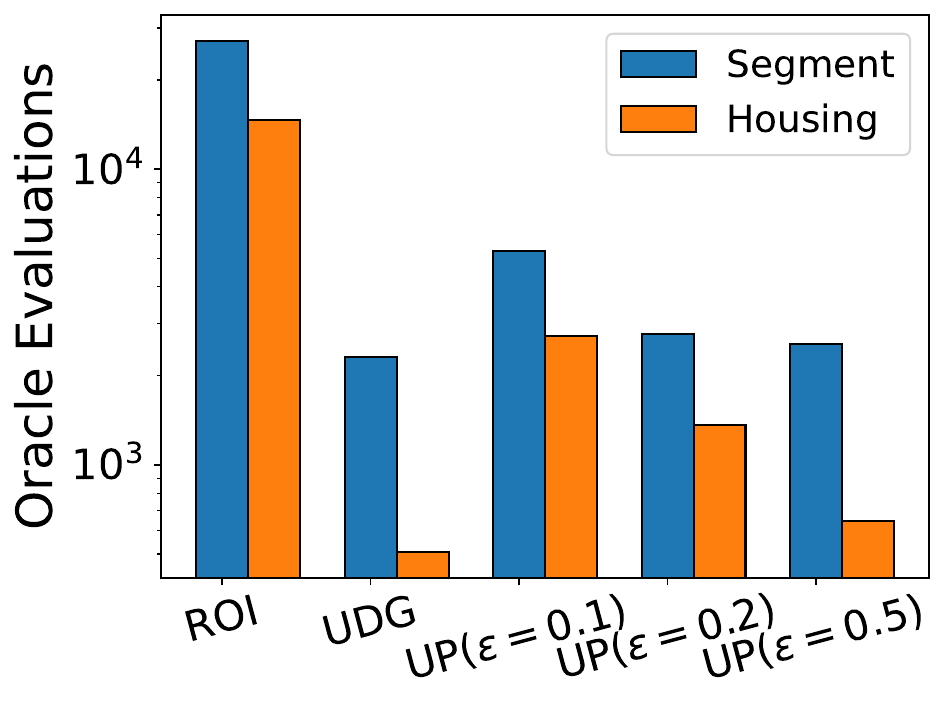}
\label{fig:bayesian-time}
}
\caption{Comparisons of average oracle evaluations over various cost penalties: p2p-Gnutella31 and NetHept networks (fixed $\lambda_1$ or $\lambda_2$) with Profit Maximization; various datasets with Vertex Cover and Bayesian A-Optimal Design applications. The results of UDG are one-run snap of ten runs.}
\end{figure}

\begin{figure}[t]
\centering
\subfigure[$q$ vs $f-c$; Eu-Email]{
\centering
\includegraphics[width=0.22\textwidth]{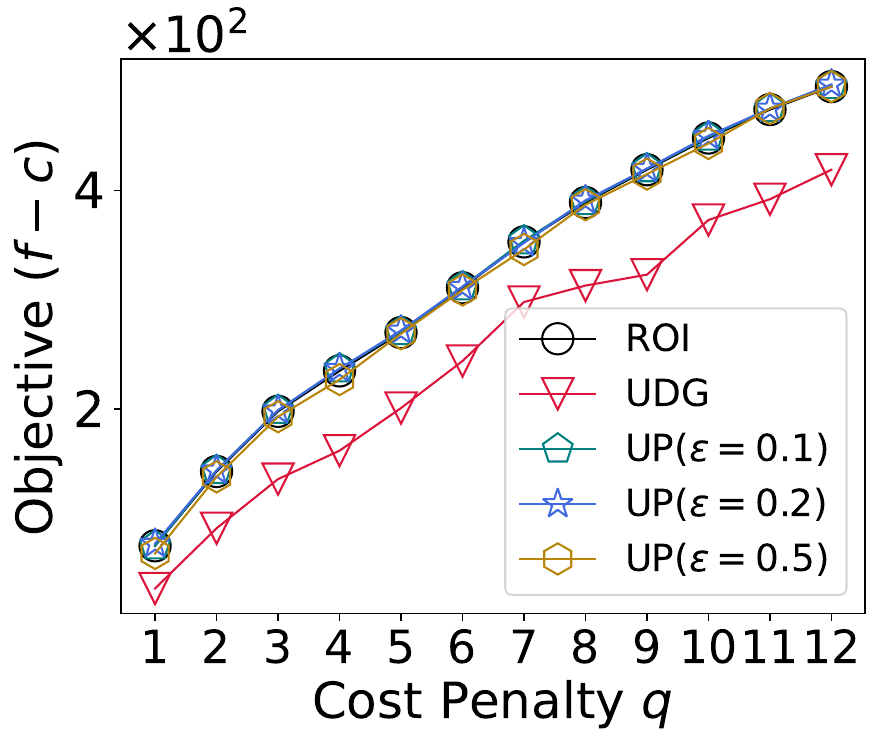}
\label{fig:VC-email}
}
\subfigure[$q$ vs $f-c$; Protein]{
\includegraphics[width=0.22\textwidth]{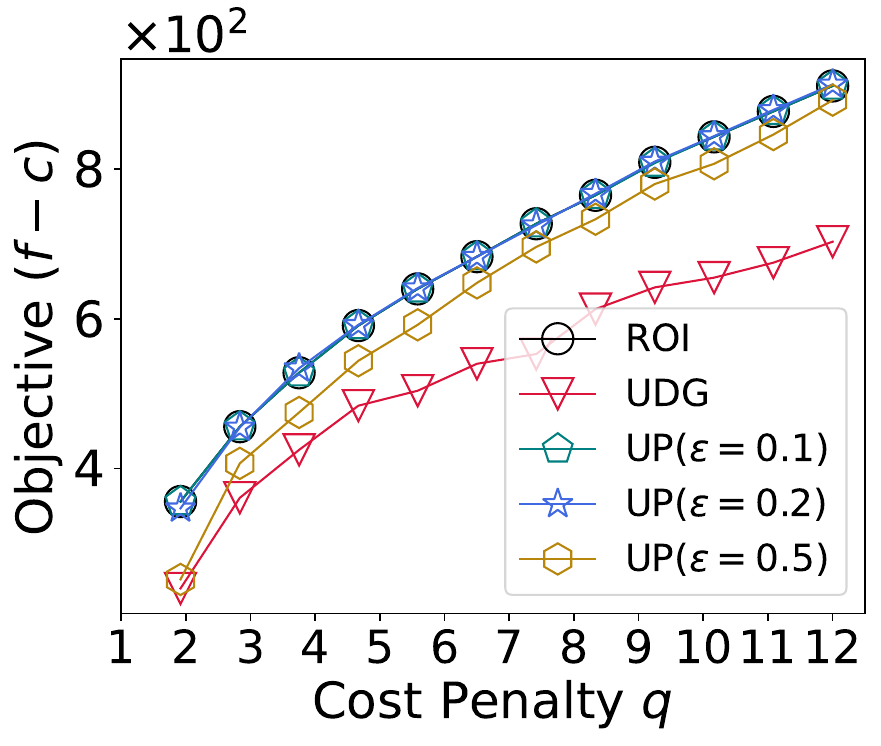}
\label{fig:VC-protein}
}
\subfigure[$p$ vs $f-c$; Segment]{
\centering
\includegraphics[width=0.22\textwidth]{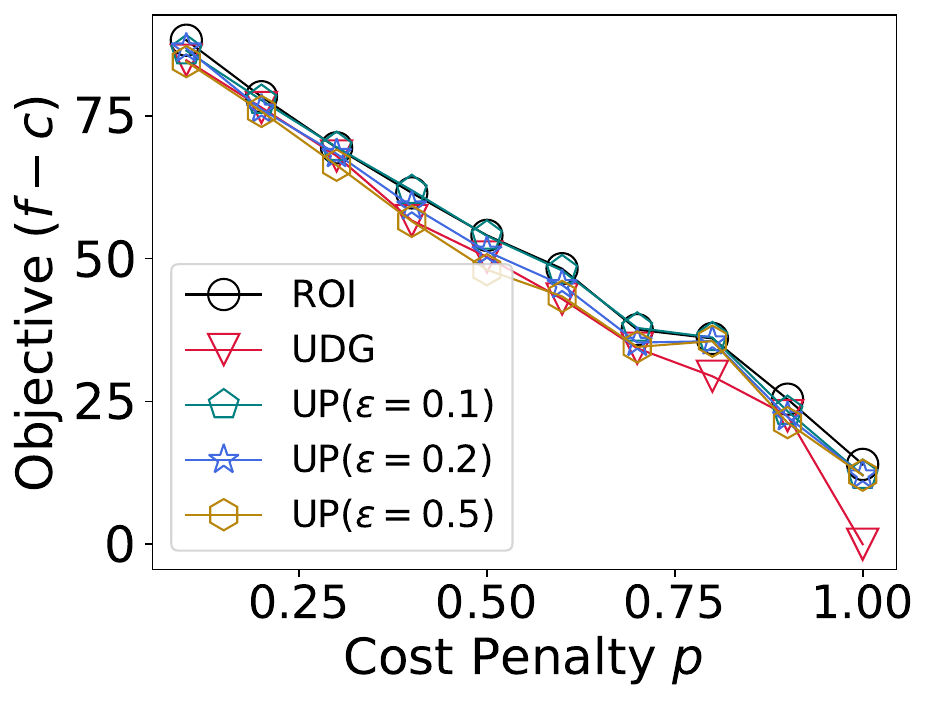}
\label{fig:bayesian-segment}
}
\subfigure[$p$ vs $f-c$; Housing]{
\includegraphics[width=0.22\textwidth]{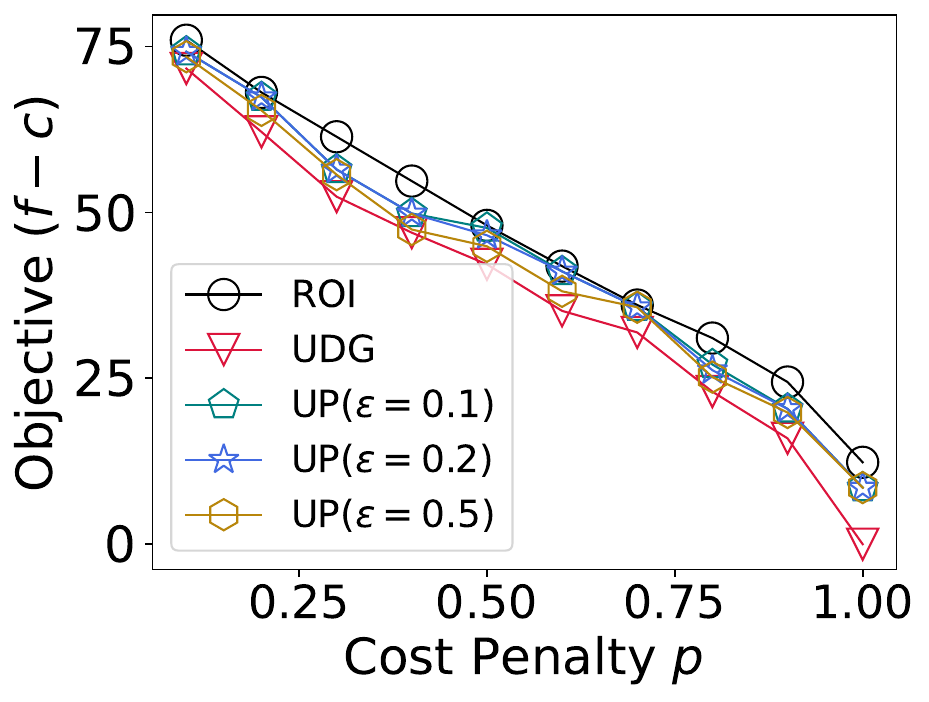}
\label{fig:bayesian-housing}
}
\caption{Experiments with Directed Vertex Cover (Eu-Email and Protein networks) and Bayesian A-Optimal Design (Segment and Housing data).}
\end{figure}


\paragraph{\bf Bayesian A-Optimal Design. }
Given a measurement matrix $\mathbf{X}=\left[x_1, x_2, \ldots, x_n\right] \in \mathbb{R}^{d}$, a linear model $\boldsymbol{y}_S=\mathbf{X}_S^{\mathrm{T}} \boldsymbol{\theta}+$ $\zeta_S$, a modular cost function $c:2^n \rightarrow \mathbb{R}_{>0}$, where $\boldsymbol{\theta}$ has a Gaussian prior distribution $\boldsymbol{\theta} \sim$ $\mathcal{N}(0, \Sigma)$, the normal i.i.d. noise $\zeta_1, \zeta_2, \cdots, \zeta_n \sim \mathcal{N}\left(0, \sigma^2 \right)$, the objective is to find a submatrix $\mathbf{X}_S$ such that 
\[\argmax _{S \subseteq\{1,2, \ldots, n\}} f(S)-c(S).\] 
Here, 
$f(S)=\operatorname{Tr}(\Sigma)-\operatorname{Tr}\left(\Sigma^{-1}+ \frac{1}{\sigma^{2}} \mathbf{X}_S \mathbf{X}_S^{\mathrm{T}}\right)^{-1}$ is the Bayesian A-Optimality function, which is a non-negative, monotone and $\gamma$-weakly submodular function \cite{bian2017guarantees, harshaw2019submodular}. $\operatorname{Tr}(\cdot)$ denotes the trace of a matrix. 
Moreover, the marginal gain of adding a measurement $e$ to $S$ can be efficiently calculated by $f(e | S)=\frac{||z_e||^2}{\sigma^2+\left\langle x_e, z_e\right\rangle}$, where $z_e=\mathbf{M}_S^{-1} x_e$ and $\mathbf{M}_S=\Sigma^{-1}+\mathbf{X}_S \mathbf{X}_S^T$. 
The modular cost of a measurement $e \in [n]$ is defined to be proportional to its $f$ value, formally, $c(S)=\sum_{e \in S} c(e) = p\sum_{e \in S}f(e)$; $p\in (0,1)$ is the cost penalty.

We use the Boston Housing \cite{harrison1978hedonic} and Segment Data \cite{qian2021multiobjective} for this application. To prepare the datasets for analysis, we performed pre-processing tasks that involved normalizing the features to achieve a zero mean and a standard deviation of 1. In the Bayesian A-Optimality function $f(S)$, we set $\sigma=1/\sqrt{d}$.
As the submodularity ratio is not known apriori, we implemented the $\gamma$-\textsc{Guess} algorithm ~\cite{harshaw2019submodular} to \emph{estimate} an approximate $\gamma$. 
The algorithm executes Algorithm $\textsc{UP}$ or UDG for a total of $T=\lceil \frac{1}{\delta}\ln \frac{1}{\delta} \rceil$ iterations. At each iteration $r$, there is a distinct submodularity ratio $\gamma_r$. The algorithm takes this $\gamma_r$ for the solution set selection in Line 4. The outputs of the $\gamma$-\textsc{Guess} algorithm is the set with the maximum objective value and the total number of function call.
We fix the decay threshold $\delta=0.2$ in $\gamma$-\textsc{Guess} and validate our algorithm with various cost penalty choices $p\in[0.1,1.0]$.

\subsection{Results}

\paragraph{\bf Profit Maximization}: The results varying cost penalties $\lambda_1$ and $\lambda_2$ on the benchmark networks are presented in Figures~\ref{fig:p2p-1}, \ref{fig:p2p-2}, \ref{fig:nethept-1} and \ref{fig:nethept-2}. 
The plots demonstrate that our results (Alg. UP with $\epsilon=0.1,0.2$) are close to the best-approximation algorithm ROI greedy algorithm (black-circle line) and exhibit a 15\% improvement over UDG. Regarding computational efficiency comparison in Figures \ref{fig:p2p-time} and \ref{fig:nethept-time}, our algorithm with $\epsilon=0.5$ matches the efficiency of the linear-time UDG algorithm and is 6.8 times faster than ROI. 

\paragraph{\bf Directed Vertex Cover}: Figures~\ref{fig:VC-email} and \ref{fig:VC-email} show that when the cost penalty thresholds increase, the objective values ($f-c$) also increase for every algorithm. Our algorithm demonstrates near-equivalent performance to the best-approximation algorithm \textsc{ROI}. Even with a larger $\epsilon=0.5$, our algorithm outperforms \textsc{UDG} by $25\%$ for the Protein network, particularly when dealing with high-cost penalties. The running time comparisons are depicted in Figure~\ref{fig:VC-time}.
For our algorithms with larger $\epsilon$ values, the running time aligns closely with the linear-time UDG algorithm. 

\paragraph{\bf Bayesian A-Optimal Design}:  Figures~\ref{fig:bayesian-segment}, \ref{fig:bayesian-housing}, and \ref{fig:bayesian-time} demonstrate that  
the quality of the result obtained from the modified ROI is validated against other algorithms, which is also supported by the approximation guarantee in Theorem \ref{th:modified-roi}.  Still, it needs the highest number of function calls. As for the objective values, the density maximization-based algorithms (modified ROI and UP) produce higher-quality solutions than UDG. Compared to the linear-time baseline UDG, our algorithm UP with a higher error ($\epsilon=0.5$) runs comparably fast.

%% file: sections/6-conclusion.tex
\section{Conclusion}
\label{sec:conclusion}

In this paper, we proposed a fast deterministic algorithm for balancing utility and modularity with a strong provable approximation guarantee.
We have experimentally validated our algorithms on various applications, demonstrating that our algorithms run as fast as the existing linear-time algorithms while producing comparative results as the state-of-the-art. A possible extension for future work is to investigate the approximation bound and performance when the cost function $c$ is (weakly) submodular. This would provide a more general framework and improve the applicability of the algorithms to a wider range of real-world applications.

%% file: sections/7.2-full-proof.tex
\section{Modified ROI-Greedy and Threshold Greedy Algorithms}
\label{sec:modified-algorithms}
In this section, we provide the modified algorithms and the approximation proofs that we mentioned in Section \ref{sec:unconstrained} and Section \ref{sec:experiments}.

\subsection{Modified ROI-Greedy for Weakly submdular $f$}
\label{sec:modified-roi}
As we have mentioned in Section \ref{sec:unconstrained} and Section \ref{sec:experiments}, we have modified the ROI-Greedy \cite{jin2021unconstrained} to make it handle the weakly submodular $f$ in the RUWSM problem. The pseudocode is presented in Algorithm \ref{alg:roi-weak}.

\begin{algorithm}[H]
   \caption{\textsc{$\gamma$-ROI}($f, c, \gamma$)}
   \label{alg:roi-weak}
\begin{algorithmic}[1]
   \STATE {\bfseries Input:} utility function $f$, cost function $c$, submodularity ratio $\gamma$.
   \STATE {\bfseries Output:} $\Tilde{S} = \argmax_{S_i, i \in [n]} f(S_i)-c(S_i)$.
   \vspace{.05in}
   \STATE $S_0 \gets \emptyset$
    \FOR{$i\gets 1$ to $n$}
    \STATE $v_i \gets \argmax_{u \in V \setminus S_{i-1}} \frac{f(u \mid S_{i-1})}{c(u)}$
        \IF{$f(u \mid S_{i-1}) > \gamma \cdot c(u)$}
            \STATE $S_i \gets S_{i-1} \cup \{v_i\}$  
        \ELSE
        \STATE {\sc Break}
        \ENDIF 
    \ENDFOR
\end{algorithmic}
\end{algorithm}

\subsubsection{The Algorithm}
The algorithm greedily selects an element $v_i$ with the maximal density w.r.t. the current seed set $S_{i-1}$ as a candidate element (line 5). It adds  $v_i$ to $S_{i-1}$ only if the density is larger than the submodularity ratio $\gamma$ (line 6). Finally, instead of outputting the solution set constructed at the termination of the algorithm, the algorithm outputs a solution set with the best objective value among all the intermediate solution sets. 

The algorithm will iterate at most $n$ times, and every time it takes $\mathcal{O}(n)$ oracle calls to $f$ to find the maximal density. Therefore, the worst-case runtime is bounded by $\mathcal{O}(n^2)$.

The proceeding theorem is the theoretical bound guaranteed by this algorithm.

\begin{theorem}
\label{th:unconstrained-greedy}
    Given a monotone $\gamma$-submodular function $f:2^{V} \rightarrow \mathbb{R}^{\geq 0}$ and a modular cost function $c$, after $\mathcal{O}(n^2)$ oracle calls,  Algorithm \ref{alg:roi-weak} outputs $\Tilde{S}$ such that
    \[f(\Tilde{S})-c(\Tilde{S}) \geq \gamma f(\OPT)-c(\OPT)- \frac{1}{\gamma}A\]
    where $A = c(\OPT) \log \frac{f(\OPT)}{c(\OPT)}$ and $\OPT =  \argmax_{T \subseteq V}\{f(T)-c(T)\}$.
\end{theorem}

\subsubsection{Analysis of Theorem \ref{th:unconstrained-greedy}}
The proof of Theorem \ref{th:unconstrained-greedy} can be adapted from the proof of Theorem \ref{th:unconstrained}. Next, we show how to modify some lemmas and results to derive Theorem \ref{th:unconstrained-greedy}.

\begin{lemma}
\label{lem:marginal-gain-g}
For all $1 \leq i \leq \ell$ and if $v_i$ be the element added to $S_{i-1}$, then
\[ 
f(v_i|S_{i-1}) \geq \gamma \frac{c(v_i)}{c(\OPT)} \big ( f(\OPT) - f(S_{i-1}) \big ) 
\]
\end{lemma}

\begin{proof} 
Given that $v_i$ is the element added to $S_{i-1}$ at line~5 of the Algorithm~\ref{alg:roi-weak}. Therefore,
$\forall u\in \OPT \setminus S_{i-1}$, $u$ was not the element selected with the maximal density w.r.t.  $f(S_{i-1})$. Hence, 
\begin{equation}
\frac{f(v_i|S_{i-1})}{c(v_i)}\geq \frac{f(u|S_{i-1})}{c(u)}
\label{eq:i-2-g}
\end{equation}

Since $v_i$ is added to $S_{i-1}$, we have: 
\begin{align*}
     f(\OPT)-f(S_{i-1}) &\leq f(\OPT \cup S_{i-1})-f(S_{i-1})\\
     &\leq \frac{1}{\gamma} \sum_{u\in \OPT \setminus S_{i-1}} f(u|S_{i-1})  \\
    &\leq \frac{1}{\gamma} \sum_{u\in \OPT \setminus S_{i-1}} c(u) \cdot \frac{f(v_i|S_{i-1})}{c(v_i)}  \quad\quad\quad\mbox{\it  due to Eq.~(\ref{eq:i-2-g})} \\
    &\leq \frac{1}{\gamma} \frac{f(v_i|S_{i-1})}{c(v_i)} c(\OPT)
\end{align*}

The first two inequalities hold from monotonicity and weak submodularity (Def. \ref{def:weak-sub}) of $f$. Rearranging the last inequality concludes the proof.  
\end{proof}

\begin{lemma}
\label{lem:f-Si-g}
For all $i\in [1,\ell]$ and if $v_i$ is the element added to $S_{i-1}$ in the $i$-th iteration then
\begin{equation}
    \label{eq:f-Si-g}
    f(S_{i}) \geq \left(1- \prod_{s=1}^{i} \left(1- \gamma \frac{c(v_i)}{c(\OPT)} \right)   \right) f(\OPT)
\end{equation}

\end{lemma}




Let $S_\ell$ with size $\ell$ be the set at the termination of Algorithm \ref{alg:roi-weak}. We divide the proof of Theorem~\ref{th:unconstrained-greedy} into two cases based on the relationship between $A$ and $c(S_\ell)$. 

    \textbf{Case 1}: $\gamma c(S_\ell) < A$.\\
    Consider each element $u \in \OPT \setminus S_\ell$ that was not included into $S_\ell$, we have
    \begin{equation}
    \label{eq:o1-g}
        \frac{f(u|S_\ell)}{c(u)} \leq \frac{f(u|S_u)}{c(u)} \leq \gamma
    \end{equation}
        
    Therefore, for Case 1, we have
    \begin{align*}
        f(\OPT)-f(S_\ell)-c(\OPT) &\leq \frac{1}{\gamma}\!\!\sum_{u \in \OPT \setminus S_{\ell}} f(u|S_\ell)-c(\OPT)\\
        &\leq \frac{1}{\gamma}\sum_{u \in \OPT \setminus S_{\ell}} f(u|S_\ell)-\sum_{u \in \OPT \setminus S_{\ell}} c(u) \\
        &\leq 0 \quad \quad\quad \quad \mbox{\it follows from  Eq. (\ref{eq:o1-g})}\\
        &\leq \frac{A}{\gamma}-c(S_\ell)
    \end{align*}

     Combining the above inequality with the Case 1 condition ($\gamma (1-\epsilon)c(S_\ell) < A$), we have 
     \begin{align*}
         f(\Tilde{S})-c(\Tilde{S})& \geq f(S_\ell)-c(S_\ell)\\
         &\geq \gamma f(\OPT)-c(\OPT)-\frac{1}{\gamma}A
     \end{align*}

    {This concludes the proof of the theorem for Case 1}. \hfill $\Box$

    \textbf{Case 2}: $\gamma c(S_\ell) \geq A$. \\
    The proof of this case follows from the proofs of Theorem \ref{th:unconstrained} with $\epsilon=0$.

\subsection{Threshold Greedy for Weakly submodular $f$}
\label{sec:modified-thresh}
\begin{algorithm}[tb]
   \caption{\textsc{threshold-ROI}($f, c, \gamma$)}
   \label{alg:roi-threshold}
\begin{algorithmic}[1]
   \STATE {\bfseries Input:} utility function $f$, cost function $c$, submodularity ratio $\gamma$.
   \STATE {\bfseries Output:} $\Tilde{S} = \argmax_{S_i, i \in [n]} f(S_i)-c(S_i)$.
   \vspace{.05in}
   \STATE $\tau \gets \max_{u \in V}\frac{f(u)}{c(u)}$, $S_0 \gets \emptyset$, $i \gets 1$
    \WHILE{$\tau > \gamma$}
    \FOR{$e \in V \setminus S_{i-1}$}
    \IF{$\frac{f(e\mid S_{i-1})}{c(e)} \geq \tau$}
    \STATE $S_i \gets S_{i-1} \cup \{e\}$
    \STATE $i++$
    \ENDIF
    \ENDFOR
    \STATE $\tau \gets (1-\epsilon) \tau$
    \ENDWHILE
\end{algorithmic}
\end{algorithm}

We provide more justifications for the statement (``traditional threshold adaption to the ROI-greedy admits unbounded runtime") in Section \ref{sec:unconstrained}. 

Threshold technique \cite{badanidiyuru2014fast} is usually used to design faster deterministic algorithms based on a greedy algorithm for submodular optimization problems. We show that the threshold technique can be applied to the unconstrained $f-c$ problem with a provable theoretical guarantee. However, the running time is unbounded. The pseudocode is presented in Algorithm \ref{alg:roi-threshold}.

\subsubsection{The Algorithm}
The algorithm initializes $S$ to an empty set. The threshold $\tau$ is initialized to be the largest density of a singleton (line 3), and it is dynamically updated (geometrically decay). The intuition behind the algorithm is that it starts by considering only the items with the high-density value. Then, the algorithm gradually relaxes the threshold value to consider more and more items. The parameter $\epsilon$ controls the rate at which the threshold value is relaxed.

The proceeding theorem is the theoretical bound guaranteed by this algorithm.

\begin{theorem}
\label{th:unconstrained-thresh}
    Given a monotone $\gamma$-submodular function $f:2^{V} \rightarrow \mathbb{R}^{\geq 0}$, a modular cost function $c$, an error threshold $\epsilon \in (0,1)$, after $\mathcal{O}(\frac{n}{\epsilon}\log \tau_0)$ oracle calls,  Algorithm \ref{alg:roi-weak} outputs $\Tilde{S}$ such that
    \[f(\Tilde{S})-c(\Tilde{S}) \geq \gamma(1-\epsilon) f(\OPT)-c(\OPT)- \frac{1}{\gamma(1-\epsilon)}A\]
    where $A = c(\OPT) \log \frac{f(\OPT)}{c(\OPT)}$, $\OPT =  \argmax_{T \subseteq V}\{f(T)-c(T)\}$ and $\tau_0 = \max_{u \in V}\frac{f(u)}{c(u)}$.
\end{theorem}

\textbf{\textit{Note that $\tau_0$ is unbounded since the density of an element can be exponentially large, so the running time of this Algorithm \ref{alg:roi-threshold} can be as bad as $\mathcal{O}(n^2)$ or even worse.}}

\subsubsection{Analysis of Theorem \ref{th:unconstrained-thresh}}
The proof method of the theoretical bound of Theorem \ref{th:unconstrained-thresh} is similar to the previous theorems. We first prove the lower bound of the marginal gain when adding an element to the seed set. Then, we prove the theorem by cases.

\begin{lemma}
\label{lem:marginal-gain-t}
For all $i$ such that $1 \leq i \leq \ell$ and if $v_i$ be the element added to $S_{i-1}$ then
\[ 
f(v_i|S_{i-1}) \geq \gamma(1-\epsilon) \frac{c(v_i)}{c(\OPT)} \big ( f(\OPT) - f(S_{i-1}) \big ) 
\]
\end{lemma}

\begin{proof}
    
 Let $\tau_i$ be the value of $\tau$ at some iteration. By the algorithm, we know that $v_i$ was not selected for the previous $\tau_{i-1}$ where
$\tau_{i-1} = \frac{\tau_i}{1-\epsilon}$. In other words, 
\begin{equation}
\forall ~ u\in \OPT \setminus S_{i-1}: \frac{f(u|S_{i-1})}{c(u)}\leq \frac{\tau_i}{1-\epsilon}
\label{eq:i-1-t}
\end{equation}

Since $v_i$ is added to $S_{i-1}$, we have: 
\begin{align*}
    & f(\OPT)-f(S_{i-1}) \leq \frac{1}{\gamma} \sum_{u\in \OPT \setminus S_{i-1}} f(u|S_{i-1}) \quad\mbox{\tt  by Def.~(\ref{def:weak-sub})} \\
    &\leq \frac{1}{\gamma} \sum_{u\in \OPT \setminus S_{i-1}} \frac{c(u)}{1-\epsilon} \cdot \frac{f(v_i|S_{i-1})}{c(v_i)}  \quad\quad\quad\mbox{\tt  due to Eq.~(\ref{eq:i-2})} \\
    &\leq \frac{1}{\gamma(1-\epsilon)} \frac{f(v_i|S_{i-1})}{c(v_i)} c(\OPT).
\end{align*}
       
  
\end{proof}


\textbf{\textit{Notice that the above lemma gives the same property as in the proof of Theorem \ref{th:unconstrained}. The remaining proofs of Theorem \ref{th:unconstrained-thresh} are the same as the corresponding proofs of Theorem \ref{th:unconstrained}.
Therefore, readers can find the rest of the proofs in the main paper. We jump right to the running time analysis.}}

\begin{lemma}
   \label{obs:unconstrained-time}
The time complexity of Algorithm~\ref{alg:roi-threshold} is $\mathcal{O}(\frac{n}{\epsilon}\log \frac{\tau}{\epsilon})$. 
\end{lemma}

The for-loop runs at most $n$ times for each choice of threshold $\tau$. Let $\tau_0 = \max_{e \in V}\{\frac{f(e)}{c(e)}\}$ be the initial threshold. We know that $\gamma$ is the threshold at termination of the algorithm. If $T$ is the number of runs of the while-loop, then
\[
T = 2+\log_{(1-\epsilon)^{-1}} \frac{\tau_0}{\gamma} = 2-\frac{\ln \tau_0-\ln \gamma}{\ln 1-\epsilon}
\leq 2+\epsilon^{-1} (\ln \tau_0 -\ln \gamma) \in O(\epsilon^{-1} \log \tau_0)
\]

So, the overall run time is $\mathcal{O}(\frac{n}{\epsilon}\log \tau_0)$.

%% file: sections/7.3-approximate.tex









\section{Complete Proofs for Theorem \ref{th:unconstrained-approx}}
\label{sec:unconstrained-approx}

In Algorithm \ref{alg:up}, we note that the input utility function is $\Tilde{f}$, which is $\delta$-approximate to a monotone $\gamma$-submodular set function $f$. Before presenting the theorem, we remark about the marginal gain of a $\delta$-approximate function $\Tilde{f}$.

\begin{remark}
By Definition \ref{def:approximate}, we have, for every $S \subseteq V$ and $v \in V \setminus S$
    \[ f(S) - \delta \leq \Tilde{f}(S) \leq f(S) + \delta\]
    \[ f(v \mid S) - 2\delta \leq \Tilde{f}(v \mid S) \leq \Tilde{f}(v \mid S) - 2\delta \]
\end{remark}


\begin{manualtheorem}{\textbf{3.10}}
\textit{Suppose the input utility function is $\delta$-approximate to a monotone $\gamma$-submodular set function $f$. $c$ is a modular cost function. Let $\OPT =  \argmax_{T \subseteq V}\{f(T)-c(T)\}$, after $\mathcal{O}(\frac{n}{\epsilon}\log \frac{n}{\gamma \epsilon})$ oracle calls to the $\delta$-approximate utility function,  Algorithm \ref{alg:up} outputs $\Tilde{S}$ such that}
    \begin{align*}
        f(\Tilde{S})-c(\Tilde{S}) \geq  &\gamma(1-\epsilon)f(\OPT)-c(\OPT)-\frac{1}{\gamma (1-\epsilon)}A \\
        &\quad\quad- 2\delta\left(\beta+ \frac{n}{\gamma} + 1 + n(1-\epsilon)\beta'\right)
    \end{align*}
    \textit{where $A = c(\OPT) \log \frac{f(\OPT)}{c(\OPT)}$, $\beta = \frac{c(\OPT)}{c_{\min}}, \beta' = \frac{c_{max}}{c(\OPT)}$ and  $\epsilon \in (0,1)$ is an error threshold.}
\end{manualtheorem}

%

\subsection{Overview of the Proof}
We now proceed by presenting a roadmap in Figure \ref{fig:roadmap-approx} of the proof for Theorem~\ref{th:unconstrained-approx} by illustrating the definitions, lemmas, and claims that are used to discharge the proof of the theorem. Denote $\xi=2\delta\left(\beta+ \frac{n}{\gamma} + 1 + n(1-\epsilon)\beta'\right)$. As depicted in the following tree, each node corresponds to a
theorem, lemma, claim, etc., and the incoming directed edge captures the necessity of the source node in
the proof of the destination node. Some of the edges are annotated with conditions under which the
source node leads to the destination. For instance, if the condition in Case 2.2
holds, Lemma \ref{lem:unconstrained-h_St-approx} is valid and implies the validity of Theorem \ref{th:unconstrained-approx}.  
In our proof, we will proceed with the leaf-level lemmas and claims and work our way towards the root of the tree. 
Wee will prove Lemmas~\ref{lem:marginal-gain-approx} and \ref{lem:f-Si-approx} first.


\begin{figure}
\begin{center}
\begin{tikzpicture}[scale=0.95]


\node (1) at (0,0) {\scriptsize Thm. \ref{th:unconstrained-approx}};
\node (2) at (-3,-1) {\scriptsize Def. \ref{def:weak-sub}, \ref{def:approximate}};
\node (3) at (3, -1) {$\Box$};
\node (4) at (-2, -3) {\scriptsize Alg. \ref{alg:up}, line~7, 18};
\node (5) at (7, -3) {\scriptsize Lem. \ref{lem:unconstrained-h_St-approx}};
\node (6) at (5, -4) {\scriptsize Claim \ref{claim:cost-St-approx}};
\node (7) at (7, -5) {\scriptsize Lem. \ref{lem:h-upper-bound-approx}};
\node (8) at (8, -4) {\scriptsize Cor. \ref{col:lower-opt-approx}};
\node (9) at (6, -6) {\scriptsize Lem. \ref{lem:f-Si-approx}};
\node (10) at (8, -6) {\scriptsize Claim \ref{claim:single-cost-approx}};
\node (11) at (7, -7) {\scriptsize Lem. \ref{lem:marginal-gain-approx}};
\draw[edge] (2) -- (1) node[midway,below,sloped] {{\bf \scriptsize Case 1}};
\draw[edge] (2) -- (1) node[midway,above,sloped] {{\bf \scriptsize $\scriptstyle \gamma(1-\epsilon)c(S_\ell) < A$}};

\draw[edge] (3) -- (1) node[midway,below,sloped] {{\bf \scriptsize Case 2}};
\draw[edge] (3) -- (1) node[midway,above,sloped] {{\bf \scriptsize $\scriptstyle\gamma(1-\epsilon)c(S_\ell) \geq A$}};

\draw[edge] (4) -- (3) node[midway,above,sloped] {{\bf \scriptsize Case 2.1}};
\draw[edge] (4) -- (3) node[midway,below,sloped] {{\bf \scriptsize $\scriptstyle f(S_j) \geq \gamma(1-\epsilon)f(\OPT)-\xi$}};

\draw[edge] (5) -- (3) node (12) [midway,above,sloped]  {{\bf \scriptsize Case 2.2}};
\draw[edge] (5) -- (3) node (12) [midway,below,sloped]  {{\bf \scriptsize $\scriptstyle f(S_j) < \gamma(1-\epsilon)f(\OPT)-\xi$}};

\draw[edge] (6) -- (5);
\draw[edge] (7) -- (5);
\draw[edge] (8) -- (5);
\draw[edge] (7) -- (8);
\draw[edge] (9) -- (7);
\draw[edge] (10) -- (7);
\draw[edge] (11) -- (9);
\draw[edge] (11) -- (10);

\end{tikzpicture}
\end{center}
\caption{Roadmap of the proof of Theorems \ref{th:unconstrained-approx} for Algorithm \ref{alg:up}}
\label{fig:roadmap-approx}
\end{figure}
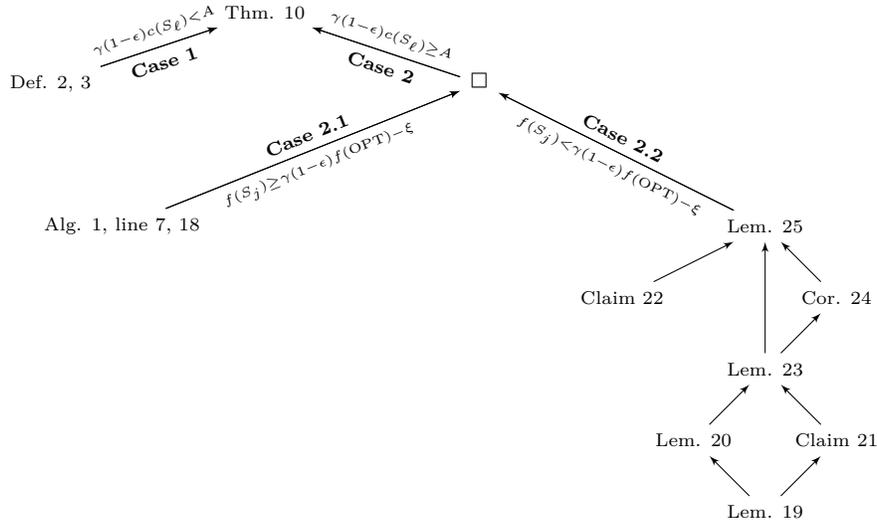

\subsection{Useful Lemmas}
Before we present the proof for the above theorem, we discuss two specific properties of Algorithm~\ref{alg:up}
in the following lemmas. Let $\ell$ be the value of $i$ at the termination of Algorithm \ref{alg:up}.

\begin{lemma}
\label{lem:marginal-gain-approx}
For all $1 \leq i \leq \ell$ and if $v_i$ be the element added to $S_{i-1}$, then
\[ 
f(v_i|S_{i-1}) \geq \gamma(1-\epsilon)\frac{c(v_i)}{c(\OPT)} \left( f(\OPT)-f(S_{i-1})\right) - 2\delta  \left( 1+n\frac{(1-\epsilon)c(v_i)}{ c(\OPT)} \right) 
\]
\end{lemma}

\begin{proof} 
Since that $v_i$ is the element added to $S_{i-1}$ at line~14 of the Algorithm~\ref{alg:up}, we have
$\displaystyle\frac{\Tilde{f}(v_i|S_{i-1})}{c(v_i)}\geq {(1-\epsilon)\tau_i}$ due to line~13 condition. Furthermore, let $\displaystyle\frac{\Tilde{f}(u|S_u)}{c(u)}$ be the density when element $u$ was at the priority queue at the current iteration. Since $v_i$ is on the top of the priority queue, 
$\forall u \in \OPT \setminus S_{i-1}$, $u$ is either equal to $v_i$ or has a density $\displaystyle\frac{\Tilde{f}(u|S_u)}{c(u)}$ smaller than $\tau_i$
in Line~13. We know $S_u \subseteq S_i$. Hence, 
\begin{equation}
\frac{\Tilde{f}(v_i|S_{i-1})}{c(v_i)}\geq {(1-\epsilon)\tau_i} \geq (1-\epsilon) \frac{\Tilde{f}(u|S_{u})}{c(u)}
\label{eq:i-2-approx}
\end{equation}

By the definition of $\delta$- approximate, 
\begin{align}
    (1-\epsilon) \frac{\Tilde{f}(u|S_{u})}{c(u)} &\leq \frac{\Tilde{f}(v_i|S_{i-1})}{c(v_i)} \nonumber \\
    \Rightarrow~~  (1-\epsilon) \frac{f(u|S_{i-1})-2\delta}{c(u)} \leq (1-\epsilon) \frac{f(u|S_{u})-2\delta}{c(u)} \leq (1-\epsilon) \frac{\Tilde{f}(u|S_{u})}{c(u)} &\leq \frac{\Tilde{f}(v_i|S_{i-1})}{c(v_i)} \leq \frac{f(v_i|S_{i-1})+2\delta}{c(v_i)} \nonumber\\
    \Rightarrow~~  (1-\epsilon) \frac{f(u|S_{i-1})-2\delta}{c(u)} &\leq \frac{f(v_i|S_{i-1})+2\delta}{c(v_i)} \nonumber\\
    \Rightarrow~~  f(u|S_{i-1}) &\leq \frac{c(u)}{1-\epsilon} \cdot \frac{f(v_i|S_{i-1})+2\delta}{c(v_i)}+2\delta
    \label{eq:approx-gain-approx}
\end{align}

Since $v_i$ is added to $S_{i-1}$, we have: 
\begin{align*}
     f(\OPT)-f(S_{i-1})  &\leq f(\OPT \cup S_{i-1})-f(S_{i-1})  \ \quad\quad\quad \mbox{\it as $f$ is monotone}\\
    &\leq \frac{1}{\gamma} \sum_{u\in \OPT \setminus S_{i-1}} f(u|S_{i-1}) \ \quad\quad\quad\quad \mbox{\it due to Def.~\ref{def:weak-sub}} \\
    &\leq \frac{1}{\gamma} \sum_{u\in \OPT \setminus S_{i-1}} \left[\frac{c(u)}{1-\epsilon} \cdot \frac{f(v_i|S_{i-1})+2\delta}{c(v_i)}+2\delta \right] \quad\quad\quad\quad\mbox{\it  due to Eq. (\ref{eq:approx-gain-approx})} \\
    &=\frac{1}{\gamma(1-\epsilon)} \cdot \frac{f(v_i|S_{i-1})+2\delta}{c(v_i)} \sum_{u\in \OPT \setminus S_{i-1}} c(u) + \sum_{u\in \OPT \setminus S_{i-1}} \frac{2\delta}{\gamma}  \\
    &\leq \frac{1}{\gamma(1-\epsilon)} \cdot \frac{c(\OPT)}{c(v_i)} \Big[f(v_i|S_{i-1})+2\delta \Big]  + \frac{2n\delta}{\gamma}\\
    &= \frac{1}{\gamma(1-\epsilon)} \cdot \frac{c(\OPT)}{c(v_i)} f(v_i|S_{i-1}) + \frac{2\delta}{\gamma} \left( \frac{c(\OPT)}{(1-\epsilon)c(v_i)} + n\right)
\end{align*}

Rearranging the last inequality concludes the proof. 
\begin{align*}
    f(v_i|S_{i-1}) &\geq \gamma(1-\epsilon)\frac{c(v_i)}{c(\OPT)} \left[f(\OPT)-f(S_{i-1}) - \frac{2\delta}{\gamma} \left( \frac{c(\OPT)}{(1-\epsilon)c(v_i)} + n\right) \right]\\
    &= \gamma(1-\epsilon)\frac{c(v_i)}{c(\OPT)} \Big[ f(\OPT)-f(S_{i-1})\Big] - 2\delta  \left( 1+n\frac{(1-\epsilon)c(v_i)}{ c(\OPT)} \right) 
\end{align*}
\end{proof}

\begin{lemma}
\label{lem:f-Si-approx}
For all $i$ such that $i$ such that $1 \leq i \leq \ell$ and if $v_i$ is the element added to $S_{i-1}$ in the $i$-th iteration and $c(v_i) \leq c(\OPT)$ then
\[ 
f(S_{i}) \geq \left(   1- \prod_{s=1}^{i} \left(1- \gamma(1-\epsilon) \frac{c(v_i)}{c(\OPT)} \right)   \right) f(\OPT) - 2\delta\left(\frac{\beta}{\gamma(1-\epsilon)}+ \frac{n}{\gamma}\right)
\]
\end{lemma}

\begin{proof} 
We prove this lemma by induction.

\textbf{Base Case: }
We verify the base case $i=1$ using Lemma \ref{lem:marginal-gain-approx}; note that
$f(\emptyset) = 0$.

\begin{align*}
    &f(S_1) = f(\{v_1\})=f(v|\emptyset) \\
    &\geq \gamma(1-\epsilon)\frac{c(v_1)}{c(\OPT)} \left[ f(\OPT)-f(S_{0})\right] - 2\delta  \left( 1+\frac{(1-\epsilon)c(v_1)}{ c(\OPT)} \right)\\
    &=\gamma(1-\epsilon)\frac{c(v_1)}{c(\OPT)}  f(\OPT) - 2\delta  \left( 1+\frac{(1-\epsilon)c(v_1)}{ c(\OPT)} \right)
\end{align*}
Since $c(v_1) \leq c(\OPT)$ and by the definition of $\beta$, $\beta = \frac{c(\OPT)}{c_{min}}$), we have:
\begin{align*}
    1+\frac{(1-\epsilon)c(v_1)}{ c(\OPT)} &\leq \frac{c(\OPT)}{c(v_1)} \left[ 1+\frac{(1-\epsilon)c(v_1)}{ c(\OPT)} \right]  \\
    &= \frac{c(\OPT)}{c(v_1)} + n(1-\epsilon) \frac{c(\OPT)}{c(v_1)}\cdot \frac{c(v_1)}{c(\OPT)}\\
    &\leq \beta + n(1-\epsilon) \\
    &\leq \frac{\beta}{\gamma(1-\epsilon)}+ \frac{n}{\gamma}
\end{align*}
Therefore, \[f(S_1) \geq \gamma(1-\epsilon)\frac{c(v_1)}{c(\OPT)} - 2\delta\left(\frac{\beta}{\gamma(1-\epsilon)}+ \frac{n}{\gamma}\right)\]

The above inequality is consistent with the one that we plug in $i=1$ to Lemma \ref{lem:f-Si-approx}. Thus, we proved this lemma for the base case.

\textbf{Inductive Steps: }
For $i>1$, as an inductive hypothesis (I.H.), we assume
\[\forall j \in [1, i -1]:~~
 f(S_{j}) \geq \left(   1- \prod_{s=1}^{j} \left(1- \gamma(1-\epsilon) \frac{c(v_s)}{c(\OPT)} \right) \right) f(\OPT) - 2\delta(\frac{\beta}{\gamma(1-\epsilon)}+\frac{n}{\gamma})
\]
is true. We need to prove the inequality for $f(S_{j+1})$, 

    \begin{align*}
        &f(S_{j+1}) = f(S_j)+f(v_{j+1}|S_j)\\
        &\geq f(S_j) + \gamma(1-\epsilon) \frac{c(v_{j+1})}{c(\OPT)} \left( f(\OPT) - f(S_j) \right) - 2\delta  \left( 1+n\frac{(1-\epsilon)c(v_{j+1})}{ c(\OPT)} \right)  \quad\quad\mbox{\it due to Lemma~\ref{lem:marginal-gain-approx}}\\
        &= \gamma(1-\epsilon) \frac{c(v_{j+1})}{c(\OPT)}f(\OPT) + \left(1-\gamma(1-\epsilon) \frac{c(v_{j+1})}{c(\OPT)} \right) f(S_j) - 2\delta  \left( 1+n\frac{(1-\epsilon)c(v_{j+1})}{ c(\OPT)} \right) \\
        &\geq \gamma(1-\epsilon) \frac{c(v_{j+1})}{c(\OPT)}f(\OPT) + \left(1-\gamma(1-\epsilon) \frac{c(v_{j+1})}{c(\OPT)} \right) \underbrace{\left[\left(   1- \prod_{s=1}^{j} \left(1- \gamma(1-\epsilon) \frac{c(v_s)}{c(\OPT)} \right)   \right)f(\OPT) - 2\delta(\frac{\beta}{\gamma(1-\epsilon)}+\frac{n}{\gamma}) \right]}_{\mbox{\it by induction hypothesis}} \\
        & \ \ \ \ \ \ \ \ \ \ \ \ \ \ \ \ \ \ \ \ \ \ \ \  - 2\delta  \left( 1+n\frac{(1-\epsilon)c(v_{j+1})}{ c(\OPT)} \right) \\
        &= \gamma(1-\epsilon) \frac{c(v_{j+1})}{c(\OPT)}f(\OPT) + \left(1-\gamma(1-\epsilon) \frac{c(v_{j+1})}{c(\OPT)} \right) \left(   1- \prod_{s=1}^{j} \left(1- \gamma(1-\epsilon) \frac{c(v_s)}{c(\OPT)} \right)   \right)f(\OPT)  \\
        & \ \ \ \ \ \ \ \ \ \ \ \ \ \ \ \ \ \ \ \ \ \ \ \ -\left(1-\gamma(1-\epsilon) \frac{c(v_{j+1})}{c(\OPT)} \right) 2\delta(\frac{\beta}{\gamma(1-\epsilon)}+\frac{n}{\gamma})   - 2\delta  \left( 1+n\frac{(1-\epsilon)c(v_{j+1})}{ c(\OPT)} \right)\\
        &= f(\OPT) - \left(1-\gamma(1-\epsilon) \frac{c(v_{j+1})}{c(\OPT)} \right)\left[\prod_{s=1}^{j} \left(1- \gamma(1-\epsilon) \frac{c(v_s)}{c(\OPT)} \right)\right]f(\OPT)\\
        & \ \ \ \ \ \ \ \ \ \ \ \ \ \ \ \ \ \ \ \ \ \ \ \ -2\delta \left[ \left(1-\gamma(1-\epsilon) \frac{c(v_{j+1})}{c(\OPT)} \right)(\frac{\beta}{\gamma(1-\epsilon)}+\frac{n}{\gamma}) +  \left( 1+n\frac{(1-\epsilon)c(v_{j+1})}{ c(\OPT)} \right)\right]  \\
        &= \left( 1- \prod_{s=1}^{j+1} (1- \gamma(1-\epsilon) \frac{c(v_s)}{c(\OPT)} )   \right) f(\OPT) -2\delta \left[ \left(1-\gamma(1-\epsilon) \frac{c(v_{j+1})}{c(\OPT)} \right)(\frac{\beta}{\gamma(1-\epsilon)}+\frac{n}{\gamma}) +  \left( 1+n\frac{(1-\epsilon)c(v_{j+1})}{ c(\OPT)} \right)\right]
    \end{align*} 

Next, we derive the second term,
\begin{align*}
    &\left(1-\gamma(1-\epsilon) \frac{c(v_{j+1})}{c(\OPT)} \right)(\frac{\beta}{\gamma(1-\epsilon)}+\frac{n}{\gamma}) +  \left( 1+n\frac{(1-\epsilon)c(v_{j+1})}{ c(\OPT)} \right)\\
    &= \frac{\beta}{\gamma(1-\epsilon)} - \gamma(1-\epsilon) \frac{c(v_{j+1})}{c(\OPT)}\cdot\frac{\beta}{\gamma(1-\epsilon)}  + \frac{n}{\gamma} - (1-\epsilon) n\frac{c(v_{j+1})}{c(\OPT)} + 1 + n\frac{(1-\epsilon)c(v_{j+1})}{ c(\OPT)}\\
    &= \frac{\beta}{\gamma(1-\epsilon)} - \frac{c(v_{j+1})}{c(\OPT)}\beta  + \frac{n}{\gamma} + 1 \\
    &= \frac{\beta}{\gamma(1-\epsilon)} + \frac{n}{\gamma} + 1 - \frac{c(v_{j+1})}{c(\OPT)}\cdot \frac{c(\OPT)}{c_{min}}\\
    & = \frac{\beta}{\gamma(1-\epsilon)} + \frac{n}{\gamma} + 1 - \frac{c(v_{j+1})}{c_{min}}\\
    &\leq \frac{\beta}{\gamma(1-\epsilon)} + \frac{n}{\gamma} 
\end{align*}

Thus, we proved for the inductive step,
\[f(S_{j+1}) \geq \left( 1- \prod_{s=1}^{j+1} (1- \gamma(1-\epsilon) \frac{c(v_s)}{c(\OPT)} )   \right) f(\OPT) - 2\delta \left(\frac{\beta}{\gamma(1-\epsilon)} + \frac{n}{\gamma} \right)\]
\end{proof}

\subsection{Theoretical Analysis of Theorem \ref{th:unconstrained-approx}}


Let $S_\ell$ be the set at the termination of Algorithm \ref{alg:up}. We divide the proof of Theorem~\ref{th:unconstrained-approx} into two cases based on the relationship between $A$ and $c(S_\ell)$. 

    \textbf{Case 1}: $\gamma (1-\epsilon)c(S_\ell) < A$.
    %
    
    Consider each element $u \in \OPT \setminus S_\ell$ that was not included into $S_\ell$; such an element can be categorized into two sets $O_1$ and $O_2$ with corresponding reasons:

     $\bullet$~ \textit{$O_1$ is the set of elements for which the line 7 condition was not satisfied.\ }
     
     The density of these elements is $\leq \gamma$ at some point before the algorithm ends. Therefore, for $u \in O_1$, if $\frac{f(u|S_u)}{c(u)}$ is the density (key) of $u$ the last time it appears in the priority queue, then 
     \[ \frac{\Tilde{f}(u|S_u)}{c(u)} \leq \gamma\]
     By the definition of approximate oracles,
     \[ \frac{f(u|S_\ell) - 2\delta}{c(u)} \leq \frac{f(u|S_u) - 2\delta}{c(u)} \leq \frac{\Tilde{f}(u|S_u)}{c(u)} \leq \gamma \]
     Then we have for $u \in O_1$
    \begin{equation}
        \label{eq:o1-approx}
     f(u|S_\ell) \leq \gamma c(u) + 2\delta       
    \end{equation}

    $\bullet$~ \textit{$O_2$ is the set of elements for which the line 18 condition was not satisfied.\ }\\
    In other words, these elements  have been considered for $\frac{\log \frac{n}{\gamma\epsilon}}{\epsilon} $ times. Note that for any $u \in O_2$, the initial density of $u$ in the priority queue is $\frac{f(u)}{c(u)}$. Then, suppose $\frac{f(u|S_u)}{c(u)}$ is the key of $u$ the last time it existed in $PQ$. Therefore, as $S_u \subseteq S_\ell$ we have 
    \[\frac{\Tilde{f}(u|S_u)}{c(u)} \leq (1-\epsilon)^{\frac{\log \frac{n}{\gamma \epsilon}}{\epsilon}}  \frac{\Tilde{f}(u)}{c(u)}\]
    By the definition of approximate oracles (Def. \ref{def:approximate}), we derive both terms in the above inequality
    \[\frac{f(u|S_\ell)-2\delta}{c(u)} \leq \frac{f(u|S_u)-2\delta}{c(u)} \leq \frac{\Tilde{f}(u|S_u)}{c(u))} \]
    
    \[\frac{\Tilde{f}(u|S_u)}{c(u)} \leq (1-\epsilon)^{\frac{\log \frac{n}{\gamma \epsilon}}{\epsilon}} \cdot \frac{\Tilde{f}(u)}{c(u)} \leq  (1-\epsilon)^{\frac{\log \frac{n}{\gamma \epsilon}}{\epsilon}} \cdot \frac{f(u)+\delta}{c(u)}\]
    Therefore, by combining the above two inequalities, we have
    \begin{equation*}
        \frac{f(u|S_\ell)-2\delta}{c(u)} \leq (1-\epsilon)^{\frac{\log \frac{n}{\gamma \epsilon}}{\epsilon}} \cdot \frac{f(u)+\delta}{c(u)}
    \end{equation*}
    Then, for every $u \in O_2$,
    \begin{align}
        f(u|S_\ell) &\leq (1-\epsilon)^{\frac{\log \frac{n}{\gamma \epsilon}}{\epsilon}} [ f(u) + \delta ] + 2\delta \nonumber\\
        & = (1-\epsilon)^{\frac{\log \frac{n}{\gamma \epsilon}}{\epsilon}} f(u) + 2\delta +  (1-\epsilon)^{\frac{\log \frac{n}{\gamma \epsilon}}{\epsilon}} \delta
        \label{eq:o2-approx}
    \end{align}

    Therefore, we derive
    \begin{align*}
        &f(\OPT)-f(S_\ell)-c(\OPT) \\
        &\leq f(\OPT \cup S_\ell)-f(S_\ell)-c(\OPT) \quad \quad\quad \quad \mbox{\it due to monotonicity of $f$}\\
        &\leq \frac{1}{\gamma}\sum_{u \in \OPT \setminus S_{\ell}} f(u|S_\ell)-c(\OPT) \quad \quad\quad \quad \quad \mbox{\it due to Definition~\ref{def:weak-sub}}\\
        &\leq \frac{1}{\gamma}\left(\sum_{u \in O_1} f(u|S_\ell)+\sum_{u \in O_2} f(u|S_\ell)\right)-c(\OPT)\\
        &\leq \frac{1}{\gamma}\left(\sum_{u \in O_1} f(u|S_\ell)+\sum_{u \in O_2} f(u|S_\ell)\right)-\sum_{u \in O_1} c(u) \quad \quad\quad \quad\quad \quad \mbox{\it as $c$ is modular}\\
        &=\left(\sum_{u \in O_1} \frac{1}{\gamma} f(u|S_\ell) - c(u)\right)~~+~~\frac{1}{\gamma}\sum_{u \in O_2} f(u|S_\ell)\\
        &\leq \sum_{u \in O_1} \frac{2\delta}{\gamma} ~~+~~ \frac{1}{\gamma} \sum_{u \in O_2} f(u|S_\ell) \quad \quad \quad\quad \quad\quad \quad \mbox{\it from Eq. (\ref{eq:o1-approx})}\\
        &\leq \sum_{u \in O_1} \frac{2\delta}{\gamma} ~~+~~ \frac{1}{\gamma} \sum_{u \in O_2}\left[ (1-\epsilon)^{\frac{\log \frac{n}{\gamma \epsilon}}{\epsilon}}f(u) + 2\delta +  (1-\epsilon)^{\frac{\log \frac{n}{\gamma \epsilon}}{\epsilon}} \delta \right] \quad \quad\quad\quad \quad \quad \mbox{\it from Eq. (\ref{eq:o2-approx})}\\
        &\leq \sum_{u \in O_1 \cup O_2} \frac{2\delta}{\gamma} ~~+~~ \frac{1}{\gamma} \sum_{u \in O_2}\left[ (1-\epsilon)^{\frac{\log \frac{n}{\gamma \epsilon}}{\epsilon}}\Big(f(u) + \delta \Big) \right]\\
        &\leq \frac{2n\delta}{\gamma} ~~+~~ \frac{1}{\gamma} \sum_{u \in O_2}\left[ (1-\epsilon)^{\frac{\log \frac{n}{\gamma \epsilon}}{\epsilon}}\Big(f(u) + \delta \Big) \right]\\
        &\leq \frac{2n\delta}{\gamma} ~~+~~  \epsilon \delta+ \epsilon f(\OPT) 
    \end{align*}
    To get the last inequality, we need to show the following holds
    \begin{align*}
        \frac{1}{\gamma} \sum_{u \in O_2}\left[ (1-\epsilon)^{\frac{\log \frac{n}{\gamma \epsilon}}{\epsilon}}\Big(f(u) + \delta \Big) \right] &\leq   \frac{1}{\gamma} \sum_{u \in O_2} (1+\epsilon)^{-\frac{\log \frac{n}{\gamma \epsilon}}{\epsilon}}\Big(f(u) + \delta \Big) \quad \quad\quad \quad \mbox{\it due to $1-\epsilon \leq \frac{1}{1+\epsilon}$}\\
        &\leq  \frac{1}{\gamma} \sum_{u \in O_2} (1+\epsilon)^{-\frac{\log \frac{n}{\gamma \epsilon}}{\log (1+\epsilon)}}\Big(f(u) + \delta \Big) \quad \quad\quad \quad \mbox{\it due to $\epsilon \leq \log (1+\epsilon)$}\\
        &=    \frac{1}{\gamma} \sum_{u \in O_2} (1+\epsilon)^{-\log_{1+\epsilon} \frac{n}{\gamma \epsilon}}\Big(f(u) + \delta \Big) \\
        &=   \frac{1}{\gamma} \sum_{u \in O_2} \frac{\gamma \epsilon}{n} \Big(f(u) + \delta \Big)\\
        &\leq  \epsilon \delta+ \epsilon f(\OPT) 
    \end{align*}
    
    The last inequality holds because $O_2 \subseteq \OPT$, $|\OPT| \leq n$ and $f(u) \leq f(\OPT)$.
     Thus, rearranging the above inequality and combining it with the condition for Case 1 ($0< \frac{1}{\gamma (1-\epsilon)}A-c(S_\ell)$), we obtain
     \[ 
     f(\OPT)-f(S_\ell)-c(\OPT) - \epsilon f(\OPT) - \frac{2n\delta}{\gamma}  - \epsilon \delta \leq 0 \leq \frac{1}{\gamma (1-\epsilon)}A-c(S_\ell) 
     \]
     
     So, 
     \begin{align*}
         &f(\Tilde{S})-c(\Tilde{S}) \geq f(S_\ell)-c(S_\ell)\\
         &\geq (1-\epsilon)f(\OPT)-c(\OPT)-\frac{1}{\gamma (1-\epsilon)}A - \frac{2n\delta}{\gamma} - \epsilon \delta\\
         &\geq \gamma(1-\epsilon)f(\OPT)-c(\OPT)-\frac{1}{\gamma (1-\epsilon)}A - 2\delta\left(\beta+ \frac{n}{\gamma} + 1 + n(1-\epsilon)\beta'\right)
     \end{align*}
    {This concludes the proof of the theorem for case 1}. \hfill $\Box$

\textbf{Case 2}: $\gamma (1-\epsilon)c(S_\ell) \geq A$. 

In this case, there exist $t\in[1, \ell]$ such that:
    \begin{equation}
        c(S_{t-1})< \frac{1}{\gamma (1-\epsilon)} A  \leq c(S_t)
        \label{eq:th1-case2-approx}
    \end{equation}
Denote $\xi = 2\delta\left(\beta+ \frac{n}{\gamma} + 1 + n(1-\epsilon)\beta'\right)$. 
Consider $j \in [1, t-1]$, we further divide Case 2 into two sub-cases: (2.1) and (2.2), where for some $j \in [1,t-1]$, sub-case 2.1 corresponds to $f(S_{j}) \geq \gamma (1-\epsilon)f(\OPT) - \xi$  and sub-case 2.2 corresponds to $f(S_{j}) < \gamma (1-\epsilon)f(\OPT)-\xi$ for all $j \in [1,t-1]$.

\textbf{Case 2.1}: For some $j\in [1,t-1]$, $f(S_{j}) \geq \gamma (1-\epsilon)f(\OPT) - \xi$.

Note that Algorithm~\ref{alg:up} outputs $\Tilde{S} = \argmax_{S_i, i \in [n]} f(S_i)-c(S_i)$.
Therefore, 
\begin{align*}
f(\Tilde{S}) - c(\Tilde{S}) & > f(S_j) - c(S_j) \\
&\geq  \gamma (1-\epsilon)f(\OPT) - \xi - \displaystyle\frac{A}{\gamma(1-\epsilon)}  \quad\quad\quad\quad \mbox{\it follows from  condition in Case 2.1 and Eq.~(\ref{eq:th1-case2-approx})} \\
& \geq  \gamma (1-\epsilon)f(\OPT) - c(\OPT) - \displaystyle\frac{A}{\gamma(1-\epsilon)} - \xi
\end{align*}

This concludes the proof of the theorem for sub-case 2.1.  \hfill $\Box$

\textbf{Case 2.2}: For all $j\in [1,t-1]$, $f(S_{j}) < \gamma (1-\epsilon)f(\OPT) - \xi$.

For the \textbf{sub-case 2.2}, as illustrated in the proof-path, we have already
proved the Lemmas~\ref{lem:marginal-gain-approx} and \ref{lem:f-Si-approx}. We will now prove
the necessary properties conditioned on the sub-case 2.2. 

\begin{claim}
\label{claim:single-cost-approx}
Under the Case 2.2 (i.e., 
there exist $t\in[1, \ell]$ such that
$c(S_{t-1})< \frac{1}{\gamma (1-\epsilon)} A  \leq c(S_t)$
and for all $j\in[1, t-1]$, $f(S_{j}) < \gamma (1-\epsilon)f(\OPT) -\xi$),
the following holds: 
 \[\forall~j \in [1, t-1]: ~c(v_j) \leq c(\OPT)\]
\end{claim}

\begin{proof}
We prove this claim by contradiction. Assume there exist some $j \in [1, t-1]$ that $c(v_j)>c(\OPT)$. Also, recall that as per the condition in sub-case 2.2, we have
\begin{align}
    &f(S_j) < \gamma(1-\epsilon)f(\OPT) - \xi \nonumber\\ 
\Leftrightarrow &~~ f(\OPT) > \displaystyle\frac{1}{\gamma(1-\epsilon)} \Big(f(S_j)+\xi\Big) > f(S_j) >f(S_{j-1})
\label{eq:case2.2-approx}
\end{align}

Therefore, 
\begin{align*}
    f(S_j)&= f(S_{j-1}) + f(v_j|S_{j-1})\\
    &\geq f(S_{j-1}) + \underbrace{\gamma(1-\epsilon) \frac{c(v_j)}{c(\OPT)} \big ( f(\OPT) - f(S_{j-1}) \big ) - 2\delta  \left( 1+n\frac{(1-\epsilon)c(v_i)}{ c(\OPT)} \right)}_{\mbox{\it  follows from Lemma~\ref{lem:marginal-gain-approx}}} \\
    &= \gamma(1-\epsilon)\frac{c(v_j)}{c(\OPT)} f(\OPT) + \Big[ 1- \gamma(1-\epsilon)\frac{c(v_j)}{c(\OPT)} \Big]f(S_{j-1}) - 2\delta  \left( 1+n\frac{(1-\epsilon)c(v_i)}{ c(\OPT)} \right)\\
    &\geq { \gamma(1-\epsilon)\frac{c(v_j)}{c(\OPT)} f(\OPT)} + \gamma(1-\epsilon)\Big[ 1- \frac{c(v_j)}{c(\OPT)} \Big]f(S_{j-1})- 2\delta  \left( 1+n\frac{(1-\epsilon)c(v_i)}{ c(\OPT)} \right)\\
    &= {\gamma(1-\epsilon) f(\OPT) + \gamma(1-\epsilon)\Big[ \frac{c(v_j)}{c(\OPT)} - 1 \Big]f(\OPT)} \\ &\quad \quad - \gamma(1-\epsilon)\Big[ \frac{c(v_j)}{c(\OPT)} - 1 \Big] f(S_{j-1})- 2\delta  \left( 1+n\frac{(1-\epsilon)c(v_i)}{ c(\OPT)} \right)\\
    &= \gamma(1-\epsilon) f(\OPT) \\ &\quad\quad + \gamma(1-\epsilon)\Big[ \frac{c(v_j)}{c(\OPT)} - 1 \Big] \big( f(\OPT) - f(S_{j-1}) \big)- 2\delta  \left( 1+n\frac{(1-\epsilon)c(v_i)}{ c(\OPT)} \right)\\
    &\geq \gamma(1-\epsilon) f(\OPT) - 2\delta  \left( 1+n\frac{(1-\epsilon)c(v_i)}{ c(\OPT)} \right)  \ \ \ \ \ \ \ \ \ \ \mbox{\it follows from the assumption $c(v_j)>c(\OPT)$ and Eq.~(\ref{eq:case2.2-approx})}\\
    &\geq \gamma(1-\epsilon) f(\OPT) - 2\delta  \left( 1+n(1-\epsilon)\beta' \right) \ \ \ \ \ \ \ \ \ \ \ \mbox{\it follows from $\beta' = \frac{c_{max}}{c(\OPT)}$ and $\frac{c(v_i)}{ c(\OPT)} \leq \beta'$}\\
    &\geq \gamma(1-\epsilon) f(\OPT) - \xi
\end{align*}

The last inequality hold from $\xi = 2\delta\left(\beta+ \frac{n}{\gamma} + 1 + n(1-\epsilon)\beta'\right)$ and $\beta>0$.
Note that, using the assumption, we have concluded that $f(S_j) \geq \gamma(1-\epsilon) f(\OPT) -\xi$, which contradicts the premise of the claim (see Eq.~(\ref{eq:case2.2-approx})). Therefore, the assumption is incorrect, and the claim holds.
\end{proof}
            
\begin{claim}
\label{claim:cost-St-approx}
Under the Case 2.2 (i.e., 
there exist $t\in[1, \ell]$ such that
$c(S_{t-1})< \frac{1}{\gamma (1-\epsilon)} A  \leq c(S_t)$
and for all $j\in[1, t-1]$, $f(S_{j}) < \gamma (1-\epsilon)f(\OPT)-\xi$), define $B_j=A-\gamma(1-\epsilon) c(S_{j})$, the following holds: 
\[ \forall~j \in [1, t-1]:~ c(v_t) > B_{t-1}\]
\end{claim}

\begin{proof}

\begin{align*}
c(v_t) &= c(S_t) - c(S_{t-1})\\
&\geq \frac{1}{\gamma(1-\epsilon)}A- c(S_{t-1}) \quad\quad\quad {\mbox{\it follows from the case 2 condition Eq.~(\ref{eq:th1-case2-approx})}} \\
& =  \frac{1}{\gamma(1-\epsilon)}\big(A-\gamma (1-\epsilon)c(S_{t-1})\big) \\
&\geq  A- \gamma (1-\epsilon)c(S_{t-1}) =  B_{t-1}
\end{align*}

\end{proof}

\begin{lemma} 
\label{lem:h-upper-bound-approx}
Under the case 2.2 (i.e., 
there exist $t\in[1, \ell]$ such that
$c(S_{t-1})< \frac{1}{\gamma (1-\epsilon)} A  \leq c(S_t)$
and for all $j\in[1, t-1]$, $f(S_{j}) < \gamma (1-\epsilon)f(\OPT)-\xi$), the following
holds:
for all $j\in [1,t-1]$,
if $B_j=A-\gamma(1-\epsilon) c(S_{j})$, then
\[ 
   f(\OPT) - c(\OPT) - 2\delta\left(\frac{\beta}{\gamma(1-\epsilon)}+ \frac{n}{\gamma}\right) \leq \frac{B_j(f(\OPT) - f(S_{j}))}{c(\OPT)} +f(S_{j}) 
\]
\end{lemma}

\begin{proof}

We divide this proof into two cases.

$\bullet$ For \textbf{$B_j > c(\OPT)$}, the lemma holds immediately.

\smallskip
$\bullet$ We focus on \textbf{$B_j \leq c(\OPT)$}.

\begin{align*}
    & \displaystyle\frac{B_j(f(\OPT) - f(S_{j}))}{c(\OPT)} +f(S_{j}) \\
        &= \displaystyle\left(1-\frac{B_j}{c(\OPT)}\right)f(S_{j}) + \frac{B_j}{c(\OPT)}f(\OPT)\\
        &\geq  \displaystyle\left(1-\frac{B_j}{c(\OPT)}\right)\underbrace{\left[\left(   1- \prod_{k=1}^{j}  \left(1- \gamma(1-\epsilon) \frac{c(v_k)}{c(\OPT)} \right)   \right) f(\OPT)- 2\delta\left(\frac{\beta}{\gamma(1-\epsilon)}+ \frac{n}{\gamma}\right)\right]}_{\mbox{\it follows from Lemma~\ref{lem:f-Si-approx}}} + \displaystyle\frac{B_j}{c(\OPT)}f(\OPT)\\
        &=f(\OPT) - \displaystyle\left(\left(1-\frac{B_j}{c(\OPT)}\right)\prod_{k=1}^{j} \left(1- \gamma(1-\epsilon) \frac{c(v_k)}{c(\OPT)} \right)\right)f(\OPT) - \left(1-\frac{B_j}{c(\OPT)}\right)2\delta\left(\frac{\beta}{\gamma(1-\epsilon)}+ \frac{n}{\gamma}\right)\\
        &=\left( 1- \displaystyle\left(1-\frac{B_j}{c(\OPT)}\right)\prod_{k=1}^{j} \left(1- \gamma(1-\epsilon) \frac{c(v_k)}{c(\OPT)} \right) \right)f(\OPT) - \left(1-\frac{B_j}{c(\OPT)}\right)2\delta\left(\frac{\beta}{\gamma(1-\epsilon)}+ \frac{n}{\gamma}\right)
\end{align*}

We proceed with the following known inequality: for any 
\[x_1, x_2, \cdots, x_n, y \in \mathbb{R}^{+} \mbox{ and } x_i \leq y \mbox{ for } i \leq n\]
\begin{equation}
1-\prod_{i=1}^n\left(1-\frac{x_i}{y}\right) \geq 1-\left(1-\frac{\sum_{i=1}^n x_i}{n y}\right)^n
\label{eq:prodsum-approx}
\end{equation}
Consider $\forall k\in [1,j], x_k = \gamma(1-\epsilon) c(v_k)$, $x_{k+1} = B_j$ and $y = c(\OPT)$. Note that, 
$\forall k\in [1,j], \gamma(1-\epsilon) c(v_k) \leq c(\OPT)$, which follows from Claim~\ref{claim:single-cost-approx} and 
$\gamma(1-\epsilon) \leq 1$. Also, $B_j \leq c(\OPT)$ follows from the condition for this case. 

Therefore, by applying Eq.~(\ref{eq:prodsum-approx}), we derive the following

\begin{align*}
    &\displaystyle\frac{B_j(f(\OPT) - f(S_{j}))}{c(\OPT)} +f(S_{j}) \\
    &\geq \left( 1- \displaystyle\left(1- \frac{B_j+\gamma(1-\epsilon)\cdot \sum_{k=1}^{j}c(v_k)}{(j+1)\cdot c(\OPT)} \right)^{j+1} \right)f(\OPT)- \left(1-\frac{B_j}{c(\OPT)}\right)2\delta\left(\frac{\beta}{\gamma(1-\epsilon)}+ \frac{n}{\gamma}\right)
\end{align*}

Since $B_j \leq c(\OPT)$, $0 \leq \left(1-\frac{B_j}{c(\OPT)}\right)\leq 1$. Proceeding further,

\begin{align*}
    &\displaystyle\frac{B_j(f(\OPT) - f(S_{j}))}{c(\OPT)} +f(S_{j}) \\
    &\geq \left( 1- \displaystyle\left(1- \frac{B_j+\gamma(1-\epsilon)\cdot \sum_{k=1}^{j}c(v_k)}{(j+1)\cdot c(\OPT)} \right)^{j+1} \right)f(\OPT) - 2\delta\left(\frac{\beta}{\gamma(1-\epsilon)}+ \frac{n}{\gamma}\right)\\
    &= \left( 1- \displaystyle\left(1- \frac{A-\gamma(1-\epsilon)c(S_j)+\gamma(1-\epsilon)\cdot \sum_{k=1}^{j}c(v_k)}{(j+1)\cdot c(\OPT)} \right)^{j+1} \right)f(\OPT)- 2\delta\left(\frac{\beta}{\gamma(1-\epsilon)}+ \frac{n}{\gamma}\right)\\
    & \ \ \ \ \ \ \ \ \ \ \ \ \ \ \ \ \ \ \ \ \ \ \ \ \ \ \ \ \mbox{\it follows from the definition of $B_j$} \\[1em]
        &= \left( 1- \displaystyle\left( 1- \frac{A}{(j+1)\cdot c(\OPT)} \right)^{j+1} \right)f(\OPT) - 2\delta\left(\frac{\beta}{\gamma(1-\epsilon)}+ \frac{n}{\gamma}\right)\quad\quad\quad\mbox{\it as $\sum_{k=1}^{j}c(v_k) = c(S_j)$}\\[1em]
        &= \left( 1- \displaystyle\left( 1- \frac{1}{j+1}\cdot \ln \frac{f(\OPT)}{c(\OPT)} \right)^{j+1} \right)f(\OPT) - 2\delta\left(\frac{\beta}{\gamma(1-\epsilon)}+ \frac{n}{\gamma}\right) \ \ \ \ \ \ \ \ \ \ \ \ \ \ \ \ \mbox{\it as $A=c(\OPT)\ln\frac{f(\OPT)}{c(\OPT)}$} \\[1em]
        &\geq \left(1- e^{-\ln \frac{f(\OPT)}{c(\OPT)}}\right)f(\OPT) - 2\delta\left(\frac{\beta}{\gamma(1-\epsilon)}+ \frac{n}{\gamma}\right) \quad \quad \quad \mbox{\it as $(1- \frac{c}{t} )^t \leq e^{-c}$ for any positive $t$} \\[1em]
        &=\left( 1- \frac{c(\OPT)}{f(\OPT)} \right)f(\OPT) - 2\delta\left(\frac{\beta}{\gamma(1-\epsilon)}+ \frac{n}{\gamma}\right)\\
        &= f(\OPT)-c(\OPT) - 2\delta\left(\frac{\beta}{\gamma(1-\epsilon)}+ \frac{n}{\gamma}\right)
\end{align*}

    
\end{proof}
   
\begin{corollary} 
\label{col:lower-opt-approx}
Under the Case 2.2 (i.e., 
there exist $t\in[1, \ell]$ such that
$c(S_{t-1})< \frac{1}{\gamma (1-\epsilon)} A  \leq c(S_t)$
and for all $j\in[1, t-1]$, $f(S_{j}) < \gamma (1-\epsilon)f(\OPT) - \xi$), the following holds: 
for all $j \in  [1, t-1],$ 
if $B_j=A-\gamma(1-\epsilon) c(S_{j})$ then \[f(\OPT) -f(S_{j}) \geq \displaystyle\frac{1}{\gamma(1-\epsilon)}c(\OPT)\]
\end{corollary}

\begin{proof}

   Note that, for the case $f(S_{j}) -c(S_{j})
            \geq \gamma (1-\epsilon)f(\OPT)-c(\OPT)-\displaystyle\frac{1}{\gamma (1-\epsilon)}A -\xi$ , Theorem~\ref{th:unconstrained-approx} holds immediately. 
   Hence, we consider the case where $f(S_{j}) -c(S_{j})
            <\gamma (1-\epsilon)f(\OPT)-c(\OPT)-\displaystyle\frac{1}{\gamma (1-\epsilon)}A -\xi$. 
   
   \begin{align*}
       f(S_{j}) -c(S_{j}) &
        <\gamma (1-\epsilon)f(\OPT)-c(\OPT)-\frac{1}{\gamma (1-\epsilon)}A -\xi\\
         &\leq \gamma (1-\epsilon)\Big(f(\OPT)-c(\OPT)\Big)-\frac{1}{\gamma (1-\epsilon)}A -\xi\\
        &\leq \gamma (1-\epsilon)\Big(f(\OPT)-c(\OPT)\Big)-\frac{1}{\gamma (1-\epsilon)}A -\xi\\
        &=\gamma (1-\epsilon)\Big(f(\OPT)-c(\OPT)\Big)- c(S_{j}) - \Big(\frac{1}{\gamma (1-\epsilon)}A -c(S_{j})\Big) -\xi\\
        &=\gamma (1-\epsilon)\Big(f(\OPT)-c(\OPT)\Big)- c(S_{j}) - \frac{1}{\gamma (1-\epsilon)}\Big(A -\gamma (1-\epsilon)c(S_{j})\Big) -\xi\\
        &=\gamma (1-\epsilon)\Big(f(\OPT)-c(\OPT)\Big)- c(S_{j}) - \frac{1}{\gamma (1-\epsilon)}B_j -\xi\\
        &\leq \gamma (1-\epsilon) \underbrace{\left[\displaystyle\frac{B_j(f(\OPT) - f(S_{j}))}{c(\OPT)} +f(S_{j})+2\delta\left(\frac{\beta}{\gamma(1-\epsilon)}+ \frac{n}{\gamma}\right)\right]}_{\mbox{\it due to Lemma~\ref{lem:h-upper-bound-approx}}}  - c(S_{j}) - \frac{1}{\gamma (1-\epsilon)}B_j  -\xi \\
        &\leq \displaystyle\frac{B_j(f(\OPT) - f(S_{j}))}{c(\OPT)} +f(S_{j}) - c(S_{j}) - \frac{1}{\gamma (1-\epsilon)}B_j  + \underbrace{2\delta(\beta ~+~ \frac{n}{\gamma}) ~~-~~ \xi}_{\mbox{\it $<0$ since $\xi = 2\delta\left(\beta+ \frac{n}{\gamma} + 1 + n(1-\epsilon)\beta'\right)$ }}\\
        &\leq \displaystyle\frac{B_j(f(\OPT) - f(S_{j}))}{c(\OPT)} +f(S_{j}) - c(S_{j}) - \frac{1}{\gamma (1-\epsilon)}B_j
   \end{align*}
   
    The above inequality implies:
    \[ \frac{B_j(f(\OPT) - f(S_{j}))}{c(\OPT)} \geq \frac{1}{\gamma (1-\epsilon)}B_j \]

    Rearranging the term concludes the proof of the corollary.
    \[f(\OPT) -f(S_{j}) \geq \frac{1}{\gamma(1-\epsilon)}c(\OPT)\]
\end{proof}


The following lemma discharges the proof for Theorem~\ref{th:unconstrained-approx} for the sub-case 2.2.

\begin{lemma} 
    \label{lem:unconstrained-h_St-approx}
Under the Case 2.2 (i.e., 
there exist $t\in[1, \ell]$ such that
$c(S_{t-1})< \frac{1}{\gamma (1-\epsilon)} A  \leq c(S_t)$
and for all $j\in[1, t-1]$, $f(S_{j}) < \gamma (1-\epsilon)f(\OPT)-\xi$),
the following holds:
    \[f(S_t)-c(S_t) \geq \gamma(1-\epsilon)f(\OPT)-c(\OPT) - \frac{1}{\gamma(1-\epsilon)}A - 2\delta\left(\beta+ \frac{n}{\gamma} + 1 + n(1-\epsilon)\beta'\right) \]
\end{lemma}

\begin{proof}

Consider $v_t$ the element added to $S_{t-1}$.
\begin{align*}
    & f(S_t)-c(S_t) \\
    &=f(v_t|S_{t-1})+f(S_{t-1})-{c(S_t)}\\
    &\geq \left[\gamma(1-\epsilon) \displaystyle\frac{c(v_t)}{c(\OPT)} \Big ( f(\OPT) - f(S_{t-1}) \Big ) - 2\delta  \left( 1+n\frac{(1-\epsilon)c(v_t)}{ c(\OPT)} \right)\right]  +  f(S_{t-1})-c(S_t) \quad\quad\quad \mbox{\it due to Lemma~\ref{lem:marginal-gain-approx}} \\
    & = \left[\gamma(1-\epsilon) \displaystyle\frac{c(v_t)}{c(\OPT)} \Big ( f(\OPT) - f(S_{t-1}) \Big ) - \phi \right] +f(S_{t-1})- {\Big(c(S_{t-1})+c(v_t)\Big)} \quad\quad\quad \mbox{\it denote $\phi = 2\delta  \left( 1+n\frac{(1-\epsilon)c(v_t)}{ c(\OPT)} \right)$}\\    
    &= \left[\gamma(1-\epsilon) \displaystyle\frac{B_{t-1}+c(v_t)-B_{t-1}}{c(\OPT)} \Big ( f(\OPT) - f(S_{t-1}) \Big ) - \phi\right]  +f(S_{t-1})- {\left(\displaystyle\frac{1}{\gamma(1-\epsilon)}(A-B_{t-1})+c(v_t)\right)}
\end{align*}
The last equality follows from the definition of  $B_{t-1} = A-\gamma (1-\epsilon) c(S_{t-1})$.

\smallskip
Proceeding further, 

\begin{align*}
    &f(S_t)-c(S_t)\\
    &\geq \left[\gamma(1-\epsilon) \displaystyle\frac{B_{t-1}+c(v_t)-B_{t-1}}{c(\OPT)} \Big ( f(\OPT) - f(S_{t-1}) \Big ) - \phi\right]  +f(S_{t-1})- {\left(\displaystyle\frac{1}{\gamma(1-\epsilon)}(A-B_{t-1})+c(v_t)\right)}\\
    &\geq  {\gamma(1-\epsilon) \displaystyle\frac{B_{t-1}}{c(\OPT)} \Big( f(\OPT) - f(S_{t-1})\Big)}+ f(S_{t-1})    + {\gamma(1-\epsilon) \displaystyle\frac{c(v_t)-B_{t-1}}{c(\OPT)} \left( f(\OPT) - f(S_{t-1}) \right) }\\
        &\quad\quad - \phi - \displaystyle\frac{1}{\gamma(1-\epsilon)}A-\left(c(v_t)-\frac{1}{\gamma(1-\epsilon)}B_{t-1}\right)\\
    &\geq  \gamma(1-\epsilon)  \underbrace{\left[\displaystyle\frac{B_{t-1}}{c(\OPT)} \left( f(\OPT) - f(S_{t-1}) \right) + f(S_{t-1})\right]}_{} ~~+~~ \gamma(1-\epsilon) \displaystyle\frac{c(v_t)-B_{t-1}}{c(\OPT)} \left( f(\OPT) - f(S_{t-1}) \right) \\
        &\quad\quad- \phi- {\displaystyle\frac{1}{\gamma(1-\epsilon)}A-\Big[c(v_t)-B_{t-1}\Big]}\\  
    &\geq \gamma(1-\epsilon) \underbrace{\left(f(\OPT)-c(\OPT) - 2\delta\left(\frac{\beta}{\gamma(1-\epsilon)}+ \frac{n}{\gamma}\right) \right)}_{\mbox{\it apply Lemma~\ref{lem:h-upper-bound-approx}}} ~~-~~ \phi - \displaystyle\frac{1}{\gamma(1-\epsilon)}A \\
        &\quad\quad + \Big[c(v_t)-B_{t-1}\Big]\left( \gamma(1-\epsilon) \displaystyle\frac{f(\OPT) - f(S_{t-1}) }{c(\OPT)} -1 \right) \\
    &\geq \gamma(1-\epsilon) f(\OPT)-c(\OPT) - \gamma(1-\epsilon)2\delta\left(\frac{\beta}{\gamma(1-\epsilon)}+ \frac{n}{\gamma}\right)  - \phi - \displaystyle\frac{1}{\gamma(1-\epsilon)}A \\
        &\quad\quad + \Big[c(v_t)-B_{t-1}\Big]\left( \gamma(1-\epsilon) \displaystyle\frac{f(\OPT) - f(S_{t-1}) }{c(\OPT)} -1 \right) \\
    &\geq \gamma(1-\epsilon) f(\OPT)-c(\OPT) - 2\delta\left(\beta+ \frac{n}{\gamma}\right)  - \phi - \displaystyle\frac{1}{\gamma(1-\epsilon)}A \\
        &\quad\quad + \left[c(v_t)-B_{t-1}\right]\left( \gamma(1-\epsilon) \displaystyle\frac{f(\OPT) - f(S_{t-1}) }{c(\OPT)} -1 \right) 
\end{align*}

Note that, in Claim \ref{claim:cost-St-approx}, we have established $c(v_t)\geq B_{t-1}$. Furthermore, from Corollary \ref{col:lower-opt-approx}, we know
$f(\OPT) - f(S_j) \geq \displaystyle\frac{1}{\gamma(1-\epsilon)}{c(\OPT)}$. Therefore, 
\[
[c(v_t)-B_{t-1}]\left( \gamma(1-\epsilon) \displaystyle\frac{f(\OPT) - f(S_{t-1})}{c(\OPT)} -1 \right)
\]
is non-negative. Hence, 
If we denote $\beta' = \frac{c_{max}}{c(\OPT)}$, then $\phi = 2\delta  \left( 1+n\frac{(1-\epsilon)c(v_i)}{ c(\OPT)} \right) \leq 2\delta  \left( 1+n(1-\epsilon)\beta' \right)$
\begin{align*}
   &f(\Tilde{S})-c(\Tilde{S}) \geq f(S_t)-c(S_t)\\
   &\geq \gamma(1-\epsilon)f(\OPT)-c(\OPT)-\frac{1}{\gamma (1-\epsilon)}A - 2\delta\left(\beta+ \frac{n}{\gamma}\right)  - \phi\\
   &\geq \gamma(1-\epsilon)f(\OPT)-c(\OPT)-\frac{1}{\gamma (1-\epsilon)}A - 2\delta\left(\beta+ \frac{n}{\gamma}\right) - 2\delta  \left( 1+n(1-\epsilon)\beta' \right)\\
   &= \gamma(1-\epsilon)f(\OPT)-c(\OPT)-\frac{1}{\gamma (1-\epsilon)}A - 2\delta\left(\beta+ \frac{n}{\gamma} + 1 + n(1-\epsilon)\beta'\right)
\end{align*}

\end{proof}

This concludes the proof for Theorem~\ref{th:unconstrained-approx} for Case 2.2. 

\medskip
\begin{lemma}
    \label{lem:time-approx}
    The time complexity of Algorithm \ref{alg:up} is 
    $\mathcal{O}(\frac{n}{\epsilon} \log \frac{n}{\epsilon \gamma})$.
\end{lemma}
\begin{proof}
    As per line 17 of Algorithm \ref{alg:up}, every element can be considered for adding to the solution set for at most $\frac{\log \frac{n}{\gamma \epsilon}}{\epsilon}$ times. Each such consideration involves one oracle call to the function $f$. Therefore, the total number of oracle calls is $\mathcal{O}(\frac{n}{\epsilon} \log \frac{n}{\epsilon \gamma})$. 
\end{proof}

%% file: main.bbl
\begin{thebibliography}{64}
\providecommand{\natexlab}[1]{#1}
\providecommand{\url}[1]{\texttt{#1}}
\expandafter\ifx\csname urlstyle\endcsname\relax
  \providecommand{\doi}[1]{doi: #1}\else
  \providecommand{\doi}{doi: \begingroup \urlstyle{rm}\Url}\fi

\bibitem[Amanatidis et~al.(2020)Amanatidis, Fusco, Lazos, Leonardi, and Reiffenh{\"{a}}user]{amanatidis2020fast}
Georgios Amanatidis, Federico Fusco, Philip Lazos, Stefano Leonardi, and Rebecca Reiffenh{\"{a}}user.
\newblock Fast adaptive non-monotone submodular maximization subject to a knapsack constraint.
\newblock In Hugo Larochelle, Marc'Aurelio Ranzato, Raia Hadsell, Maria{-}Florina Balcan, and Hsuan{-}Tien Lin (eds.), \emph{Advances in Neural Information Processing Systems 33: Annual Conference on Neural Information Processing Systems 2020, NeurIPS 2020, December 6-12, 2020, virtual}, 2020.
\newblock URL \url{https://proceedings.neurips.cc/paper/2020/hash/c49e446a46fa27a6e18ffb6119461c3f-Abstract.html}.

\bibitem[Badanidiyuru \& Vondr{\'{a}}k(2014)Badanidiyuru and Vondr{\'{a}}k]{badanidiyuru2014fast}
Ashwinkumar Badanidiyuru and Jan Vondr{\'{a}}k.
\newblock Fast algorithms for maximizing submodular functions.
\newblock In Chandra Chekuri (ed.), \emph{Proceedings of the Twenty-Fifth Annual {ACM-SIAM} Symposium on Discrete Algorithms, {SODA} 2014, Portland, Oregon, USA, January 5-7, 2014}, pp.\  1497--1514. {SIAM}, 2014.
\newblock \doi{10.1137/1.9781611973402.110}.
\newblock URL \url{https://doi.org/10.1137/1.9781611973402.110}.

\bibitem[Bateni et~al.(2018)Bateni, Esfandiari, and Mirrokni]{bateni2018optimal}
MohammadHossein Bateni, Hossein Esfandiari, and Vahab~S. Mirrokni.
\newblock Optimal distributed submodular optimization via sketching.
\newblock In Yike Guo and Faisal Farooq (eds.), \emph{Proceedings of the 24th {ACM} {SIGKDD} International Conference on Knowledge Discovery {\&} Data Mining, {KDD} 2018, London, UK, August 19-23, 2018}, pp.\  1138--1147. {ACM}, 2018.
\newblock \doi{10.1145/3219819.3220081}.
\newblock URL \url{https://doi.org/10.1145/3219819.3220081}.

\bibitem[Bateni et~al.(2019)Bateni, Chen, Esfandiari, Fu, Mirrokni, and Rostamizadeh]{bateni2019categorical}
MohammadHossein Bateni, Lin Chen, Hossein Esfandiari, Thomas Fu, Vahab~S. Mirrokni, and Afshin Rostamizadeh.
\newblock Categorical feature compression via submodular optimization.
\newblock In Kamalika Chaudhuri and Ruslan Salakhutdinov (eds.), \emph{Proceedings of the 36th International Conference on Machine Learning, {ICML} 2019, 9-15 June 2019, Long Beach, California, {USA}}, volume~97 of \emph{Proceedings of Machine Learning Research}, pp.\  515--523. {PMLR}, 2019.
\newblock URL \url{http://proceedings.mlr.press/v97/bateni19a.html}.

\bibitem[Bian et~al.(2017)Bian, Buhmann, Krause, and Tschiatschek]{bian2017guarantees}
Andrew~An Bian, Joachim~M. Buhmann, Andreas Krause, and Sebastian Tschiatschek.
\newblock Guarantees for greedy maximization of non-submodular functions with applications.
\newblock In Doina Precup and Yee~Whye Teh (eds.), \emph{Proceedings of the 34th International Conference on Machine Learning, {ICML} 2017, Sydney, NSW, Australia, 6-11 August 2017}, volume~70 of \emph{Proceedings of Machine Learning Research}, pp.\  498--507. {PMLR}, 2017.
\newblock URL \url{http://proceedings.mlr.press/v70/bian17a.html}.

\bibitem[Bian et~al.(2021)Bian, Qian, Neumann, and Yu]{bian2021fast}
Chao Bian, Chao Qian, Frank Neumann, and Yang Yu.
\newblock Fast pareto optimization for subset selection with dynamic cost constraints.
\newblock In Zhi{-}Hua Zhou (ed.), \emph{Proceedings of the Thirtieth International Joint Conference on Artificial Intelligence, {IJCAI} 2021, Virtual Event / Montreal, Canada, 19-27 August 2021}, pp.\  2191--2197. ijcai.org, 2021.
\newblock \doi{10.24963/IJCAI.2021/302}.
\newblock URL \url{https://doi.org/10.24963/ijcai.2021/302}.

\bibitem[Bodek \& Feldman(2022)Bodek and Feldman]{bodek2022maximizing}
Kobi Bodek and Moran Feldman.
\newblock Maximizing sums of non-monotone submodular and linear functions: Understanding the unconstrained case.
\newblock In Shiri Chechik, Gonzalo Navarro, Eva Rotenberg, and Grzegorz Herman (eds.), \emph{30th Annual European Symposium on Algorithms, {ESA} 2022, September 5-9, 2022, Berlin/Potsdam, Germany}, volume 244 of \emph{LIPIcs}, pp.\  23:1--23:17. Schloss Dagstuhl - Leibniz-Zentrum f{\"{u}}r Informatik, 2022.
\newblock \doi{10.4230/LIPICS.ESA.2022.23}.
\newblock URL \url{https://doi.org/10.4230/LIPIcs.ESA.2022.23}.

\bibitem[Borgs et~al.(2014)Borgs, Brautbar, Chayes, and Lucier]{borgs2014maximizing}
Christian Borgs, Michael Brautbar, Jennifer~T. Chayes, and Brendan Lucier.
\newblock Maximizing social influence in nearly optimal time.
\newblock In Chandra Chekuri (ed.), \emph{Proceedings of the Twenty-Fifth Annual {ACM-SIAM} Symposium on Discrete Algorithms, {SODA} 2014, Portland, Oregon, USA, January 5-7, 2014}, pp.\  946--957. {SIAM}, 2014.
\newblock \doi{10.1137/1.9781611973402.70}.
\newblock URL \url{https://doi.org/10.1137/1.9781611973402.70}.

\bibitem[Buchbinder \& Feldman(2018)Buchbinder and Feldman]{buchbinder2018deterministic}
Niv Buchbinder and Moran Feldman.
\newblock Deterministic algorithms for submodular maximization problems.
\newblock \emph{{ACM} Trans. Algorithms}, 14\penalty0 (3):\penalty0 32:1--32:20, 2018.
\newblock \doi{10.1145/3184990}.
\newblock URL \url{https://doi.org/10.1145/3184990}.

\bibitem[Buchbinder et~al.(2015)Buchbinder, Feldman, Naor, and Schwartz]{buchbinder2015tight}
Niv Buchbinder, Moran Feldman, Joseph Naor, and Roy Schwartz.
\newblock A tight linear time (1/2)-approximation for unconstrained submodular maximization.
\newblock \emph{{SIAM} J. Comput.}, 44\penalty0 (5):\penalty0 1384--1402, 2015.
\newblock \doi{10.1137/130929205}.
\newblock URL \url{https://doi.org/10.1137/130929205}.

\bibitem[Buchbinder et~al.(2019)Buchbinder, Feldman, and Garg]{buchbinder2019deterministic}
Niv Buchbinder, Moran Feldman, and Mohit Garg.
\newblock Deterministic ({\textonehalf} + {\(\epsilon\)})-approximation for submodular maximization over a matroid.
\newblock In Timothy~M. Chan (ed.), \emph{Proceedings of the Thirtieth Annual {ACM-SIAM} Symposium on Discrete Algorithms, {SODA} 2019, San Diego, California, USA, January 6-9, 2019}, pp.\  241--254. {SIAM}, 2019.
\newblock \doi{10.1137/1.9781611975482.16}.
\newblock URL \url{https://doi.org/10.1137/1.9781611975482.16}.

\bibitem[C{\u{a}}linescu et~al.(2011)C{\u{a}}linescu, Chekuri, P{\'{a}}l, and Vondr{\'{a}}k]{calinescu2011maximizing}
Gruia C{\u{a}}linescu, Chandra Chekuri, Martin P{\'{a}}l, and Jan Vondr{\'{a}}k.
\newblock Maximizing a monotone submodular function subject to a matroid constraint.
\newblock \emph{{SIAM} J. Comput.}, 40\penalty0 (6):\penalty0 1740--1766, 2011.
\newblock \doi{10.1137/080733991}.
\newblock URL \url{https://doi.org/10.1137/080733991}.

\bibitem[Chen et~al.(2009)Chen, Wang, and Yang]{chen2009efficient}
Wei Chen, Yajun Wang, and Siyu Yang.
\newblock Efficient influence maximization in social networks.
\newblock In John F.~Elder IV, Fran{\c{c}}oise Fogelman{-}Souli{\'{e}}, Peter~A. Flach, and Mohammed~Javeed Zaki (eds.), \emph{Proceedings of the 15th {ACM} {SIGKDD} International Conference on Knowledge Discovery and Data Mining, Paris, France, June 28 - July 1, 2009}, pp.\  199--208. {ACM}, 2009.
\newblock \doi{10.1145/1557019.1557047}.
\newblock URL \url{https://doi.org/10.1145/1557019.1557047}.

\bibitem[Chen \& Kuhnle(2023)Chen and Kuhnle]{chen2023approximation}
Yixin Chen and Alan Kuhnle.
\newblock Approximation algorithms for size-constrained non-monotone submodular maximization in deterministic linear time.
\newblock In Ambuj~K. Singh, Yizhou Sun, Leman Akoglu, Dimitrios Gunopulos, Xifeng Yan, Ravi Kumar, Fatma Ozcan, and Jieping Ye (eds.), \emph{Proceedings of the 29th {ACM} {SIGKDD} Conference on Knowledge Discovery and Data Mining, {KDD} 2023, Long Beach, CA, USA, August 6-10, 2023}, pp.\  250--261. {ACM}, 2023.
\newblock \doi{10.1145/3580305.3599259}.
\newblock URL \url{https://doi.org/10.1145/3580305.3599259}.

\bibitem[Cohen et~al.(2014)Cohen, Delling, Pajor, and Werneck]{cohen2014sketch}
Edith Cohen, Daniel Delling, Thomas Pajor, and Renato~F. Werneck.
\newblock Sketch-based influence maximization and computation: Scaling up with guarantees.
\newblock In Jianzhong Li, Xiaoyang~Sean Wang, Minos~N. Garofalakis, Ian Soboroff, Torsten Suel, and Min Wang (eds.), \emph{Proceedings of the 23rd {ACM} International Conference on Conference on Information and Knowledge Management, {CIKM} 2014, Shanghai, China, November 3-7, 2014}, pp.\  629--638. {ACM}, 2014.
\newblock \doi{10.1145/2661829.2662077}.
\newblock URL \url{https://doi.org/10.1145/2661829.2662077}.

\bibitem[Crawford et~al.(2019)Crawford, Kuhnle, and Thai]{crawford2019submodular}
Victoria~G. Crawford, Alan Kuhnle, and My~T. Thai.
\newblock Submodular cost submodular cover with an approximate oracle.
\newblock In Kamalika Chaudhuri and Ruslan Salakhutdinov (eds.), \emph{Proceedings of the 36th International Conference on Machine Learning, {ICML} 2019, 9-15 June 2019, Long Beach, California, {USA}}, volume~97 of \emph{Proceedings of Machine Learning Research}, pp.\  1426--1435. {PMLR}, 2019.
\newblock URL \url{http://proceedings.mlr.press/v97/crawford19a.html}.

\bibitem[Cui et~al.(2023)Cui, Han, Tang, and Huang]{cui2023constrained}
Shuang Cui, Kai Han, Jing Tang, and He~Huang.
\newblock Constrained subset selection from data streams for profit maximization.
\newblock In Ying Ding, Jie Tang, Juan~F. Sequeda, Lora Aroyo, Carlos Castillo, and Geert{-}Jan Houben (eds.), \emph{Proceedings of the {ACM} Web Conference 2023, {WWW} 2023, Austin, TX, USA, 30 April 2023 - 4 May 2023}, pp.\  1822--1831. {ACM}, 2023.
\newblock \doi{10.1145/3543507.3583490}.
\newblock URL \url{https://doi.org/10.1145/3543507.3583490}.

\bibitem[Das \& Kempe(2011)Das and Kempe]{das2011submodular}
Abhimanyu Das and David Kempe.
\newblock Submodular meets spectral: Greedy algorithms for subset selection, sparse approximation and dictionary selection.
\newblock In Lise Getoor and Tobias Scheffer (eds.), \emph{Proceedings of the 28th International Conference on Machine Learning, {ICML} 2011, Bellevue, Washington, USA, June 28 - July 2, 2011}, pp.\  1057--1064. Omnipress, 2011.
\newblock URL \url{https://icml.cc/2011/papers/542\_icmlpaper.pdf}.

\bibitem[Duetting et~al.(2022)Duetting, Fusco, Lattanzi, Norouzi{-}Fard, and Zadimoghaddam]{dutting2022deletion}
Paul Duetting, Federico Fusco, Silvio Lattanzi, Ashkan Norouzi{-}Fard, and Morteza Zadimoghaddam.
\newblock Deletion robust submodular maximization over matroids.
\newblock In Kamalika Chaudhuri, Stefanie Jegelka, Le~Song, Csaba Szepesv{\'{a}}ri, Gang Niu, and Sivan Sabato (eds.), \emph{International Conference on Machine Learning, {ICML} 2022, 17-23 July 2022, Baltimore, Maryland, {USA}}, volume 162 of \emph{Proceedings of Machine Learning Research}, pp.\  5671--5693. {PMLR}, 2022.
\newblock URL \url{https://proceedings.mlr.press/v162/duetting22a.html}.

\bibitem[Ene \& Nguyen(2019)Ene and Nguyen]{ene2017nearly}
Alina Ene and Huy~L. Nguyen.
\newblock A nearly-linear time algorithm for submodular maximization with a knapsack constraint.
\newblock In Christel Baier, Ioannis Chatzigiannakis, Paola Flocchini, and Stefano Leonardi (eds.), \emph{46th International Colloquium on Automata, Languages, and Programming, {ICALP} 2019, July 9-12, 2019, Patras, Greece}, volume 132 of \emph{LIPIcs}, pp.\  53:1--53:12. Schloss Dagstuhl - Leibniz-Zentrum f{\"{u}}r Informatik, 2019.
\newblock \doi{10.4230/LIPICS.ICALP.2019.53}.
\newblock URL \url{https://doi.org/10.4230/LIPIcs.ICALP.2019.53}.

\bibitem[Fahrbach et~al.(2019)Fahrbach, Mirrokni, and Zadimoghaddam]{fahrbach2019submodular}
Matthew Fahrbach, Vahab~S. Mirrokni, and Morteza Zadimoghaddam.
\newblock Submodular maximization with nearly optimal approximation, adaptivity and query complexity.
\newblock In Timothy~M. Chan (ed.), \emph{Proceedings of the Thirtieth Annual {ACM-SIAM} Symposium on Discrete Algorithms, {SODA} 2019, San Diego, California, USA, January 6-9, 2019}, pp.\  255--273. {SIAM}, 2019.
\newblock \doi{10.1137/1.9781611975482.17}.
\newblock URL \url{https://doi.org/10.1137/1.9781611975482.17}.

\bibitem[Feige(1998)]{feige1998threshold}
Uriel Feige.
\newblock A threshold of ln \emph{n} for approximating set cover.
\newblock \emph{J. {ACM}}, 45\penalty0 (4):\penalty0 634--652, 1998.
\newblock \doi{10.1145/285055.285059}.
\newblock URL \url{https://doi.org/10.1145/285055.285059}.

\bibitem[Feldman et~al.(2023)Feldman, Nutov, and Shoham]{feldman2020practical}
Moran Feldman, Zeev Nutov, and Elad Shoham.
\newblock Practical budgeted submodular maximization.
\newblock \emph{Algorithmica}, 85\penalty0 (5):\penalty0 1332--1371, 2023.
\newblock \doi{10.1007/S00453-022-01071-2}.
\newblock URL \url{https://doi.org/10.1007/s00453-022-01071-2}.

\bibitem[Geng et~al.(2022)Geng, Gong, Liu, and Wu]{geng2022bicriteria}
Mengxue Geng, Shufang Gong, Bin Liu, and Weili Wu.
\newblock Bicriteria algorithms for maximizing the difference between submodular function and linear function under noise.
\newblock In Qiufen Ni and Weili Wu (eds.), \emph{Algorithmic Aspects in Information and Management - 16th International Conference, {AAIM} 2022, Guangzhou, China, August 13-14, 2022, Proceedings}, volume 13513 of \emph{Lecture Notes in Computer Science}, pp.\  133--143. Springer, 2022.
\newblock \doi{10.1007/978-3-031-16081-3\_12}.
\newblock URL \url{https://doi.org/10.1007/978-3-031-16081-3\_12}.

\bibitem[Gong et~al.(2023)Gong, Liu, Geng, and Fang]{gong2023algorithms}
Shufang Gong, Bin Liu, Mengxue Geng, and Qizhi Fang.
\newblock Algorithms for maximizing monotone submodular function minus modular function under noise.
\newblock \emph{J. Comb. Optim.}, 45\penalty0 (3):\penalty0 96, 2023.
\newblock \doi{10.1007/S10878-023-01026-5}.
\newblock URL \url{https://doi.org/10.1007/s10878-023-01026-5}.

\bibitem[Halabi \& Jegelka(2020)Halabi and Jegelka]{el2020optimal}
Marwa~El Halabi and Stefanie Jegelka.
\newblock Optimal approximation for unconstrained non-submodular minimization.
\newblock In \emph{Proceedings of the 37th International Conference on Machine Learning, {ICML} 2020, 13-18 July 2020, Virtual Event}, volume 119 of \emph{Proceedings of Machine Learning Research}, pp.\  3961--3972. {PMLR}, 2020.
\newblock URL \url{http://proceedings.mlr.press/v119/halabi20a.html}.

\bibitem[Halabi et~al.(2022)Halabi, Srinivas, and Lacoste{-}Julien]{el2022data}
Marwa~El Halabi, Suraj Srinivas, and Simon Lacoste{-}Julien.
\newblock Data-efficient structured pruning via submodular optimization.
\newblock In Sanmi Koyejo, S.~Mohamed, A.~Agarwal, Danielle Belgrave, K.~Cho, and A.~Oh (eds.), \emph{Advances in Neural Information Processing Systems 35: Annual Conference on Neural Information Processing Systems 2022, NeurIPS 2022, New Orleans, LA, USA, November 28 - December 9, 2022}, 2022.
\newblock URL \url{http://papers.nips.cc/paper\_files/paper/2022/hash/ed5854c456e136afa3faa5e41b1f3509-Abstract-Conference.html}.

\bibitem[Harrison \& Rubinfeld(1978)Harrison and Rubinfeld]{harrison1978hedonic}
David Harrison and Daniel~L Rubinfeld.
\newblock Hedonic housing prices and the demand for clean air.
\newblock \emph{Journal of Environmental Economics and Management}, 5\penalty0 (1):\penalty0 81--102, 1978.
\newblock ISSN 0095-0696.
\newblock \doi{https://doi.org/10.1016/0095-0696(78)90006-2}.
\newblock URL \url{https://www.sciencedirect.com/science/article/pii/0095069678900062}.

\bibitem[Harshaw et~al.(2019)Harshaw, Feldman, Ward, and Karbasi]{harshaw2019submodular}
Chris Harshaw, Moran Feldman, Justin Ward, and Amin Karbasi.
\newblock Submodular maximization beyond non-negativity: Guarantees, fast algorithms, and applications.
\newblock In Kamalika Chaudhuri and Ruslan Salakhutdinov (eds.), \emph{Proceedings of the 36th International Conference on Machine Learning, {ICML} 2019, 9-15 June 2019, Long Beach, California, {USA}}, volume~97 of \emph{Proceedings of Machine Learning Research}, pp.\  2634--2643. {PMLR}, 2019.
\newblock URL \url{http://proceedings.mlr.press/v97/harshaw19a.html}.

\bibitem[Hassidim \& Singer(2017)Hassidim and Singer]{hassidim2017submodular}
Avinatan Hassidim and Yaron Singer.
\newblock Submodular optimization under noise.
\newblock In Satyen Kale and Ohad Shamir (eds.), \emph{Proceedings of the 30th Conference on Learning Theory, {COLT} 2017, Amsterdam, The Netherlands, 7-10 July 2017}, volume~65 of \emph{Proceedings of Machine Learning Research}, pp.\  1069--1122. {PMLR}, 2017.
\newblock URL \url{http://proceedings.mlr.press/v65/hassidim17a.html}.

\bibitem[Huang et~al.(2017)Huang, Wang, Bevilacqua, Xiao, and Lakshmanan]{HuangWBXL17}
Keke Huang, Sibo Wang, Glenn~S. Bevilacqua, Xiaokui Xiao, and Laks V.~S. Lakshmanan.
\newblock Revisiting the stop-and-stare algorithms for influence maximization.
\newblock \emph{Proc. {VLDB} Endow.}, 10\penalty0 (9):\penalty0 913--924, 2017.
\newblock \doi{10.14778/3099622.3099623}.
\newblock URL \url{http://www.vldb.org/pvldb/vol10/p913-Huang.pdf}.

\bibitem[Jin et~al.(2021)Jin, Yang, Yang, Shi, Huang, and Xiao]{jin2021unconstrained}
Tianyuan Jin, Yu~Yang, Renchi Yang, Jieming Shi, Keke Huang, and Xiaokui Xiao.
\newblock Unconstrained submodular maximization with modular costs: Tight approximation and application to profit maximization.
\newblock \emph{Proc. {VLDB} Endow.}, 14\penalty0 (10):\penalty0 1756--1768, 2021.
\newblock \doi{10.14778/3467861.3467866}.
\newblock URL \url{http://www.vldb.org/pvldb/vol14/p1756-jin.pdf}.

\bibitem[Kempe et~al.(2003)Kempe, Kleinberg, and Tardos]{kempe2003maximizing}
David Kempe, Jon~M. Kleinberg, and {\'{E}}va Tardos.
\newblock Maximizing the spread of influence through a social network.
\newblock In Lise Getoor, Ted~E. Senator, Pedro~M. Domingos, and Christos Faloutsos (eds.), \emph{Proceedings of the Ninth {ACM} {SIGKDD} International Conference on Knowledge Discovery and Data Mining, Washington, DC, USA, August 24 - 27, 2003}, pp.\  137--146. {ACM}, 2003.
\newblock \doi{10.1145/956750.956769}.
\newblock URL \url{https://doi.org/10.1145/956750.956769}.

\bibitem[Kuhnle(2021)]{Kuhnle21}
Alan Kuhnle.
\newblock Quick streaming algorithms for maximization of monotone submodular functions in linear time.
\newblock In Arindam Banerjee and Kenji Fukumizu (eds.), \emph{The 24th International Conference on Artificial Intelligence and Statistics, {AISTATS} 2021, April 13-15, 2021, Virtual Event}, volume 130 of \emph{Proceedings of Machine Learning Research}, pp.\  1360--1368. {PMLR}, 2021.
\newblock URL \url{http://proceedings.mlr.press/v130/kuhnle21a.html}.

\bibitem[Lattanzi et~al.(2020)Lattanzi, Mitrovic, Norouzi{-}Fard, Tarnawski, and Zadimoghaddam]{lattanzi2020fully}
Silvio Lattanzi, Slobodan Mitrovic, Ashkan Norouzi{-}Fard, Jakub Tarnawski, and Morteza Zadimoghaddam.
\newblock Fully dynamic algorithm for constrained submodular optimization.
\newblock In Hugo Larochelle, Marc'Aurelio Ranzato, Raia Hadsell, Maria{-}Florina Balcan, and Hsuan{-}Tien Lin (eds.), \emph{Advances in Neural Information Processing Systems 33: Annual Conference on Neural Information Processing Systems 2020, NeurIPS 2020, December 6-12, 2020, virtual}, 2020.
\newblock URL \url{https://proceedings.neurips.cc/paper/2020/hash/9715d04413f296eaf3c30c47cec3daa6-Abstract.html}.

\bibitem[Leskovec et~al.(2007)Leskovec, Kleinberg, and Faloutsos]{leskovec2007graph}
Jure Leskovec, Jon~M. Kleinberg, and Christos Faloutsos.
\newblock Graph evolution: Densification and shrinking diameters.
\newblock \emph{{ACM} Trans. Knowl. Discov. Data}, 1\penalty0 (1):\penalty0 2, 2007.
\newblock \doi{10.1145/1217299.1217301}.
\newblock URL \url{https://doi.org/10.1145/1217299.1217301}.

\bibitem[Li et~al.(2023)Li, Mehr, and Horowitz]{li2023submodularity}
Ruolin Li, Negar Mehr, and Roberto Horowitz.
\newblock Submodularity of optimal sensor placement for traffic networks.
\newblock \emph{Transportation Research Part B: Methodological}, 171:\penalty0 29--43, 2023.

\bibitem[Li et~al.(2022)Li, Feldman, Kazemi, and Karbasi]{li2022submodular}
Wenxin Li, Moran Feldman, Ehsan Kazemi, and Amin Karbasi.
\newblock Submodular maximization in clean linear time.
\newblock In Sanmi Koyejo, S.~Mohamed, A.~Agarwal, Danielle Belgrave, K.~Cho, and A.~Oh (eds.), \emph{Advances in Neural Information Processing Systems 35: Annual Conference on Neural Information Processing Systems 2022, NeurIPS 2022, New Orleans, LA, USA, November 28 - December 9, 2022}, 2022.
\newblock URL \url{http://papers.nips.cc/paper\_files/paper/2022/hash/6faf3b8ed0df532c14d0fc009e451b6d-Abstract-Conference.html}.

\bibitem[Minoux(2005)]{minoux2005accelerated}
Michel Minoux.
\newblock Accelerated greedy algorithms for maximizing submodular set functions.
\newblock In \emph{Optimization Techniques: Proceedings of the 8th IFIP Conference on Optimization Techniques W{\"u}rzburg, September 5--9, 1977}, pp.\  234--243. Springer, 2005.

\bibitem[Mirzasoleiman et~al.(2015)Mirzasoleiman, Badanidiyuru, Karbasi, Vondr{\'{a}}k, and Krause]{mirzasoleiman2015lazier}
Baharan Mirzasoleiman, Ashwinkumar Badanidiyuru, Amin Karbasi, Jan Vondr{\'{a}}k, and Andreas Krause.
\newblock Lazier than lazy greedy.
\newblock In Blai Bonet and Sven Koenig (eds.), \emph{Proceedings of the Twenty-Ninth {AAAI} Conference on Artificial Intelligence, January 25-30, 2015, Austin, Texas, {USA}}, pp.\  1812--1818. {AAAI} Press, 2015.
\newblock \doi{10.1609/AAAI.V29I1.9486}.
\newblock URL \url{https://doi.org/10.1609/aaai.v29i1.9486}.

\bibitem[Nemhauser et~al.(1978)Nemhauser, Wolsey, and Fisher]{nemhauser1978analysis}
George~L. Nemhauser, Laurence~A. Wolsey, and Marshall~L. Fisher.
\newblock An analysis of approximations for maximizing submodular set functions - {I}.
\newblock \emph{Math. Program.}, 14\penalty0 (1):\penalty0 265--294, 1978.
\newblock \doi{10.1007/BF01588971}.
\newblock URL \url{https://doi.org/10.1007/BF01588971}.

\bibitem[Nguyen et~al.(2016)Nguyen, Thai, and Dinh]{NguyenTD16}
Hung~T. Nguyen, My~T. Thai, and Thang~N. Dinh.
\newblock Stop-and-stare: Optimal sampling algorithms for viral marketing in billion-scale networks.
\newblock In Fatma {\"{O}}zcan, Georgia Koutrika, and Sam Madden (eds.), \emph{Proceedings of the 2016 International Conference on Management of Data, {SIGMOD} Conference 2016, San Francisco, CA, USA, June 26 - July 01, 2016}, pp.\  695--710. {ACM}, 2016.
\newblock \doi{10.1145/2882903.2915207}.
\newblock URL \url{https://doi.org/10.1145/2882903.2915207}.

\bibitem[Nie et~al.(2023)Nie, Nadew, Zhu, Aggarwal, and Quinn]{nie2023framework}
Guanyu Nie, Yididiya~Y. Nadew, Yanhui Zhu, Vaneet Aggarwal, and Christopher~John Quinn.
\newblock A framework for adapting offline algorithms to solve combinatorial multi-armed bandit problems with bandit feedback.
\newblock In Andreas Krause, Emma Brunskill, Kyunghyun Cho, Barbara Engelhardt, Sivan Sabato, and Jonathan Scarlett (eds.), \emph{International Conference on Machine Learning, {ICML} 2023, 23-29 July 2023, Honolulu, Hawaii, {USA}}, volume 202 of \emph{Proceedings of Machine Learning Research}, pp.\  26166--26198. {PMLR}, 2023.
\newblock URL \url{https://proceedings.mlr.press/v202/nie23b.html}.

\bibitem[Nikolakaki et~al.(2021)Nikolakaki, Ene, and Terzi]{nikolakaki2021efficient}
Sofia~Maria Nikolakaki, Alina Ene, and Evimaria Terzi.
\newblock An efficient framework for balancing submodularity and cost.
\newblock In Feida Zhu, Beng~Chin Ooi, and Chunyan Miao (eds.), \emph{{KDD} '21: The 27th {ACM} {SIGKDD} Conference on Knowledge Discovery and Data Mining, Virtual Event, Singapore, August 14-18, 2021}, pp.\  1256--1266. {ACM}, 2021.
\newblock \doi{10.1145/3447548.3467367}.
\newblock URL \url{https://doi.org/10.1145/3447548.3467367}.

\bibitem[Padmanabhan et~al.(2023)Padmanabhan, Zhu, Basu, and Pavan]{PZBP23}
Madhavan~R. Padmanabhan, Yanhui Zhu, Samik Basu, and Aduri Pavan.
\newblock Maximizing submodular functions under submodular constraints.
\newblock In Robin~J. Evans and Ilya Shpitser (eds.), \emph{Uncertainty in Artificial Intelligence, {UAI} 2023, July 31 - 4 August 2023, Pittsburgh, PA, {USA}}, volume 216 of \emph{Proceedings of Machine Learning Research}, pp.\  1618--1627. {PMLR}, 2023.

\bibitem[Papadimitriou \& Yannakakis(1988)Papadimitriou and Yannakakis]{papadimitriou1988optimization}
Christos~H. Papadimitriou and Mihalis Yannakakis.
\newblock Optimization, approximation, and complexity classes (extended abstract).
\newblock In Janos Simon (ed.), \emph{Proceedings of the 20th Annual {ACM} Symposium on Theory of Computing, May 2-4, 1988, Chicago, Illinois, {USA}}, pp.\  229--234. {ACM}, 1988.
\newblock \doi{10.1145/62212.62233}.
\newblock URL \url{https://doi.org/10.1145/62212.62233}.

\bibitem[Perrault et~al.(2021)Perrault, Healey, Wen, and Valko]{perrault2022approximation}
Pierre Perrault, Jennifer Healey, Zheng Wen, and Michal Valko.
\newblock On the approximation relationship between optimizing ratio of submodular {(RS)} and difference of submodular {(DS)} functions.
\newblock \emph{CoRR}, abs/2101.01631, 2021.
\newblock URL \url{https://arxiv.org/abs/2101.01631}.

\bibitem[Qi(2024)]{Qi22}
Benjamin Qi.
\newblock On maximizing sums of non-monotone submodular and linear functions.
\newblock \emph{Algorithmica}, 86\penalty0 (4):\penalty0 1080--1134, 2024.
\newblock \doi{10.1007/S00453-023-01183-3}.
\newblock URL \url{https://doi.org/10.1007/s00453-023-01183-3}.

\bibitem[Qian(2021)]{qian2021multiobjective}
Chao Qian.
\newblock Multiobjective evolutionary algorithms are still good: Maximizing monotone approximately submodular minus modular functions.
\newblock \emph{Evol. Comput.}, 29\penalty0 (4):\penalty0 463--490, 2021.
\newblock \doi{10.1162/EVCO\_A\_00288}.
\newblock URL \url{https://doi.org/10.1162/evco\_a\_00288}.

\bibitem[Qian et~al.(2017)Qian, Shi, Yu, Tang, and Zhou]{qian2017subset}
Chao Qian, Jing{-}Cheng Shi, Yang Yu, Ke~Tang, and Zhi{-}Hua Zhou.
\newblock Subset selection under noise.
\newblock In Isabelle Guyon, Ulrike von Luxburg, Samy Bengio, Hanna~M. Wallach, Rob Fergus, S.~V.~N. Vishwanathan, and Roman Garnett (eds.), \emph{Advances in Neural Information Processing Systems 30: Annual Conference on Neural Information Processing Systems 2017, December 4-9, 2017, Long Beach, CA, {USA}}, pp.\  3560--3570, 2017.
\newblock URL \url{https://proceedings.neurips.cc/paper/2017/hash/d7a84628c025d30f7b2c52c958767e76-Abstract.html}.

\bibitem[Roostapour et~al.(2019)Roostapour, Neumann, Neumann, and Friedrich]{roostapour2022pareto}
Vahid Roostapour, Aneta Neumann, Frank Neumann, and Tobias Friedrich.
\newblock Pareto optimization for subset selection with dynamic cost constraints.
\newblock In \emph{The Thirty-Third {AAAI} Conference on Artificial Intelligence, {AAAI} 2019, The Thirty-First Innovative Applications of Artificial Intelligence Conference, {IAAI} 2019, The Ninth {AAAI} Symposium on Educational Advances in Artificial Intelligence, {EAAI} 2019, Honolulu, Hawaii, USA, January 27 - February 1, 2019}, pp.\  2354--2361. {AAAI} Press, 2019.
\newblock \doi{10.1609/AAAI.V33I01.33012354}.
\newblock URL \url{https://doi.org/10.1609/aaai.v33i01.33012354}.

\bibitem[Sipos et~al.(2012)Sipos, Swaminathan, Shivaswamy, and Joachims]{sipos2012temporal}
Ruben Sipos, Adith Swaminathan, Pannaga Shivaswamy, and Thorsten Joachims.
\newblock Temporal corpus summarization using submodular word coverage.
\newblock In Xue{-}wen Chen, Guy Lebanon, Haixun Wang, and Mohammed~J. Zaki (eds.), \emph{21st {ACM} International Conference on Information and Knowledge Management, CIKM'12, Maui, HI, USA, October 29 - November 02, 2012}, pp.\  754--763. {ACM}, 2012.
\newblock \doi{10.1145/2396761.2396857}.
\newblock URL \url{https://doi.org/10.1145/2396761.2396857}.

\bibitem[Stelzl et~al.(2005)Stelzl, Worm, Lalowski, Haenig, Brembeck, Goehler, Stroedicke, Zenkner, Schoenherr, Koeppen, Timm, Mintzlaff, Abraham, Bock, Kietzmann, Goedde, Toksöz, Droege, Krobitsch, Korn, Birchmeier, Lehrach, and Wanker]{stelzl2005human}
Ulrich Stelzl, Uwe Worm, Maciej Lalowski, Christian Haenig, Felix~H. Brembeck, Heike Goehler, Martin Stroedicke, Martina Zenkner, Anke Schoenherr, Susanne Koeppen, Jan Timm, Sascha Mintzlaff, Claudia Abraham, Nicole Bock, Silvia Kietzmann, Astrid Goedde, Engin Toksöz, Anja Droege, Sylvia Krobitsch, Bernhard Korn, Walter Birchmeier, Hans Lehrach, and Erich~E. Wanker.
\newblock A human protein-protein interaction network: A resource for annotating the proteome.
\newblock \emph{Cell}, 122\penalty0 (6):\penalty0 957--968, 2005.
\newblock ISSN 0092-8674.
\newblock \doi{https://doi.org/10.1016/j.cell.2005.08.029}.
\newblock URL \url{https://www.sciencedirect.com/science/article/pii/S0092867405008664}.

\bibitem[Sviridenko(2004)]{sviridenko2004note}
Maxim Sviridenko.
\newblock A note on maximizing a submodular set function subject to a knapsack constraint.
\newblock \emph{Oper. Res. Lett.}, 32\penalty0 (1):\penalty0 41--43, 2004.
\newblock \doi{10.1016/S0167-6377(03)00062-2}.
\newblock URL \url{https://doi.org/10.1016/S0167-6377(03)00062-2}.

\bibitem[Sviridenko et~al.(2017)Sviridenko, Vondr{\'{a}}k, and Ward]{sviridenko2017optimal}
Maxim Sviridenko, Jan Vondr{\'{a}}k, and Justin Ward.
\newblock Optimal approximation for submodular and supermodular optimization with bounded curvature.
\newblock \emph{Math. Oper. Res.}, 42\penalty0 (4):\penalty0 1197--1218, 2017.
\newblock \doi{10.1287/MOOR.2016.0842}.
\newblock URL \url{https://doi.org/10.1287/moor.2016.0842}.

\bibitem[Tang et~al.(2018)Tang, Tang, and Yuan]{tang2018towards}
Jing Tang, Xueyan Tang, and Junsong Yuan.
\newblock Towards profit maximization for online social network providers.
\newblock In \emph{2018 {IEEE} Conference on Computer Communications, {INFOCOM} 2018, Honolulu, HI, USA, April 16-19, 2018}, pp.\  1178--1186. {IEEE}, 2018.
\newblock \doi{10.1109/INFOCOM.2018.8485975}.
\newblock URL \url{https://doi.org/10.1109/INFOCOM.2018.8485975}.

\bibitem[Tang \& Yuan(2016)Tang and Yuan]{tang2016optimizing}
Shaojie Tang and Jing Yuan.
\newblock Optimizing ad allocation in social advertising.
\newblock In Snehasis Mukhopadhyay, ChengXiang Zhai, Elisa Bertino, Fabio Crestani, Javed Mostafa, Jie Tang, Luo Si, Xiaofang Zhou, Yi~Chang, Yunyao Li, and Parikshit Sondhi (eds.), \emph{Proceedings of the 25th {ACM} International Conference on Information and Knowledge Management, {CIKM} 2016, Indianapolis, IN, USA, October 24-28, 2016}, pp.\  1383--1392. {ACM}, 2016.
\newblock \doi{10.1145/2983323.2983834}.
\newblock URL \url{https://doi.org/10.1145/2983323.2983834}.

\bibitem[Tang et~al.(2014)Tang, Xiao, and Shi]{TangXS14}
Youze Tang, Xiaokui Xiao, and Yanchen Shi.
\newblock Influence maximization: near-optimal time complexity meets practical efficiency.
\newblock In Curtis~E. Dyreson, Feifei Li, and M.~Tamer {\"{O}}zsu (eds.), \emph{International Conference on Management of Data, {SIGMOD} 2014, Snowbird, UT, USA, June 22-27, 2014}, pp.\  75--86. {ACM}, 2014.
\newblock \doi{10.1145/2588555.2593670}.
\newblock URL \url{https://doi.org/10.1145/2588555.2593670}.

\bibitem[Tang et~al.(2015)Tang, Shi, and Xiao]{TangSX15}
Youze Tang, Yanchen Shi, and Xiaokui Xiao.
\newblock Influence maximization in near-linear time: {A} martingale approach.
\newblock In Timos~K. Sellis, Susan~B. Davidson, and Zachary~G. Ives (eds.), \emph{Proceedings of the 2015 {ACM} {SIGMOD} International Conference on Management of Data, Melbourne, Victoria, Australia, May 31 - June 4, 2015}, pp.\  1539--1554. {ACM}, 2015.
\newblock \doi{10.1145/2723372.2723734}.
\newblock URL \url{https://doi.org/10.1145/2723372.2723734}.

\bibitem[Wang et~al.(2021)Wang, Xu, and Yang]{wang2021maximizing}
Yijing Wang, Yicheng Xu, and Xiaoguang Yang.
\newblock On maximizing the difference between an approximately submodular function and a linear function subject to a matroid constraint.
\newblock In Ding{-}Zhu Du, Donglei Du, Chenchen Wu, and Dachuan Xu (eds.), \emph{Combinatorial Optimization and Applications - 15th International Conference, {COCOA} 2021, Tianjin, China, December 17-19, 2021, Proceedings}, volume 13135 of \emph{Lecture Notes in Computer Science}, pp.\  75--85. Springer, 2021.
\newblock \doi{10.1007/978-3-030-92681-6\_7}.
\newblock URL \url{https://doi.org/10.1007/978-3-030-92681-6\_7}.

\bibitem[Yaroslavtsev et~al.(2020)Yaroslavtsev, Zhou, and Avdiukhin]{yaroslavtsev2020bring}
Grigory Yaroslavtsev, Samson Zhou, and Dmitrii Avdiukhin.
\newblock "bring your own greedy"+max: Near-optimal 1/2-approximations for submodular knapsack.
\newblock In Silvia Chiappa and Roberto Calandra (eds.), \emph{The 23rd International Conference on Artificial Intelligence and Statistics, {AISTATS} 2020, 26-28 August 2020, Online [Palermo, Sicily, Italy]}, volume 108 of \emph{Proceedings of Machine Learning Research}, pp.\  3263--3274. {PMLR}, 2020.
\newblock URL \url{http://proceedings.mlr.press/v108/yaroslavtsev20a.html}.

\bibitem[Yin et~al.(2017)Yin, Benson, Leskovec, and Gleich]{yin2017local}
Hao Yin, Austin~R. Benson, Jure Leskovec, and David~F. Gleich.
\newblock Local higher-order graph clustering.
\newblock In \emph{Proceedings of the 23rd {ACM} {SIGKDD} International Conference on Knowledge Discovery and Data Mining, Halifax, NS, Canada, August 13 - 17, 2017}, pp.\  555--564. {ACM}, 2017.
\newblock \doi{10.1145/3097983.3098069}.
\newblock URL \url{https://doi.org/10.1145/3097983.3098069}.

\bibitem[Zhang et~al.(2023)Zhang, Deng, Jian, Chen, Hu, and Yang]{zhang2023communication}
Qixin Zhang, Zengde Deng, Xiangru Jian, Zaiyi Chen, Haoyuan Hu, and Yu~Yang.
\newblock Communication-efficient decentralized online continuous dr-submodular maximization.
\newblock In Ingo Frommholz, Frank Hopfgartner, Mark Lee, Michael Oakes, Mounia Lalmas, Min Zhang, and Rodrygo L.~T. Santos (eds.), \emph{Proceedings of the 32nd {ACM} International Conference on Information and Knowledge Management, {CIKM} 2023, Birmingham, United Kingdom, October 21-25, 2023}, pp.\  3330--3339. {ACM}, 2023.
\newblock \doi{10.1145/3583780.3614817}.
\newblock URL \url{https://doi.org/10.1145/3583780.3614817}.

\bibitem[Zhu et~al.(2024)Zhu, Basu, and Pavan]{zhu2024improved}
Yanhui Zhu, Samik Basu, and A.~Pavan.
\newblock Improved evolutionary algorithms for submodular maximization with cost constraints.
\newblock In Kate Larson (ed.), \emph{Proceedings of the Thirty-Third International Joint Conference on Artificial Intelligence, {IJCAI-24}}, pp.\  7082--7090. International Joint Conferences on Artificial Intelligence Organization, 8 2024.
\newblock \doi{10.24963/ijcai.2024/783}.
\newblock URL \url{https://doi.org/10.24963/ijcai.2024/783}.
\newblock Main Track.

\end{thebibliography}
